%% file: arxiv.tex
\documentclass[sigconf,nonacm]{acmart}

\usepackage{amsmath,csquotes,xcolor,array,multicol,amsthm,thmtools}
\usepackage{adjustbox}
\usepackage{balance}

\usepackage{thm-restate}
\usepackage{numprint}
\npdecimalsign{.}

\usepackage{caption}
\usepackage{subcaption}
\usepackage{graphicx}
\graphicspath{{plots/}}
\newlength{\smallfigwidth}
\newlength{\smallfigheight}
\newlength{\smallfigsep}
\newlength{\legendheight}
\usepackage{wrapfig}

\usepackage{enumitem}
\setlist{nolistsep}
\setlist[itemize]{leftmargin=*}
\setlist[enumerate]{leftmargin=*}

\usepackage{marginnote}
\usepackage{bbm}
\usepackage{microtype}
\usepackage{verbatim} %
\usepackage{hyperref}

\usepackage{xspace}



\usepackage[ruled,vlined, algo2e,linesnumbered]{algorithm2e} %
\usepackage[noend]{algorithmic}
\usepackage{algorithm}
\newtheorem{definition}{Definition}
\newtheorem{lemma}[definition]{Lemma}

\newtheorem{proposition}[definition]{Proposition}
\newtheorem{corollary}[definition]{Corollary}
\newtheorem{theorem}[definition]{Theorem}

\newcommand{\abs}[1]{\left|#1\right|}

\renewcommand{\lg}{\log}
\renewcommand{\varepsilon}{\epsilon}

\newcommand{\poly}{\operatorname{poly}}

\newcommand{\diag}{\operatorname{diag}}

\newcommand{\Exp}[1]{\mathbf{E}\left[ #1 \right]}
\newcommand{\Var}[1]{\mathbf{Var}\left[ #1 \right]}
\newcommand{\Prob}[1]{\mathbf{Pr}\left( #1 \right)}
\newcommand{\Diag}[1]{\mathrm{diag}\left( #1 \right)}

\DeclareRobustCommand{\ALG}{%
	\ifmmode
		\operatorname{ON}
	\else
		\text{ON}\xspace
	\fi
}
\DeclareRobustCommand{\OFF}{%
	\ifmmode
		\operatorname{OFF}
	\else
		\text{OFF}\xspace
	\fi
}
\DeclareRobustCommand{\APPROXALGO}{%
	\ifmmode
		\operatorname{APPROX}
	\else
		\text{APPROX}\xspace
	\fi
}

\providecommand{\abs}[1]{\ensuremath{\left\lvert#1\right\rvert}}
\providecommand{\norm}[1]{\ensuremath{\lVert#1\rVert}}

\newcommand{\sn}[1]{\marginnote{\tiny }}
\newcommand{\pp}[1]{\marginnote{\tiny }}
\newcommand{\yd}[1]{\marginnote{\tiny }}

\newcommand{\pr}{\operatorname{pr}}

\newcommand{\GooglePlus}{\textsf{GooglePlus}\xspace}
\newcommand{\TwitterFollows}{\textsf{TwitterFollows}\xspace}
\newcommand{\Flickr}{\textsf{Flickr}\xspace}
\newcommand{\YouTube}{\textsf{YouTube}\xspace}
\newcommand{\Pokec}{\textsf{Pokec}\xspace}
\newcommand{\Flixster}{\textsf{Flixster}\xspace}
\newcommand{\LiveJournal}{\textsf{LiveJournal}\xspace}

\newcommand{\SumOpinions}{\mathcal{S}}
\newcommand{\Polarization}{\mathcal{P}}
\newcommand{\avgExpressed}{\bar{z}}
\newcommand{\Disagreement}{\mathcal{D}}
\newcommand{\Internal}{\mathcal{I}}
\newcommand{\Controversy}{\mathcal{C}}
\newcommand{\DisCon}{\mathcal{DC}}

\allowdisplaybreaks

\AtBeginDocument{%
  \providecommand\BibTeX{{%
    \normalfont B\kern-0.5em{\scshape i\kern-0.25em b}\kern-0.8em\TeX}}}

\copyrightyear{2024} 
\acmYear{2024} 
\setcopyright{acmlicensed}\acmConference[WWW '24]{Proceedings of the ACM Web Conference 2024}{May 13--17, 2024}{Singapore, Singapore}
\acmBooktitle{Proceedings of the ACM Web Conference 2024 (WWW '24), May 13--17, 2024, Singapore, Singapore}
\acmDOI{10.1145/3589334.3645572}
\acmISBN{979-8-4007-0171-9/24/05}

\begin{document}

\title[Sublinear-Time Opinion Estimation in the Friedkin--Johnsen Model]
	{Sublinear-Time Opinion Estimation\\ in the Friedkin--Johnsen Model}

\author{Stefan Neumann}\authornote{This work was done while the author was at KTH Royal Institute of Technology.}
\affiliation{%
  \institution{TU Wien}
  \city{Vienna}
  \country{Austria}}
\email{stefan.neumann@tuwien.ac.at}

\author{Yinhao Dong}
\affiliation{%
  \institution{School of Computer Science and Technology, University of Science and Technology of China}
   \city{Hefei}
       \country{China}
  }
\email{yhdong@mail.ustc.edu.cn}

\author{Pan Peng}
\authornote{Corresponding Author.}
\affiliation{%
  \institution{School of Computer Science and Technology, University of Science and Technology of China}
  \city{Hefei}
   \country{China}
  }
\email{ppeng@ustc.edu.cn}

\renewcommand{\shortauthors}{Stefan Neumann, Yinhao Dong, Pan Peng}

\begin{abstract}
  Online social networks are ubiquitous parts of modern societies and the
  discussions that take place in these networks impact people's opinions on
  diverse topics, such as politics or vaccination. One of the most popular
  models to formally describe this opinion formation process is the
  Friedkin--Johnsen (FJ) model, which allows to define measures, such as the
  polarization and the disagreement of a network. Recently, Xu, Bao and
  Zhang~(WebConf'21) showed that all opinions and relevant measures in the FJ
  model can be approximated in near-linear time. However, their algorithm
  requires the \emph{entire} network and the opinions of \emph{all} nodes as
  input. Given the sheer size of online social networks and increasing
  data-access limitations, obtaining the entirety of this data might, however,
  be unrealistic in practice. In this paper, we show that node opinions and all
  relevant measures, like polarization and disagreement, can be efficiently
  approximated in time that is \emph{sublinear} in the size of the network.
  Particularly, our algorithms only require query-access to the network and do
  not have to preprocess the graph. Furthermore, we use a connection between FJ
  opinion dynamics and personalized PageRank, and show that in $d$-regular
  graphs, we can deterministically approximate each node's opinion by only
  looking at a constant-size neighborhood, independently of the network size. We
  also experimentally validate that our estimation algorithms perform well in
  practice.  
\end{abstract}



\maketitle

\section{Introduction}
Online social networks are used by billions of people on a daily basis and they
are central to today's societies. However, recently they have come under
scrutiny for allegedly creating echo chambers or filter bubbles, and for
increasing the polarization in societies. Research on the (non-)existence of
such phenomena is a highly active topic and typically studied
empirically~\cite{garimella2017reducing,garimella2018political,pariser2011filter}.

Besides the empirical work, in recent years it has become popular to study such
questions also
theoretically~\cite{bhalla2021local,chitra2020analyzing,zhu2021minimizing,matakos2017measuring,musco2018minimizing}.
These works typically rely on opinion formation models from sociology, which
provide an abstraction of how people form their opinions, based on their inner
beliefs and peer pressure from their neighbors. A popular model in this line of
work is the Friedkin--Johnsen~(FJ) model~\cite{friedkin1990social}, which
stipulates that every node has an \emph{innate opinion}, which is fixed and kept
private, and an \emph{expressed opinion}, which is updated over time and publicly
known in the network.

More formally, the \emph{Friedkin--Johnsen opinion dynamics}~\cite{friedkin1990social} are as follows.
Let $G=(V,E,w)$ be an undirected, weighted graph.
Each node~$u$ has an \emph{innate opinion}
$s_u\in[0,1]$ and an \emph{expressed opinion} $z_u\in[0,1]$. While the innate
opinions are fixed, %
the expressed opinions are updated over time~$t$
based on the update rule
\begin{align}
\label{eq:update-rule}
	z_u^{(t+1)}
	= \frac{s_u + \sum_{(u,v)\in E} w_{uv} z_u^{(t)}}{1 + \sum_{(u,v)\in E} w_{uv}},
\end{align}
i.e., the expressed opinions are weighted averages over a node's innate
opinion and the expressed opinions of the node's neighbors.
It is known that in the limit, for $t\to\infty$, the equilibrium expressed
opinions converge to $z^*=(I+L)^{-1}s$. Here, $I$ is the identity matrix and
$L=D-A$ denotes 
the Laplacian of~$G$, where $D$ is the weighted degree matrix and $A$ is the
weighted adjacency matrix.

\begin{table*}[t!]
\caption{Definition of all measures and the running times of our algorithms for
	an $n$-node weighted graph $G=(V,E,w)$ and error parameters
	$\varepsilon,\delta>0$. Let $\bar{\kappa}$ be an upper bound on
	$\kappa(\tilde{S})$, where $\tilde{S} = (I+D)^{-1/2} (I+L) (I+D)^{-1/2}$, and
	$r=O(\bar{\kappa} \log(\varepsilon^{-1} n\bar{\kappa} (\max_u w_u) ))$.
    Let $\bar{d}$ denote the average (unweighted) degree in $G$.
	Our algorithms succeed with probability $1-\delta$.}
\label{tab:measures}
\centering
\begin{adjustbox}{max width=\textwidth}
\begin{tabular}{ccccc}
\toprule
\multicolumn{2}{c}{\textbf{Measures}}  & \multicolumn{3}{c}{\textbf{Running Times}} \\
Name & Definition & Error & Given Oracle for $s_u$ & Given Oracle for $z_u^*$ \\
\cmidrule(lr){1-2}
\cmidrule(lr){3-5}
\emph{Innate opinion} & $s_u$ & $\pm \varepsilon$ & $O(1)$ & $O(\min\{d_u,w_u^2\varepsilon^{-2}\lg\delta^{-1}\})$ \\
\emph{Expressed opinion} & $z_u^*$ & $\pm \varepsilon$ & $O(\varepsilon^{-2} r^3 \lg r)$ & $O(1)$ \\
\emph{Average expressed opinion}
	& $\avgExpressed = \frac{1}{n} \sum_{u\in V} z_u^*$
	& $\pm \varepsilon$
	& $O(\varepsilon^{-2}\lg\delta^{-1})$
	& $O(\varepsilon^{-2}\lg\delta^{-1})$ \\
\emph{Sum of user opinions}
	& $\SumOpinions=\sum_{u\in V} z_u^*$
	& $\pm \varepsilon n$
	& $O(\varepsilon^{-2}\lg\delta^{-1})$
	& $O(\varepsilon^{-2}\lg\delta^{-1})$ \\
\emph{Polarization~\cite{musco2018minimizing}} &
	$\Polarization = \sum_{u\in V} (z_u^* - \bar{z})^2$
	& $\pm \varepsilon n$
	& $O(\varepsilon^{-4}  r^3 \log \delta^{-1} \lg r)$
	& $O(\varepsilon^{-2}\lg\delta^{-1})$ \\
\emph{Disagreement~\cite{musco2018minimizing}} &
	$\Disagreement=\sum_{(u,v)\in E} w_{uv} (z_u^*-z_v^*)^2$
	& $\pm \varepsilon n$
	& $O(\varepsilon^{-4}  r^3 \log \delta^{-1} \lg r)$
	& $O(\varepsilon^{-2} \bar{d} \lg^2 \delta^{-1})$ \\
\emph{Internal conflict~\cite{chen2018quantifying}} &
	$\Internal=\sum_{u\in V} (s_u - z_u^*)^2$
	& $\pm \varepsilon n$
	& $O(\varepsilon^{-4}  r^3 \log \delta^{-1} \lg r)$
	& $O(\varepsilon^{-2} \bar{d} \lg^2 \delta^{-1})$ \\
\emph{Controversy~\cite{chen2018quantifying, matakos2017measuring}} &
	$\Controversy=\sum_{u\in V} (z_u^*)^2$
	& $\pm \varepsilon n$
	& $O(\varepsilon^{-4}  r^3 \log \delta^{-1} \lg r)$
	& $O(\varepsilon^{-2}\lg\delta^{-1})$ \\
\emph{Disagreement-controversy~\cite{xu2021fast, musco2018minimizing}} &
	$\DisCon=\Disagreement+\Controversy$
	& $\pm \varepsilon n$
	& $O(\varepsilon^{-4}  r^3 \log \delta^{-1} \lg r)$
	& $O(\varepsilon^{-2} \bar{d} \lg^2 \delta^{-1})$ \\
\emph{Squared Norm $s$} &
	$\norm{s}_2^2$
	& $\pm \varepsilon n$
	& $O(\varepsilon^{-2}\lg\delta^{-1})$ 
	& $O(\varepsilon^{-2} \bar{d} \lg^2 \delta^{-1})$ \\
\bottomrule
\end{tabular}
\end{adjustbox}
\end{table*}

It is known that after convergence, the expressed opinions~$z^*$ do not reach a
consensus in general. This allows to study the distributions of opinions and to
define measures, such as \emph{polarization} and \emph{disagreement}.  More
concretely, the \emph{polarization}~$\Polarization$ is given by the variance of
the expressed opinions, i.e.,
$\Polarization = \sum_{u\in V} (z_u^* - \bar{z})^2$,
where $\bar{z}=\frac{1}{n} \sum_u z_u^*$ is the
average expressed opinion.  The \emph{disagreement}~$\Disagreement$
measures the stress among neighboring nodes in network, i.e.,
$\Disagreement=\sum_{(u,v)\in E} w_{uv} (z_u^*-z_v^*)^2$.  Similarly, we list
and formally define several other measures that are frequently used in the
literature in Table~\ref{tab:measures}; these are also the quantities for which
we provide efficient estimators in this paper.

In recent years a lot of attention has been devoted to studying how these
measures behave when the FJ model is undergoing interventions, such as changes
to the graph structure based on abstractions of timeline
algorithms~\cite{bhalla2021local,chitra2020analyzing} or adversarial
interference with node
opinions~\cite{gaitonde2020adversarial,chen2021adversarial,tu2023adversaries}.
The goal of these studies is to understand how much the disagreement and the
polarization increase or decrease after changes to the model parameters.

To conduct these studies, it is
necessary to simulate the FJ model highly efficiently. When done
na\"ively, this takes cubic time and is infeasible for large networks.
Therefore, Xu, Bao and Zhang~\cite{xu2021fast} provided a near-linear time
algorithm which approximates all relevant quantities up to a small additive
error and showed that their algorithm works very well in practice.  
Their algorithm requires the \emph{entire} network and the opinions of
\emph{all} nodes as input. However, given the sheer size of online social
networks and increasing data-access limitations, obtaining all of this data
might be unrealistic in practice and can be viewed as a drawback of this
algorithm.

\textbf{Our contributions.}
In this paper, we raise the question whether relevant quantities of the FJ
model, such as node opinions, polarization and disagreement can be approximated,
even if we do not know the entire network and even if we only know a small
number of node opinions. We answer this question affirmatively by providing
\emph{sublinear-time algorithms}, which only require query access to the graph
and the node opinions.

Specifically, we assume that we have \emph{query access to the graph}, which
allows us to perform the following operations in time~$O(1)$:
\begin{itemize} 
\item sample  a node from the graph uniformly at random,
\item given a node~$u$, return its weighted degree~$w_u$ and unweighted degree $d_u$,
\item given a node~$u$, randomly sample a neighbor~$v$ of~$u$ with probability $w_{uv}/w_u$ or $1/d_u$.%
\end{itemize}
Furthermore, we assume that we have access to an \emph{opinion oracle}. We
consider two types of these oracles. Given a node~$u$, the first type returns
its innate opinion~$s_u$ in time~$O(1)$, and the second type returns $u$'s equilibrium
expressed opinion~$z_u^*$ in time~$O(1)$. 

Under these assumptions, we summarize our results in Table~\ref{tab:measures}. %
In the table, $\varepsilon,\delta>0$ are parameters in the sense that the approximation ratio of our algorithm depends on
$\varepsilon$ and our algorithms succeed with probability $1-\delta$.
Note for both oracle types, we can approximate the average opinion in the
network in time $O(\varepsilon^{-2}\log \delta^{-1})$.

\emph{Given oracle access to the expressed equilibrium opinions~$z_u^*$}, we can
estimate the polarization within the same time.  We can also estimate the
disagreement in time $O(\varepsilon^{-2} \bar{d} \lg^2 \delta^{-1})$, where
$\bar{d} = \frac{2\abs{E}}{\abs{V}}$ %
is the average (unweighted) degree in $G$. Observe
that since most real-world networks are very sparse, $\bar{d}$ is small in
practice and thus these quantities can be computed highly efficiently. Note that
these results also imply upper bounds on how many node opinions one needs to
know, in order to approximate these quantities.

\emph{Given oracle access to the innate opinions $s_u$}, we can estimate the
polarization and disagreement in time
$\poly(\varepsilon^{-1},\delta^{-1},\log(n),\kappa(\tilde{S}))$, where
$\kappa(\tilde{S})$ is the condition number of
$\tilde{S} = (I+D)^{-1/2} (I+L) (I+D)^{-1/2}$. Note that if $\kappa(\tilde{S})$
is small, this is sublinear
in the graph size; indeed, below we show that on our real-world datasets we have
that $\kappa(\tilde{S}) \leq 300$.  In our experiments, we also show that this
algorithm is efficient in practice.

In conclusion, our results show that \emph{even when knowing only a sublinear
number of opinions in the network, we can approximate all measures from
Table~\ref{tab:measures}}.

Our two main technical contributions are as follows.
(1)~We show that in $d$-regular graphs, every node's expressed opinion~$z_u^*$
is determined (up to small error) by a small neighborhood
\emph{whose size is independent of the graph size}.
We obtain this result by exploiting a
connection~\cite{friedkin2014two,proskurnikov2016pagerank} between FJ opinion
dynamics and \emph{personalized PageRank}.
This connection allows us to give new algorithms for approximating the node opinions
efficiently when given oracle access to $s_u$.
(2)~We show that given oracle access to $z_u^*$, we can approximate $s_u$
up to error~$\pm\varepsilon$ or within a factor of $1\pm\varepsilon$ under mild conditions.
To obtain this result, we generalize a recent technique
by Beretta and Tetek~\cite{BerettaT22} for estimating weighted sums. That is, we
first sample a set of neighbors of $u$ such that each neighbor $v$ is sampled
independently with probability $w_{uv}/w_u$, and then we use the number of collisions for each neighbor to define an estimator that takes into account the expressed opinions. %

We experimentally evaluate our algorithms on real-world datasets. We show that
expressed and innate opinions can be approximated up to small additive error
$\pm0.01$. Additionally, we show that all measures, except disagreement, can be
efficiently estimated up to a relative error of at most~$4\%$.  Our oracle-based
algorithms which have oracle access to innate opinions~$s_u$ need less than
0.01~seconds to output a given node's opinion.  Even more interestingly, our
algorithms which have oracle access to the expressed opinions~$z_u^*$ achieve
error $\pm$\numprint{0.001} for estimating node opinions and can approximate
\numprint{10000}~node opinions in less than one second, even on our largest
graph with more than 4~million nodes and more than 40~million edges.
Additionally, we also add experiments comparing the running times of the
near-linear time algorithm by Xu et al.~\cite{xu2021fast} and running the FJ
opinion dynamics using the PageRank-style update rule
from~\cite{friedkin2014two,proskurnikov2016pagerank}.  We show that the
algorithms based on the connection to personalized PageRank are faster than the
baseline~\cite{xu2021fast}, while obtaining low approximation error.  We make
our source code publicly available on GitHub~\cite{code}.

\subsection{Related Work}
In recent years, online social networks and their timeline algorithms have been
blamed for increasing polarization and disagreement in societies. To obtain a
formal understanding of the phenomena, it has become an active research area to
combine opinion formation models with abstractions of algorithmic
interventions~\cite{bhalla2021local,chitra2020analyzing,zhu2021minimizing}.
One of the most popular models in this context is the the FJ
model~\cite{friedkin1990social} since it is highly amenable to analysis. 
Researchers studied interventions, such as edge insertions or changes to node
opinions, with the goal of decreasing the polarization and disagreement in the
networks~\cite{matakos2017measuring,musco2018minimizing,zhu2021minimizing,cinus2023rebalancing}.
Other works also studied adversarial
interventions~\cite{gaitonde2020adversarial,chen2021adversarial,tu2023adversaries} and viral content~\cite{tu2022viral}, as well as fundamental properties of the FJ model~\cite{bindel2015how,racz2022towards}.

The studies above have in common that their experiments rely on
simulations of the FJ model. To do this efficiently, 
Xu, Bao and Zhang~\cite{xu2021fast} used Laplacian solvers to obtain a
near-linear time algorithm for simulating the FJ model. However, this algorithm
requires access to the entire graph and all innate node opinions. Here,
we show that even when we only have query access to the graph and oracle
access to the opinions, we can obtain efficient simulations of the FJ
model in theory and in practice.

To obtain our results, we use several subroutines from previous works. Andoni,
Krauthgamer and Pogrow gave sublinear-time algorithm for solving linear
systems~\cite{andoni2019solving}. Andersen, Chung and
Lang~\cite{andersen2006local} proposed an algorithm that approximates a
personalized PageRank vector with a small residual vector, with running time
independent of the graph size. There also exist local algorithms for approximating the entries of the personalized PageRank vectors with small error~\cite{lofgren2014fast,lofgren2016personalized}.

Our algorithms for estimating the expressed opinions make heavy use of random
walks.  Random walks have also been exploited in sublinear-time algorithms for
approximating other local graph centrality measures~\cite{bressan2018sublinear},
stationary distributions~\cite{banerjee2015fast,bressan2019approximating},
estimating effective resistances~\cite{andoni2019solving,peng2021local} and for
sampling vertices with probability proportional to their
degrees~\cite{dasgupta2014estimating}.

\subsection{Notation}
Throughout the paper we let $G=(V,E,w)$ be a connected, weighted and undirected
graph with $n$~nodes and $m$~edges.  The weighted degree of a vertex~$u$ is given by $w_u = \sum_{(u,v)\in E}
w_{uv}$; the unweighted degree of~$u$ is given by
$d_u = \abs{\{v \colon (u,v)\in E\}}$.  For a graph $G$, $L=D-A$ denotes %
the
Laplacian of~$G$, where $D=\diag(w_1,\dots,w_n)$ is the weighted degree matrix
and $A$ is the (weighted) adjacency matrix such that $A_{uv}=w_{uv}$ if
$(u,v)\in E$ and $0$ otherwise. We write $I$ to denote the identity matrix. For a positive semidefinite
matrix~$S$, its condition number~$\kappa(S)$ is the ratio between the largest
and the smallest non-zero eigenvalues of $S$, which in turn are denoted by
$\lambda_{\max}(S)$ and $\lambda_{\min}(S)$, respectively. 

\section{Access to Oracle for Innate Opinions}
\label{sec:oracle-innate-opinions}
In this section, we assume that we have access to an oracle which, given a
node~$u$, returns its innate opinion~$s_u$ in time $O(1)$.

\subsection{Estimating Expressed Opinions $z_u^*$}
\label{sec:estimate-z}
We first show that for each node~$u$ we can estimate $z_u^*$ efficiently.

\textbf{Linear system solver.} First, recall that $z^*=(I+L)^{-1}s$.  Now we
observe that $z_u^*$ can be estimated using a sublinear-time solver for linear
systems~\cite{andoni2019solving}, which performs a given number of short random
walks. We present the pseudocode in Algorithm~\ref{alg:random-walks} and details
of the random walks in the appendix. %

\begin{algorithm}[t]
\begin{algorithmic}[1]
\REQUIRE {A graph $G=(V,E,w)$, oracle access to a vector $s\in [0,1]^n$ of innate opinions, an error parameter $\varepsilon>0$ and an upper bound $\bar{\kappa}$ on $\kappa(\tilde{S})$ with
	$\tilde{S} = (I+D)^{-1/2} (I+L) (I+D)^{-1/2}$}
\STATE $r \gets \log_{1/\bar{\kappa}}(2\varepsilon^{-1} (1-\bar{\kappa})^{-1} n^{1/2} (\max_u w_u)^{1/2})$
\STATE $\ell \gets O((\frac{\varepsilon}{2r})^{-2}\log r)$
\STATE Perform $\ell$~lazy random walks with timeout of length~$r$ from $u$
\FOR{$t = 1,\dots,r$}
	\STATE Let $u_i^{(t)}$~denote the vertex of the $i$-th walk after $t$~steps
	\STATE $\hat{x}_u^{(t)} \gets \frac{1}{\ell} \sum_{i=1}^\ell \frac{s_{u^{(t)}_i}}{w_{u^{(t)}_i}}$, {where $s_{u^{(t)}_i}$ is the innate opinion of vertex $u^{(t)}_i$}
\ENDFOR
\RETURN $\tilde{z}_u^* \gets \frac{1}{2} \sum_{t=1}^r \hat{x}^{(t)}_u$
\end{algorithmic}
\caption{Random walk-based algorithm for estimating $z_u^*$}
\label{alg:random-walks}
\end{algorithm}

\begin{proposition}
\label{pro:solver}
	Let $u\in V$ and $\varepsilon>0$. Let $\bar{\kappa}$ be an upper bound on
	$\kappa(\tilde{S})$ with
	$\tilde{S} = (I+D)^{-1/2} (I+L) (I+D)^{-1/2}$.
	Algorithm~\ref{alg:random-walks} returns a value $\tilde{z}_u^*$ such that
	$\abs{\tilde{z}_u^* - z_u^*} \leq \varepsilon$ with probability
	$1-\frac{1}{r}$ for $r=O(\bar{\kappa} \log(\varepsilon^{-1} n\bar{\kappa} 
				(\max_u w_u) ))$.
	Furthermore, $z_u^*$ %
 is computed in time $O(\varepsilon^{-2} r^3 \lg r)$ and
	using the same number of queries to $s$.
\end{proposition}
Note that the running time of the algorithm depends on an upper bound of the
condition number $\kappa(\tilde{S})$, which can be small in many real networks.
Indeed, for all of our real-world datasets, we have that $\kappa(\tilde{S}) \leq
300$ (see below).
In our experiments, we will practically evaluate
Algorithm~\ref{alg:random-walks} and show that it efficiently computes accurate
estimates of~$z_u^*$.

\textbf{Relationship to personalized PageRank.}
Next, we state a formal
connection~\cite{friedkin2014two,proskurnikov2016pagerank} between personalized
PageRank and FJ opinion dynamics.

First, in \emph{personalized PageRank}, we are given a teleportation probability parameter $\alpha\in
(0,1]$ and a vector $s\in[0,1]^n$ corresponding to a probability distribution
(i.e., $\sum_u s_u = 1$). Now, the \emph{personalized PageRank} is the column-vector
$\pr(\alpha,s)$ which is the unique solution to the equation 
$$\pr(\alpha, s) = \alpha s + (1-\alpha) \pr(\alpha, s)W,$$
where $W=\frac{I+D^{-1}A}{2}$ is the lazy random walk matrix. 
One can prove the following proposition.
\begin{proposition}[\cite{friedkin2014two,proskurnikov2016pagerank}]
\label{prop:page-rank-equation}
	The FJ expressed equilibrium opinions~$z^*$ are the unique solution to the
	equation $z^* = Ms + (I-M) D^{-1} A z^*$, where $M=(I+D)^{-1}$.
\end{proposition}

Comparing the two equations for $\pr(\alpha,s)$ and $z^*$, we observe that
there is a close relationship if we identify $M$ with $\alpha$ and $D^{-1}A$ with
$W$. Thus, while personalized PageRank uses lazy random walks (based on $W$),
the FJ opinion dynamics use vanilla random walks (based on $D^{-1} A$).
Additionally, while personalized PageRank
weights every entry~$u$ in $s$ with a factor of $\alpha$, the FJ opinion dynamics
essentially reweight each entry~$u$ with a factor of $\frac{1}{1+w_u}$. Note
that if all vertices $u$ satisfy that $w_u=d$, then $\alpha = \frac{1}{1+d}$ and the FJ opinion dynamics are
essentially a generalization of personalized PageRank.

We use this connection between the FJ opinion dynamics and personalized
PageRank to obtain a novel sublinear-time algorithm to estimate
entries~$z_u^*$.  More concretely, we consider weighted $d$-regular graphs and
show that Algorithm~\ref{alg:pagerank} can estimate entries $z_u^*$.
We prove that \emph{the FJ equilibrium opinions can be approximated
by repeatedly approximating personalized PageRank vectors}. In the algorithm, we
write $\vec{0}$ to denote the $0$s-vector, $\mathbbm{1}_i$ to denote the
length-$n$ indicator vector which in position~$i$ is~1 and all other entries
are~0 and $\norm{r}_1 = \sum_u \abs{r(u)}$ to denote the $1$-norm of a
vector. The algorithm invokes the well-known \textsc{Push} operation for
approximating the personalized PageRank in \cite{andersen2006local} as a
subroutine, and interactively updates the maintained vector $p$ until a very
small probability mass is left in the corresponding residual vector $r$.
We present the details of the algorithm from~\cite{andersen2006local} in the
appendix. %
For Algorithm~\ref{alg:pagerank} we obtain the following
guarantees.

\begin{algorithm}[t]
\begin{algorithmic}[1]
\REQUIRE {A graph $G=(V,E,w)$, oracle access to a vector $s\in [0,1]^n$ of innate opinions and an error parameter $\varepsilon > 0$}

\STATE $p\gets \vec{0}$ and $r\gets \mathbbm{1}_u$

\WHILE{$\norm{r}_1 > \varepsilon$}
	\FOR{all $i$ with $r(i) \neq 0$}
		\STATE\label{alg:personalizedpr} Run the local personalized PageRank algorithm from~\cite{andersen2006local}
			for $\mathbbm{1}_i$ to get $p^{(i)}$ and $r^{(i)}$
        \ENDFOR
	
	\STATE $p \gets p + \sum_i r(i) p^{(i)}$
 
	\STATE $r \gets \sum_i r(i) r^{(i)}$
 \ENDWHILE

\RETURN \label{alg:vectorp} $z_u' \gets p^{\intercal} s$
\end{algorithmic}
\caption{Personalized PageRank-based algorithm for estimating $z_u^*$ {in $d$-regular graphs}}
\label{alg:pagerank}
\end{algorithm}

\begin{theorem}
\label{thm:estimate-z}
	Let $d\in\mathbb{N}$ be an integer. Suppose $G$ is a $d$-regular graph and
	$\varepsilon\in(0,1)$. %
	Algorithm~\ref{alg:pagerank} returns an estimate $z_u'$ of $z_u^*$ such
	that $\abs{z_u' - z_u^*} \leq \varepsilon$ in time
	$(d/\varepsilon)^{O(d \lg(1/\varepsilon))}$.
\end{theorem}

Observe that the running time \emph{is independent of the graph size $n$}
for any constant $\varepsilon>0$ and $d=O(1)$.
This is in sharp contrast to Algorithm~\ref{alg:random-walks} from
Proposition~\ref{pro:solver}, whose running
time is $\Omega(\log n)$ even for $d$-regular graphs (for which $\bar{\kappa}=O(1)$).
Additionally, observe that Algorithm~\ref{alg:pagerank} is completely
deterministic, even though it is a sublinear-time algorithm.
Together, the above algorithm implies that in $d$-bounded graphs, every node's
opinion is determined (up to a small error) by the opinions of a constant-size
neighborhood. We believe that this an interesting insight into the FJ opinion
dynamics.

\subsubsection{Proof Sketch of Theorem \ref{thm:estimate-z}}
Let $W_s = D^{-1} A$. We define $\pr'(\alpha,s)$ as the unique solution of the equation
$\pr'(\alpha,s) = \alpha s + (1-\alpha) W_s \pr'(\alpha, s)$.
Note that this differs from the classic personalized PageRank only by the fact that
we use (non-lazy) random walks (based on $W_s$) rather than lazy random walks
(based on $W = \frac{1}{2}(I + D^{-1}A)$). 
Note that by letting $R' = \alpha \sum_{i=0}^\infty (1-\alpha)^i W_s^i$, we have that %
$\pr'(\alpha,s) = R' s$.
That is, 
\begin{align}
\label{eq:pr-identity}
	\pr'(\alpha,s) = \alpha \sum_{i=0}^\infty (1-\alpha)^i (W_s^i s).
\end{align}
Furthermore, we get that
\begin{align*}
	\pr'(\alpha,s)
	&= R' s = \alpha s + (1-\alpha) R' W_s s \\
	&= \alpha s + (1-\alpha) R' \pr'(\alpha, W_s s).
\end{align*}

Theorem \ref{thm:estimate-z} follows the lemmas below, whose proofs are deferred
to the appendix. %

\begin{lemma}
\label{lem:algo-identity}
It holds that $z_u^*= \pr'(\alpha,\mathbbm{1}_u)^{\intercal} s$, where $\alpha =\frac{1}{d+1}$, and $s\in [0,1]^n$ is the vector consisting of the innate opinions of all vertices. When Algorithm~\ref{alg:pagerank} finishes, it holds that
	$p + \pr'(\alpha,r) = \pr'(\alpha,\mathbbm{1}_u)$. Thus, it holds that 
	$$z_u^*
		= \pr'(\alpha,\mathbbm{1}_u)^{\intercal} s
		= p^{\intercal} s + \pr'(\alpha,r)^{\intercal} s.$$
\end{lemma}

The following gives guarantees on the approximation error. %
\begin{lemma}
\label{lem:approx-error}
Let $p$ be the vector in Step~(\ref{alg:vectorp}) in Algorithm \ref{alg:pagerank} %
and $s$ be the vector consisting of the innate opinions of all vertices. It holds that $\abs{z_u^* - p^{\intercal} s} \leq \varepsilon$.
\end{lemma}

The running time of the algorithm is given in the following lemma and corollary.
\begin{lemma}
\label{lem:vector-drop}
	After each iteration of the while-loop in Algorithm~\ref{alg:pagerank}, $\norm{r}_1$ decreases by a factor of at least
	$\frac{1}{d+1}$ and the number of non-zero entries in $r$ increases by a factor
	of $O(d/\varepsilon)$.
\end{lemma}

\begin{corollary}
\label{cor:running-time}
	We get that $\norm{r}_1 \leq \varepsilon$ in time
	$(d/\varepsilon)^{O(d \lg(1/\varepsilon))}$, which is independent of~$n$.
\end{corollary}

\subsection{Estimating Measures}
\label{sec:estimate-measures-given-s}
Now we give a short sketch of how to estimate the
measures from Table~\ref{tab:measures}, assuming oracle access to the innate
opinions $s$. We present all the details in the appendix. %

First, for computing the sum of expressed opinions we use the well-known fact
that $\SumOpinions = \sum_{u\in V} z_u^* = \sum_{u\in V} s_u$. Since we have
oracle access to $s$, we can thus focus on estimating $\sum_{u\in V} s_u$ which
can be done by randomly sampling $O(\varepsilon^{-2} \lg \delta^{-1})$
vertices~$U'$ and then returning the estimate~$\frac{n}{\abs{U'}} \sum_{u\in U'}
s_u$. Then a standard argument for estimating sums with bounded entries gives
that our approximation error is $\pm\varepsilon$. The quantities
$\avgExpressed$ and $\norm{s}_2^2$ can be estimated similarly.

For all other quantities, we require access to some expressed equilibrium
opinions~$z_u^*$. We obtain these opinions using
Proposition~\ref{pro:solver} and then our error bounds follow a similar argument
as above. However, in our analysis, we have to ensure that the error in our
estimates for $z_u^*$ does not compound and we have to take a union bound to
ensure that all estimates $z_u^*$ satisfy the error guarantees from the
proposition. Let us take the algorithm for	estimating $\Controversy=\sum_{u\in V}(z_u^*)^2$ as an example. %
We set 
         $\varepsilon_1 = \frac{\varepsilon}{6}$,  $r_1=O(\bar{\kappa} \log(\varepsilon_1^{-1} n\bar{\kappa} 
				(\max_u w_u) ))$, $\delta_1=\frac{1}{r_1}=\frac{\delta}{2C}$,
	$\varepsilon_2 = \frac{\varepsilon}{2}$, $\delta_2 = \frac{\delta}{2}$ and $C = \varepsilon_2^{-2} \log \delta_2^{-1}$.
Our algorithm first samples $C$~vertices (i.e., $i_1, \dots, i_{C}$) from $V$ uniformly at random, %
and obtains $\tilde{z}^*_{i_1}, \dots, \tilde{z}^*_{i_{C}}$ (using Proposition~\ref{pro:solver} with error parameter $\varepsilon_1$ and success probability $1-\delta_1$). Then it returns $\frac{n}{C} \sum_{j=1}^{C} (\tilde{z}^*_{i_j})^2$.  We can then show that the estimate approximates $\Controversy$ with additive error $\pm\varepsilon n$ with success probability $1-\delta$.

\section{Access to Oracle for Expressed Opinions}
\label{sec:oracle-expressed-opinions}

In this section, we assume that we have access to an oracle which, given a
vertex~$u$, returns its expressed opinion $z_u^*$ in time $O(1)$.

\subsection{Estimating Innate Opinions~$s_u$}
Next, our goal is to estimate entries $s_u$. To this end, note that $s =
(I+L)z^*$ and, hence, we have that
$s_u = (1+w_u)z_u^* - \sum_{(u,v)\in E} w_{uv} z_v^*$.
Observe that by our assumptions, we can compute the quantity $(1+w_u)z_u^*$
exactly using our oracle access. Therefore, our main challenge in this
subsection is to efficiently approximate
$S_u := \sum_{(u,v)\in E} w_{uv} z_v^*$.

In the following lemma, we give the guarantees for an algorithm which, if the
unweighted degree~$d_u$ of~$u$ is small, computes $S_u$ exactly in time~$O(d_u)$.
Otherwise, we sample a set $U'$ consisting of
$O(w_u^2\varepsilon^{-2})$ neighbors of $u$ which were sampled
with probabilities $\frac{w_{uv}}{w_u}$ and show that the sum
$\frac{w_u}{\abs{U'}} \sum_{v\in U'} z_v^*$ is an estimator for $S_u$ and  has error
$\pm\varepsilon$ with probability $1-\delta$.
\begin{lemma}
\label{lem:estimate-s}
	Let $\varepsilon,\delta\in(0,1)$. Then with probability at least $1-\delta$
	we can return an estimate of $s_u$ with additive error $\pm\varepsilon$ in
	time $O(\min\{d_u,w_u^2\varepsilon^{-2}\lg\delta^{-1}\})$.
\end{lemma}
Note that the running time of the estimate is highly efficient in practice,
as in real graphs most vertices have very small degrees.

Next, we show that if we make some mild assumptions on the expressed
opinions~$z_u^*$, then we can significantly reduce the time required to
estimate~$S_u$. In particular, for unweighted graphs we obtain almost a
quadratic improvement for the second term of the running time. 
Furthermore, we even obtain a \emph{multiplicative} error for estimating $S_u$.
\begin{proposition}
\label{prop:estimate-S}
	Let $\varepsilon,\delta\in(0,1)$ and set $S_u = \sum_{(u,v)\in E} z_v^* w_{uv}$.
	Suppose $z_u^* \in [c, 1)$ where $c\in (0,1)$ is a constant.
	Then with probability at least $1-\delta$ we can return an estimate of $S_u$
	with $(1\pm\varepsilon)$-multiplicative error in
	time $O(\min\{d_u,d_u^{1/2}\varepsilon^{-1}\lg\delta^{-1}\})$. 
\end{proposition}
This result is obtained by generalizing a result from Beretta and Tetek~\cite{BerettaT22}
and considering a slightly more complicated estimator than above. The new
estimator also samples a (multi-)set of neighbors~$U'$ of~$u$, but it
additionally takes into account \emph{collisions}. Let
$k=O(d_u^{1/2}\varepsilon^{-1}\log\delta^{-1})$ and let $v_1, \dots, v_k$ be
$k$~vertices picked independently at random and with replacement from all neighbors~$v$ of~$u$ with
probabilities proportional to their weights, i.e., $w_{uv}/w_u$.
Let $T$ be the \emph{set} of sampled vertices (i.e., while $U'$ may contain some
vertices multiple times, $T$ does not).  For each $t\in T$ define $c_t$ to be
the number of times vertex $t$ is sampled.
Our estimator in Proposition \ref{prop:estimate-S} is defined as follows:
\begin{equation*}
	\tilde{S}_u=w_u^2 \cdot {\binom{k}{2}}^{-1} \cdot \sum_{t \in T} \frac{\binom{c_t}{2}\cdot z^*_t}{w_{ut}}.
\end{equation*}

Using the proposition above, we obtain the following corollary for estimating
$s_u$ highly efficiently, even with multiplicative error.
\begin{corollary}
\label{cor:estimate-s-new}
	Let $\varepsilon,\delta\in(0,1)$ and set $S_u = \sum_{(u,v)\in E} z_v^* w_{uv}$. Suppose $z_u^* \in [c, 1)$ where $c\in (0,1)$ is a constant. 
        If $S_u\leq 1$, then with probability at least $1-\delta$
	we can return an estimate of $s_u$ with additive error $\pm\varepsilon$ in
	time $O(\min\{d_u,d_u^{1/2}\varepsilon^{-1}\lg\delta^{-1}\})$.
        If $S_u \leq \frac{(1+w_u)z_u^*}{2}$, then  with probability at least $1-\delta$
	we can return an estimate of $s_u$ with $(1\pm\varepsilon)$-multiplicative error in
	time $O(\min\{d_u,d_u^{1/2}\varepsilon^{-1}\lg\delta^{-1}\})$.
\end{corollary}

\subsection{Estimating Measures}
\label{sec:estimate-measures-given-z}

Now we give a short sketch of how to estimate the
measures from Table~\ref{tab:measures}, assuming oracle access to the expressed
opinions $z^*$. See the appendix %
for details. 

First, we note that we can calculate the sum of expressed
opinions~$\SumOpinions$, the polarization~$\Polarization$ and the
controversy~$\Controversy$ in the same way as approximating $\SumOpinions$ in
Section~\ref{sec:estimate-measures-given-s}.

For all other quantities, we require access to some innate
opinions~$s_u$, which we obtain via Lemma~\ref{lem:opinion-sampler}. This allows
us to estimate all other quantities.

\begin{lemma}
\label{lem:opinion-sampler}
	Let $\varepsilon,\delta \in (0,1)$ and $C\in\mathbb{N}$. Then there
	exists an algorithm which in time $O(C \bar{d} \lg \delta^{-1})$
	samples a (multi-)set of vertices~$S$ uniformly at random from~$V$ with
	$\abs{S}=C$ and it returns
	estimated innate opinions $\tilde{s}_u$ for all $u\in S$ such that with
	probability $1-\delta$ it holds that 
	$\abs{s_u - \tilde{s}_u} \leq \varepsilon$ for all $u\in S$.
\end{lemma}
\begin{proof}
	\emph{Step~1:} We introduce an \emph{opinion sampler} which samples a (multi-)set of
	vertices~$S$ uniformly at random from~$V$ with $\abs{S}=C$.  With
	probability at least $9/10$ it returns estimated innate opinions~$\tilde{s}_u$
	such that $\abs{s_u - \tilde{s}_u} \leq \varepsilon$ for all $u\in S$ and its
	running time is at most~$10T$, where $T=O(C \bar{d})$ is defined below.

	The opinion sampler samples $C$~vertices $i_1, \dots, i_{C}$ from $V$
	uniformly at random, and obtains $\tilde{s}_{i_1}, \dots, \tilde{s}_{i_{C}}$
	using Lemma~\ref{lem:estimate-s} with error parameter $\varepsilon$ and
	success probability $1-\frac{\delta}{2C}$. Observe that by a union bound it
	holds that $\abs{s_u - \tilde{s}_u} \leq \varepsilon$ for all $u\in S$
	with probability $1 - C \cdot \frac{\delta}{2C} = 1 - \frac{\delta}{2}$.

	Next, consider the running time of the opinion sampler. According to
	Lemma~\ref{lem:estimate-s}, for each $i_j \in S$, estimating $s_{i_j}$ takes
	time $O(d_{i_j})$. Note that for all $j \in [C]$,
	$\Prob{i_j=s}=\frac{1}{n}$ where $s\in V$, and thus the expected time to compute
	$\tilde{s}_{i_j}$ is $\Exp{d_{i_j}}=\frac{1}{n}\sum_{s\in V}d_s = \bar{d}$.
	Hence the expected running time of the opinion sampler is
	$T := O(C \bar{d})$. Now Markov's inequality implies that the probability
	that the opinion sampler has running time at most $10T$ is at least $9/10$.
	
	\emph{Step~2:} We repeatedly use the opinion sampler to prove the lemma. We do this
	as follows. We run the opinion sampler from above and if it finishes within
	time~$10T$, we return the estimated opinions it computed.  Otherwise, we
	restart this procedure and re-run the opinion sampler from scratch.  We
	perform the restarting procedure at most $\tau$~times.  Note that this
	procedure never runs for more than $O(\tau T)$~time.

	Observe that all $\tau$~runs of the opinion sampler require more than
	$10T$~time with probability at most $0.1^{\tau} \leq \frac{\delta}{2}$ for
	$\tau=O(\lg\delta^{-1})$. Furthermore, if the opinion sampler finishes
	then with probability at least $1-\frac{\delta}{2}$ all estimated innate
	opinions satisfy the guarantees from the lemma. Plugging in the parameters
	from above, we get that the algorithm deterministically runs in time
	$O(C \bar{d} \lg \delta^{-1})$ and by a union bound it satisfies the
	guarantees for the innate opinions with probability at least $1-\delta$.
\end{proof}

Let us take the algorithm for estimating $\norm{s}_2^2=\sum_{u\in V} s_u^2$ as
an example.  We set 
        $\varepsilon_1 = \frac{\varepsilon}{6}$, $\delta_1 = \frac{\delta}{2}$, 
	$\varepsilon_2 = \frac{\varepsilon}{2}$, $\delta_2 = \frac{\delta}{2}$ and $C = \varepsilon_2^{-2} \log \delta_2^{-1} = O(\varepsilon^{-2} \log \delta^{-1})$.
	According to Lemma~\ref{lem:opinion-sampler}, in time $O(C \bar{d} \lg \delta^{-1})$, we can sample a (multi-)set of vertices~$S=\{i_1, i_2, \dots, i_C\}$ uniformly at random from~$V$ and obtain
	estimated innate opinions $\tilde{s}_u$ for all $u\in S$ such that with
	probability $1-\delta_1$ it holds that 
	$\abs{s_u - \tilde{s}_u} \leq \varepsilon_1$ for all $u\in S$. We return $\frac{n}{C} \sum_{u\in S} \tilde{s}_{u}^2$.
Obviously, the running time is 
$O(C \bar{d} \lg \delta^{-1}) = O(\varepsilon^{-2} \bar{d} \lg^2 \delta^{-1})$.
The error guarantees are shown in the appendix. %

\section{Experiments}
\label{sec:experiments}
We experimentally evaluate our algorithms.  We run our experiments on a MacBook
Pro with a 2~GHz Quad-Core Intel Core~i5 and 16 GB RAM. We implement
Algorithm~\ref{alg:random-walks} in C++11 and perform the random walks in
parallel; all other algorithms are implemented in Python.  Our source code is
available online~\cite{code}.

The focus of our experiments is to assess the approximation quality and the
running times of our algorithms. As a baseline, we compare against the near-linear time
algorithm by Xu et al.~\cite{xu2021fast} which is
available on GitHub~\cite{xu2021fastcode}.
We run their algorithm with $\varepsilon=10^{-6}$ and 100~iterations.
We do not compare against an exact baseline, since the experiments
in~\cite{xu2021fast} show that their algorithm has a negligible error in
practice and since the exact computation is infeasible for our large datasets
(in the experiments of~\cite{xu2021fast}, their algorithm's relative error is
 less than $10^{-6}$ and matrix inversion does not scale to graphs with more
 than \numprint{56000}~nodes).

We use real-world datasets from KONECT~\cite{konect2013kunegis} and report their
basic statistics in Table~\ref{tab:experiments}. 
Since the datasets only consist of unweighted graphs and do not contain node
opinions, we generate the innate opinions synthetically using (1)~a uniform
distribution, (2)~a scaled version of the exponential distribution and (3)~opinions
based on the second eigenvector of~$L$.
In the main text, we present our results for opinions from
the uniform distribution, where we assigned the innate opinions~$s$ uniformly at
random in $[0,1]$. We present our results for other opinion distributions in
the appendix. %

\begin{table}[t!]
\caption{Statistics of our datasets. Here,
	$n$ and $m$ denote the number of nodes and edges in the largest connected
	components of the graph and $\kappa(\tilde{S})$ is the condition number of the
	matrix $\tilde{S} = (I+D)^{-1/2} (I+L) (I+D)^{-1/2}$.
 } 
\label{tab:experiments}
\centering
\begin{adjustbox}{max width=\textwidth}
\begin{tabular}{c cccc}
\toprule
\textbf{Dataset}
	& \multicolumn{4}{c}{\textbf{Statistics}} \\
\cmidrule(lr){2-5}
& $n$ & $m$ & $\bar{d}$ & $\kappa(\tilde{S})$ \\
\midrule
\GooglePlus & \numprint{201949} & \numprint{1133956} & \numprint{5.6} &
	\numprint{88.6} \\
\TwitterFollows & \numprint{404719} & \numprint{713319} & \numprint{1.8} &
	\numprint{12.5} \\
\YouTube & \numprint{1134890} & \numprint{2987624} & \numprint{2.6} &
	\numprint{20.9} \\
\Pokec & \numprint{1632803} & \numprint{22301964} & \numprint{13.7} &
	\numprint{56.7} \\
\Flixster & \numprint{2523386} & \numprint{7918801} & \numprint{3.1} &
	\numprint{19.0} \\
\Flickr & \numprint{2173370} & \numprint{22729227} & \numprint{10.5} &
	\numprint{103.7} \\
\LiveJournal & \numprint{4843953} & \numprint{42845684} & \numprint{8.8} &
	\numprint{297.4} \\
\bottomrule
\end{tabular}

\end{adjustbox}
\end{table}

\begin{table*}[t!]
\caption{Running times for estimating $z^*$ on different datasets with uniform
	opinions. 
	We report the running time for the Laplacian solver from~\cite{xu2021fast}.
	For the PageRank-style updates from
	Proposition~\ref{prop:page-rank-equation} we present running time and
	average error $\norm{\tilde{z}^* - z^*}_2/n$ after 50~iterations.
	For Algorithm~\ref{alg:random-walks} we present the running time for
	estimating \numprint{10000} opinions~$z_u^*$ using \numprint{600}~steps and
	\numprint{4000}~random walks; we also present the average query time for
	estimating a single opinion~$z_u^*$. The running times are averages of
	10~runs and we also report standard deviations.
 } 
\label{tab:running-times}
\centering
\begin{adjustbox}{max width=\textwidth}
\begin{tabular}{c c cc cc}
\toprule
\textbf{Dataset}
	& \textbf{Laplacian solver}~\cite{xu2021fast}
	& \multicolumn{2}{c}{\textbf{PageRank-Style Updates}}
	& \multicolumn{2}{c}{\textbf{Algorithm~\ref{alg:random-walks}}} \\
\cmidrule(lr){2-2}
\cmidrule(lr){3-4}
\cmidrule(lr){5-6}
& time (sec) & time (sec) & avg.\ error & time (sec) & time per vertex (sec) \\
\midrule
\GooglePlus & \numprint{6.2} ($\pm$ 0.2) & \numprint{0.6} ($\pm$ 0.0) & $3.2 \cdot 10^{-7}$ & 9.1 ($\pm$ 0.6) & $2.2 \cdot 10^{-3}$ \\
\TwitterFollows & \numprint{5.1} ($\pm$ 0.0) & \numprint{0.6} ($\pm$ 0.2) & $2.6 \cdot 10^{-10}$ & 2.4 ($\pm$ 0.1) & $6.1 \cdot 10^{-4}$ \\
\YouTube & \numprint{9.7} ($\pm$ 0.2) & \numprint{2.7} ($\pm$ 0.2) & $1.7 \cdot 10^{-9}$ & 4.5 ($\pm$ 0.1) & $1.1 \cdot 10^{-3}$ \\
\Pokec & \numprint{82.1} ($\pm$ 1.4) & \numprint{16.8} ($\pm$ 1.7) & $2.9 \cdot 10^{-8}$ & 36.2 ($\pm$ 0.2) & $9.0 \cdot 10^{-3}$ \\
\Flixster & \numprint{20.0} ($\pm$ 0.4) & \numprint{5.7} ($\pm$ 0.4) & $3.4 \cdot 10^{-10}$ & 6.9 ($\pm$ 0.3) & $1.7 \cdot 10^{-3}$ \\
\Flickr & \numprint{61.9} ($\pm$ 0.7) & \numprint{13.7} ($\pm$ 0.6) & $3.6 \cdot 10^{-8}$ & 34.0 ($\pm$ 1.9) & $8.5 \cdot 10^{-3}$ \\
\LiveJournal & \numprint{153.6} ($\pm$ 1.9) & \numprint{40.8} ($\pm$ 0.6) & $2.7 \cdot 10^{-7}$ & 22.5 ($\pm$ 0.4) & $5.6 \cdot 10^{-3}$ \\
\bottomrule
\end{tabular}

\end{adjustbox}
\end{table*}

\textbf{Experimental setup for oracle-based algorithms.}
We evaluate our oracle-based algorithms which will be the
main focus of this section.
In our experiments, we sample \numprint{10000}~vertices uniformly at random. For
each vertex, we estimate either the expressed opinion~$z_u^*$ or the
innate opinion~$s_u$.  Based on these estimates, we approximate the measures
from Table~\ref{tab:measures}; we do not report $\avgExpressed$ since (up to
rescaling) it is the same as $\SumOpinions$. Given an opinion estimate
$\tilde{s}_u$, we report the absolute error~$\abs{s_u - \tilde{s}_u}$. %
For the
measures, such as polarization~$\Polarization$, we report relative errors
$\abs{\Polarization - \tilde{\Polarization}}/\Polarization$, where
$\tilde{\Polarization}$ is an estimate of $\Polarization$. %
As our algorithms
are randomized, we perform each experiment 10~times and report means and standard
deviations.

\textbf{Results given oracles access to innate opinions~$s$.}
First, we report our results using an oracle for the innate opinions~$s$. We use
Algorithm~\ref{alg:random-walks} to obtain estimates of $z_u^*$ for
\numprint{10000} randomly chosen vertices~$u$. Then we estimate the measures from
Table~\ref{tab:measures} using the algorithms from
Section~\ref{sec:estimate-measures-given-s}. When not stated otherwise, we use
Algorithm~\ref{alg:random-walks} with \numprint{4000}~random walks of
length~\numprint{600}.

\begin{figure*}
	\centering
	\begin{subfigure}{0.24\textwidth}
		\includegraphics[width=\textwidth]{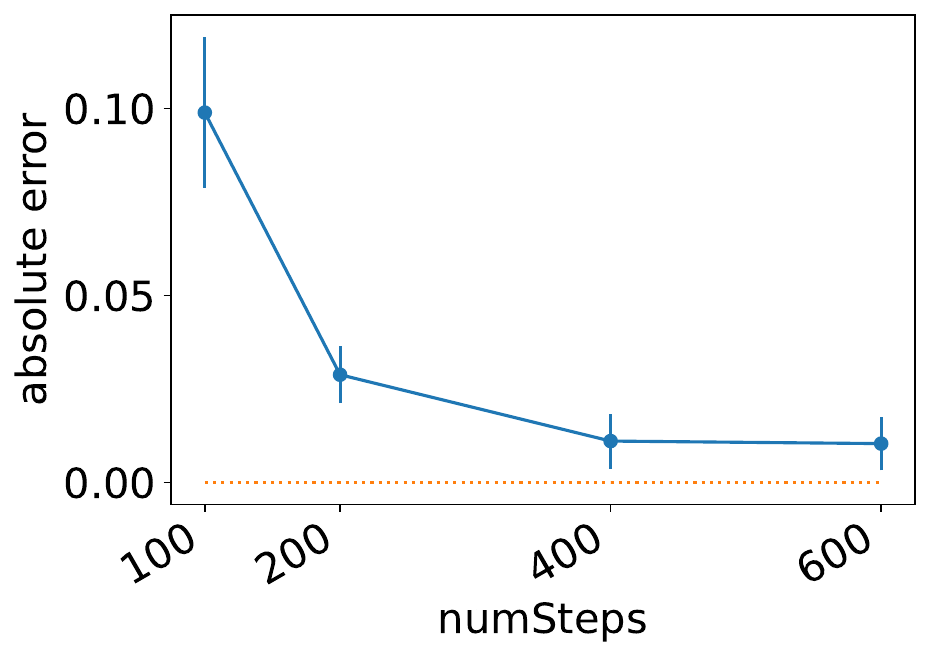}
		\caption{\Pokec, vary \#steps}
		\label{fig:pokec-steps}
	\end{subfigure}
	\hfill
	\begin{subfigure}{0.24\textwidth}
		\includegraphics[width=\textwidth]{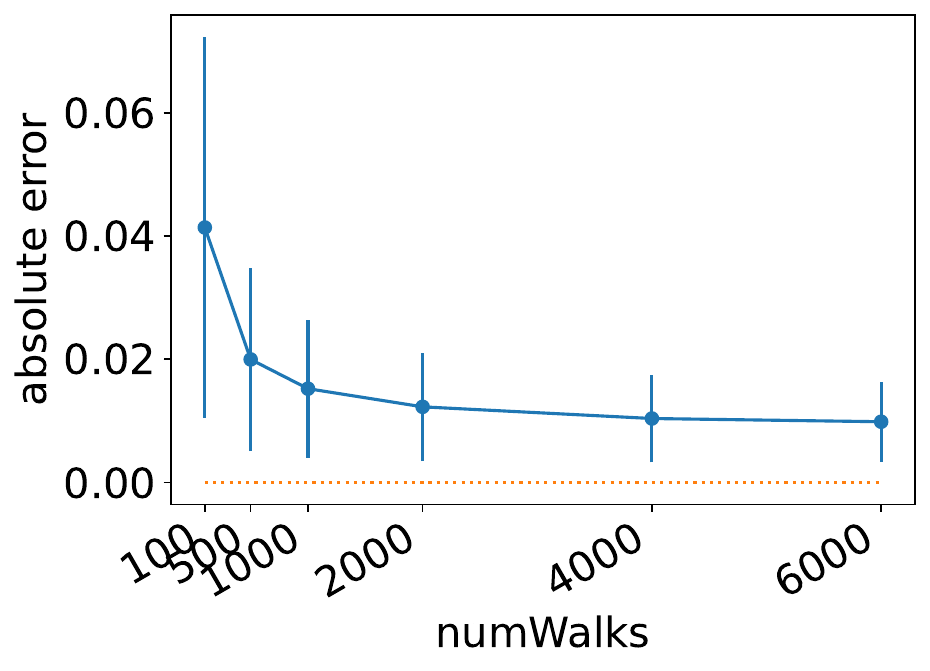}
		\caption{\Pokec, vary \#walks.}
		\label{fig:pokec-walks}
	\end{subfigure}
	\hfill
	\begin{subfigure}{0.24\textwidth}
		\includegraphics[width=\textwidth]{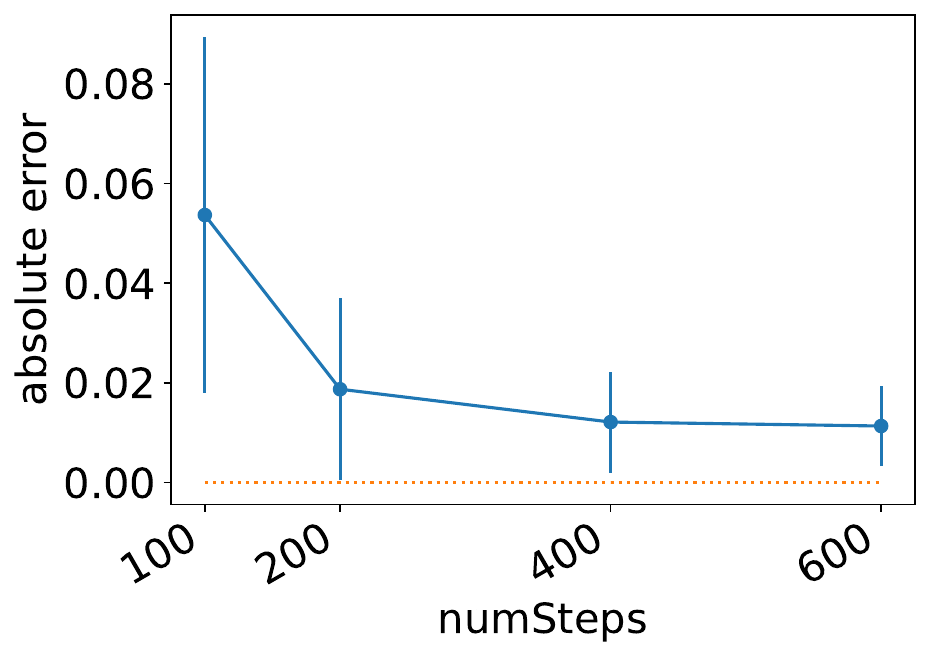}
		\caption{\LiveJournal, vary \#steps}
		\label{fig:livejournal-steps}
	\end{subfigure}
	\hfill
	\begin{subfigure}{0.24\textwidth}
		\includegraphics[width=\textwidth]{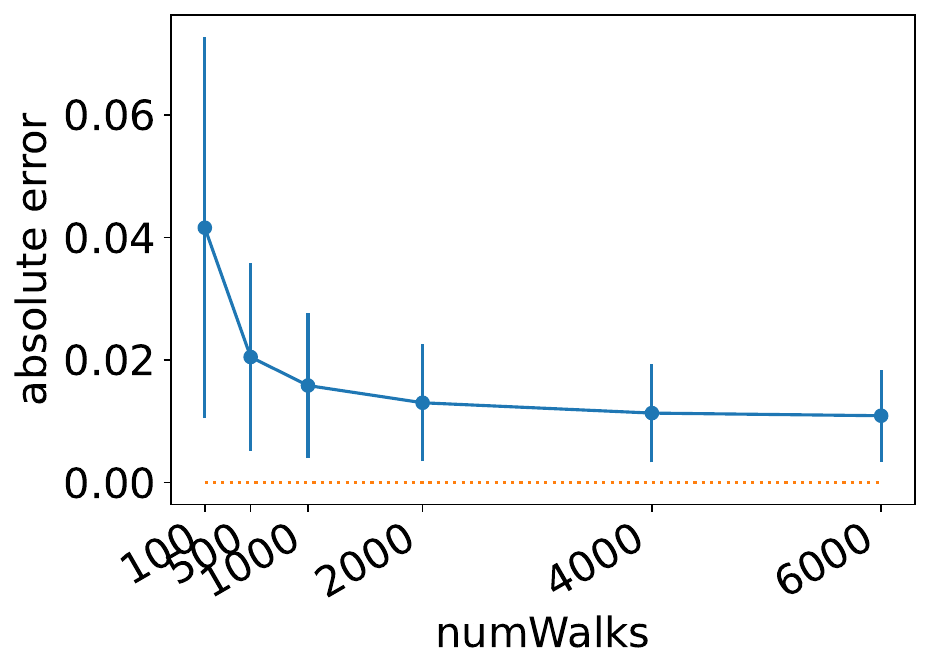}
		\caption{\LiveJournal, vary \#walks}
		\label{fig:livejournal-walks}
	\end{subfigure}
	\caption{Absolute error when estimating expressed opinions~$z_u^*$ using an
		oracle for innate opinions~$s_u$ via Algorithm~\ref{alg:random-walks}.
		We report means and standard deviations across 10~experiments.
		Figure~\ref{fig:estimating-opinions-oracle-innate}(\subref{fig:pokec-steps}) and~Figure~\ref{fig:estimating-opinions-oracle-innate}(\subref{fig:livejournal-steps}) use 4000~walks and vary the number of steps;
		Figure~\ref{fig:estimating-opinions-oracle-innate}(\subref{fig:pokec-walks}) and~Figure~\ref{fig:estimating-opinions-oracle-innate}(\subref{fig:livejournal-walks}) use 600~steps and vary the number of walks.
		Innate opinions were generated using the uniform distribution.
	}
	\label{fig:estimating-opinions-oracle-innate}
\end{figure*}

\emph{Dependency on algorithm parameters.}
In Figure~\ref{fig:estimating-opinions-oracle-innate} we present our results on
\Pokec and on \LiveJournal for estimating $z_u^*$ with varying number of steps
and random walks. We observe that increasing the parameters
decreases the absolute error. For both parameters and datasets, the error-curve
gets flatter once the error reaches $\pm0.01$, even though the standard
deviations keep on decreasing.

We looked into this phenomenon in more detail and found the following reasons:
(1)~Since in Algorithm~\ref{alg:random-walks} at each step the random walks
terminate with a certain probability, it is very unlikely that we observe ``very
long'' random walks (this is also corroborated by our running time analysis in
the appendix). %
Hence, at some
point we cannot increase the accuracy of our estimation by increasing the number
of random walk steps. Thus our only effective parameter to improve
the accuracy of Algorithm~\ref{alg:random-walks} is the number of random walks.
(2)~Next, observe that Proposition~\ref{pro:solver} suggests that to get
error~$\pm\varepsilon$ we need a running time depending on~$\varepsilon^{-2}$.
We looked into this phenomenon and found that (averaged across all datasets) the
error decays more slowly in practice. We blame this on hidden constants in the
running time analysis of Proposition~\ref{pro:solver} and believe that our
practical experiments are not yet in this asymptotic regime with error
dependency $\varepsilon^{-2}$.

\emph{Comparison across different datasets.}
In Table~\ref{tab:experiments-oracle-innate-uniform} we report the approximation
results. On all datasets, $z_u^*$ can be approximated with an average error of
$\pm0.01$. We observe that all measures except disagreement can be approximated
with an error of at most 4\%.

For the disagreement $\Disagreement=\sum_{(u,v)\in E} w_{uv} (z_u^*-z_v^*)^2$,
we obtain much higher errors for the following reason. 
We compute $\Disagreement  = \frac{1}{2} (\norm{s}_2^2 - \Internal - \Controversy)$ using
the conservation law from~\cite{chen2018quantifying} and using estimates for
$\norm{s}_2^2$, $\Internal$ and $\Controversy$. Therefore, the estimates' errors 
compound. Additionally, in practice we have that
$\Disagreement \ll \norm{s}_2^2$ and $\Disagreement\ll \Controversy$ since
typically the quantities $(z_u^*-z_v^*)^2$ that we are summing over in the
definition of~$\Disagreement$ are very close to zero. This ``amplifies'' the
effect of the approximation errors from estimating $\norm{s}_2^2$ and
$\Controversy$.  It is possible to obtain more accurate estimates of
$\Disagreement$ when we are able to sample edges from the graph uniformly at
random, and we report results in the appendix.

\begin{table*}[t!]
\caption{Errors for different datasets given an oracle for innate opinions; we
	report means and standard deviations (in parentheses) across 10~experiments. We ran
	Algorithm~\ref{alg:random-walks} with \numprint{600}~steps and
	\numprint{4000}~random walks; we estimated the opinions of
	\numprint{10000}~random vertices.  Innate opinions were generated using the
	uniform distribution.}
\label{tab:experiments-oracle-innate-uniform}
\centering
\begin{adjustbox}{max width=\textwidth}
\begin{tabular}{c cccccccc}
\toprule
\textbf{Dataset} & \textbf{Absolute Error} & \multicolumn{7}{c}{\textbf{Relative Error in \%}} \\
\cmidrule(lr){3-9} 
& $z_u^*$ & $\SumOpinions$ & $\Polarization$ & $\Disagreement$ &
$\Internal$ & $\Controversy$ & $\DisCon$ & $\norm{s}_2^2$ \\
\midrule
\input{tables/measures_Uniform_givenS_error_numStepsSamples600_numWalksRepetitions4000_numSampledVertices10000.tex}
\bottomrule
\end{tabular}
\end{adjustbox}
\end{table*}

\begin{table*}[t!]
\caption{Errors for different datasets given an oracle for expressed opinions;
	we report means and standard deviations (in parentheses) across 10~experiments. We ran our
	algorithm with threshold~400 and 5~repetitions; we estimated the opinions of
	\numprint{10000}~random vertices.  Innate opinions were generated using the
	uniform distribution.}
\label{tab:experiments-oracle-expressed-uniform}
\centering
\begin{adjustbox}{max width=\textwidth}
\begin{tabular}{c cccccccc}
\toprule
\textbf{Dataset} & \textbf{Absolute Error} & \multicolumn{7}{c}{\textbf{Relative Error in \%}} \\
\cmidrule(lr){3-9} 
& $s_u$ & $\SumOpinions$ & $\Polarization$ & $\Disagreement$ &
$\Internal$ & $\Controversy$ & $\DisCon$ & $\norm{s}_2^2$ \\
\midrule
\input{tables/measures_Uniform_givenZ_error_numStepsSamples400_numWalksRepetitions5_numSampledVertices10000.tex}
\bottomrule
\end{tabular}
\end{adjustbox}
\end{table*}

\emph{Running time analysis.}
In the last two columns of Table~\ref{tab:running-times}, we present
Algorithm~\ref{alg:random-walks}'s total running times and the running times per
vertex to obtain the results from
Table~\ref{tab:experiments-oracle-innate-uniform}.
We observe that on all datasets, the algorithms need less than $10^{-2}$~seconds
to estimate the opinion of a given vertex. 
In the appendix, %
we present additional experiments showing that
Algorithm~\ref{alg:random-walks} scales linearly in the number of random walks
and we show that after a certain threshold, increasing the number of random walk
steps does not increase the running time anymore (for the reason discussed
above).

\textbf{Results given oracles access to expressed opinions~$z^*$.}
Next, we report our results given oracle access to expressed opinions~$z^*$. We
implement the algorithm from Lemma~\ref{lem:estimate-s} for estimating innate
opinions~$s_u$, where we introduce a \emph{threshold} $t$. If $d_u<t$,
we compute $s_u = (1+w_u)z_u^* - \sum_{(u,v)\in E} w_{uv} z_v^*$ exactly in time $O(d_u)$;
otherwise, we use the random sampling strategy from the lemma and pick
$t$~neighbors of $u$ uniformly at random to estimate $s_u$. We
repeat this procedure 5~times and pick the median answer. We set $t=400$.

Table~\ref{tab:experiments-oracle-expressed-uniform} presents our results. We
observe that the absolute errors when estimating $s_u$ are extremely small
($\pm0.001$). This is because all of our datasets have very small average degrees
(see Table~\ref{tab:experiments}) and, thus, for most randomly picked vertices
$s_u$ we are computing the answer exactly since $d_u\leq t$.

However, we note that
when resorting to random sampling for high-degree nodes, we typically have a
relatively large error: to obtain error $\pm\varepsilon$ for
$s_u = (1+w_u)z_u^* - \sum_{(u,v)\in E} w_{uv} z_v^*$, we have to
estimate $\sum_{(u,v)\in E} w_{uv} z_v^*$ with absolute error~$\pm\varepsilon$,
which is impractical since this sum is typically very large. The same is the
case for the random sampling scheme from Proposition~\ref{prop:estimate-S}.
															   
For the measures from Table~\ref{tab:measures}, we observe that again all
relative errors are less than 3\%, except for disagreement where we have the
same issue as described above.

Next, we briefly discuss the running time of the algorithms.  First, computing
the entire vector $s=(I+L)z^*$ exactly is highly efficient since it only
involves a matrix--vector multiplication and can be done on all datasets in less
than one second. Our oracle-based algorithm, which starts from scratch for each
opinion that shall be estimated, can estimate the opinions of
\numprint{10000}~nodes in less than 1~second for each datasets, including
\LiveJournal which contains more than 4~million nodes and more than 40~million
edges. Hence, we spend less than $10^{-4}$~seconds for each opinion $s_u$.
Here, the algorithm benefits from the fact that the vertices in our datasets
have very low degrees and thus $s_u$ can be estimated highly efficiently.

\textbf{Evaluation of PageRank-style update rule.}
Finally, we also evaluate the usefulness of the PageRank-style update
rule~\cite{friedkin2014two,proskurnikov2016pagerank} from
Proposition~\ref{prop:page-rank-equation}. To this end, we implement an
algorithm which initializes a vector in which all entries are set to the average
of the node opinions~$\frac{1}{n} \sum_u s_u$ and then we apply the update rule
from the proposition for 50~iterations.

In Table~\ref{tab:running-times} we report the running times of the
baseline~\cite{xu2021fast} and the PageRank-style algorithm. We also report the
average error $\norm{\tilde{z}^*-z^*}_2/n$, where $\tilde{z}^*$ is the output of
the PageRank-style algorithm and $z^*$ is the output of~\cite{xu2021fast}. We
find that the PageRank-style algorithm is at least 3.7~times faster than the
baseline and its errors are very small. We also find that the error decays
exponentially in the number of iterations of applying the update rule (see the
appendix); this resembles prior convergence
results~\cite{proskurnikov2017opinion,ghaderi2014opinion}.

\section{Conclusion}
In this paper, we studied the popular Friedkin--Johnsen model for opinion
dynamics. We showed that all relevant quantities, such as single node opinions
and measures like polarization and disagreement, can be provably approximated in
sublinear time. We used a connection between the FJ model and personalized PageRank
to show that for $d$-regular graphs, each node's expressed opinion can be
approximated by only looking at a constant-size neighborhood.
Furthermore, to obtain our sublinear-time estimator for innate opinions we
presented new results for estimating weighted sums and we showed that we can
achieve small additive and multiplicative errors under mild conditions.  We also
evaluated our algorithms experimentally and showed that for all measures except
disagreement, they achieve a small error of less than 4\% in practice.
They are also very fast in practice.

\begin{acks}
This research has been funded by the Wallenberg AI, Autonomous Systems and
Software Program (WASP) funded by the Knut and Alice Wallenberg Foundation, and
the Vienna Science and Technology Fund (WWTF) [Grant ID: 10.47379/VRG23013].
Y.D. and P.P. are supported in part by NSFC grant 62272431 and ``the Fundamental Research Funds for the Central Universities''.
\end{acks}

\clearpage

\bibliographystyle{ACM-Reference-Format}
\balance
\bibliography{main}

\newpage

\appendix

\input{appendix}

\end{document}

%% file: tables/measures_Uniform_givenS_error_numStepsSamples600_numWalksRepetitions4000_numSampledVertices10000.tex
\GooglePlus & 0.011 ($\pm$0.008) & 0.6 ($\pm$0.4) & 2.5 ($\pm$0.9) & 36.2 ($\pm$9.9) & 1.2 ($\pm$0.4) & 3.4 ($\pm$0.4) & 1.5 ($\pm$0.7) & 0.9 ($\pm$0.4) \\ 
\TwitterFollows & 0.011 ($\pm$0.007) & 0.8 ($\pm$0.7) & 1.9 ($\pm$0.5) & 31.1 ($\pm$6.5) & 0.8 ($\pm$0.5) & 3.6 ($\pm$0.7) & 1.5 ($\pm$0.7) & 1.1 ($\pm$0.9) \\ 
\Flixster & 0.013 ($\pm$0.007) & 0.4 ($\pm$0.2) & 2.5 ($\pm$1.5) & 34.2 ($\pm$3.5) & 1.6 ($\pm$1.0) & 4.2 ($\pm$0.3) & 1.8 ($\pm$0.4) & 0.4 ($\pm$0.3) \\ 
\Pokec & 0.010 ($\pm$0.007) & 0.7 ($\pm$0.5) & 3.3 ($\pm$1.8) & 52.1 ($\pm$20.0) & 0.1 ($\pm$0.1) & 3.4 ($\pm$0.4) & 1.7 ($\pm$0.9) & 1.0 ($\pm$0.7) \\ 
\Flickr & 0.012 ($\pm$0.007) & 0.6 ($\pm$0.2) & 1.7 ($\pm$0.6) & 36.3 ($\pm$6.4) & 1.0 ($\pm$0.7) & 4.0 ($\pm$0.5) & 1.8 ($\pm$0.8) & 1.0 ($\pm$0.3) \\ 
\YouTube & 0.012 ($\pm$0.007) & 0.4 ($\pm$0.2) & 1.6 ($\pm$1.3) & 31.4 ($\pm$5.2) & 1.8 ($\pm$1.2) & 4.0 ($\pm$0.4) & 1.9 ($\pm$0.6) & 0.7 ($\pm$0.4) \\ 
\LiveJournal & 0.011 ($\pm$0.008) & 0.4 ($\pm$0.3) & 3.6 ($\pm$2.0) & 41.4 ($\pm$8.9) & 1.4 ($\pm$0.8) & 3.9 ($\pm$0.4) & 1.9 ($\pm$0.5) & 0.7 ($\pm$0.2) \\ 

%% file: tables/measures_Uniform_givenZ_error_numStepsSamples400_numWalksRepetitions5_numSampledVertices10000.tex
\GooglePlus & 0.000 ($\pm$0.006) & 0.2 ($\pm$0.3) & 1.6 ($\pm$1.5) & 8.3 ($\pm$8.5) & 1.3 ($\pm$0.9) & 0.4 ($\pm$0.5) & 0.6 ($\pm$0.9) & 0.8 ($\pm$0.9) \\ 
\TwitterFollows & 0.001 ($\pm$0.026) & 0.3 ($\pm$0.2) & 1.2 ($\pm$1.2) & 5.2 ($\pm$2.9) & 2.4 ($\pm$1.0) & 0.6 ($\pm$0.4) & 0.9 ($\pm$0.6) & 1.2 ($\pm$0.7) \\ 
\Flixster & 0.001 ($\pm$0.019) & 0.3 ($\pm$0.1) & 0.5 ($\pm$0.4) & 6.8 ($\pm$3.4) & 1.7 ($\pm$1.2) & 0.6 ($\pm$0.2) & 0.8 ($\pm$0.5) & 1.0 ($\pm$0.6) \\ 
\Pokec & 0.000 ($\pm$0.003) & 0.1 ($\pm$0.1) & 2.7 ($\pm$1.1) & 8.4 ($\pm$5.8) & 1.2 ($\pm$0.4) & 0.2 ($\pm$0.2) & 0.4 ($\pm$0.3) & 0.5 ($\pm$0.3) \\ 
\Flickr & 0.001 ($\pm$0.022) & 0.3 ($\pm$0.2) & 1.2 ($\pm$1.3) & 8.4 ($\pm$5.8) & 2.4 ($\pm$1.1) & 0.5 ($\pm$0.3) & 0.9 ($\pm$0.5) & 1.3 ($\pm$0.5) \\ 
\YouTube & 0.000 ($\pm$0.012) & 0.3 ($\pm$0.2) & 1.7 ($\pm$1.2) & 6.5 ($\pm$2.3) & 1.4 ($\pm$0.6) & 0.6 ($\pm$0.5) & 0.9 ($\pm$0.5) & 1.2 ($\pm$0.6) \\ 
\LiveJournal & 0.000 ($\pm$0.009) & 0.2 ($\pm$0.1) & 2.5 ($\pm$1.1) & 16.4 ($\pm$6.8) & 0.7 ($\pm$0.5) & 0.4 ($\pm$0.3) & 1.0 ($\pm$0.6) & 1.4 ($\pm$0.7) \\ 

%% file: appendix.tex
\section{Omitted Proofs}

\subsection{Useful Tools}
The following lemma is well-known (see, e.g., \cite[Lemma~2.2]{bhattacharyya22property}).
\begin{lemma}
\label{lem:sum}
	Let $a,b\in\mathbb{R}$ with $a<b$, $n\in\mathbb{N}$ and $x\in[a,b]^n$.
	Let $\Sigma = \sum_{i=1}^n x_i$.
	For any $\varepsilon,\delta\in(0,1)$, there exists an algorithm which
	samples a set $S$ of $s=O(\varepsilon^{-2} \lg \delta^{-1})$ indices from $[n]$ and returns an
	estimate $\tilde{\Sigma}=\frac{n}{s}\sum_{i\in S}x_i$ such that $\abs{\Sigma - \tilde{\Sigma}} \leq \varepsilon(b-a) n$ with
	probability at least $1-\delta$.
	The algorithm takes time $O(\varepsilon^{-2} \lg \delta^{-1})$.
\end{lemma}

We will also use the following solver for Laplacian systems.

\begin{theorem}[Andoni et al.~\cite{andoni2019solving}]
\label{thm:sdd-solver}
	Consider a symmetric diagonally dominant (SDD) matrix $S \in \mathbb{R}^{n\times n}$,
	a vector $b\in \mathbb{R}^n$, $u\in [n],\varepsilon>0$, and $\bar{\kappa} \geq 1$, 
	where
    \begin{itemize}
        \item $b\in \mathbb{R}^n$ is in the range of $S$ (equivalently, orthogonal to the kernel of $S$),
        \item $\bar{\kappa}$ is an upper bound on the condition number
		$\kappa(\tilde{S})$, where $\tilde{S} \triangleq D^{-1/2}SD^{-1/2}$ and
		$D \triangleq \Diag{S_{11}, \dots, S_{nn}} $.
    \end{itemize}
	There exists a randomized algorithm that, given $S$, $b$, $u$, $\varepsilon$, and $\bar{\kappa}$,
	outputs $\hat{x}_u \in \mathbb{R}$ with the following guarantee. Suppose
	$x^*$ is the solution for $Sx=b$, then 
    $$
    \forall u \in[n], \quad \Prob{\left|\hat{x}_u-x_u^*\right| \leq \epsilon\left\|x^*\right\|_{\infty}} \geq 1-\frac{1}{r}
    $$
    for suitable $r=O\left(\bar{\kappa} \log \left(\epsilon^{-1} \bar{\kappa}\|b\|_0 \cdot \frac{\max _{i \in[n]} D_{i i}}{\min _{i \in[n]} D_{i i}}\right)\right)$. The algorithm runs in time $O\left(f \epsilon^{-2} r^3 \log r\right)$, where $f$ is the time to make a step in a random walk in the weighted graph formed by the non-zeros of $S$.
\end{theorem}

Next, we state a conservation law between different measures. We note that the
version of this law in~\cite{chen2018quantifying} is stated for mean-centered
opinions but it also holds when the opinions are not centered. We use the law in
the discussion in Section~\ref{sec:experiments}.
\begin{lemma}[Chen et al.~\cite{chen2018quantifying}]
\label{lem:conservation-law}
	We have the following conservation law:
	$\Internal + 2\Disagreement + \Controversy = \norm{s}_2^2$.
\end{lemma}
\begin{proof}
	First, note that
	\begin{align*}
		\Internal &= \sum_{u\in V} (s_u - z_u^*)^2 \\
		&= \sum_{u\in V} s_u^2
			- 2 \sum_{u\in V} s_u z_u^*
			+ \sum_{u\in V} (z_u^*)^2 \\
		&= \norm{s}_2^2
			- 2 s^\intercal (I+L)^{-1} s
			+ s^\intercal (I+L)^{-2} s \\
		&= s^\intercal (I+L)^{-1} ((I+L)^2 - 2(I+L) + I) (I+L)^{-1} s \\
		&= s^\intercal (I+L)^{-1} (I+ 2L + L^2 - 2(I+L) + I) (I+L)^{-1} s \\
		&= s^\intercal (I+L)^{-1} L^2 (I+L)^{-1} s.
	\end{align*}
	Now since $\Controversy = \sum_u (z_u^*)^2 = \norm{z}_2^2 = s^\intercal
	(I+L)^{-2} s$ and $\Disagreement = s^\intercal (I+L)^{-1} L (I+L)^{-1} s$,
	we get the desired equality since
	\begin{align*}
		\Internal + 2\Disagreement + \Controversy
		&= s^\intercal (I+L)^{-1} (L^2 + 2L + I) (I+L)^{-1} s \\
		&= s^\intercal (I+L)^{-1} (I+L)^2 (I+L)^{-1} s \\
		&= \norm{s}_2^2.
	\end{align*}
\end{proof}

\subsection{Proof of Proposition~\ref{pro:solver}}
\label{sec:proof:pro:solver}
First, we briefly present the details of the lazy random walks with timeout
that are used in Algorithm~\ref{alg:random-walks}. One random-walk step from
vertex~$v$ is performed as follows. With probability~$1/2$, the walk stays
at~$v$. With probability $w_{vw} / (2 (1 + w_{v}))$, it moves to neighbor~$w$
of~$v$. With the remaining probability of $1/2 - w_v / (2(1+w_v))$, the random
walk terminates. Note that this corresponds to a random walk with timeouts on
the matrix~$I+L$.

To obtain the result of the proposition, we use Theorem~\ref{thm:sdd-solver}.
The theorem considers a term $\frac{\max_i D_{ii}}{\min_i D_{ii}}$ which
in our setting becomes $\frac{\max_i (I+L)_{ii}}{\min_i (I+L)_{ii}} \leq
\max_u w_u+1$. Here, we used that $(I+L)_{ii} \geq 1$ for all $i\in[n]$,
and $\norm{z^*}_\infty = \max_u z_u^*\leq 1$. To apply the theorem, we also
utilize the fact that $\norm{s}_0 \leq n$.

\subsection{Proof of Proposition~\ref{prop:page-rank-equation}}
	First, observe that the expressed equilibrium opinions~$z^*$ must satisfy
	the update rule in Equation~\eqref{eq:update-rule} with equality. Thus, by
	expressing the update rule in matrix notation, we obtain that $z^*$ is the unique solution to the equation
		$z^* = (I+D)^{-1} (s + A z^*).$

	Next, set $M=(I+D)^{-1}$. Then a calculation reveals that:
	\begin{align*}
		M
		&= (I+D)^{-1} \\
		&= ((I+D^{-1}) D)^{-1} \\
		&= (I+D^{-1})^{-1} D^{-1} \\
		&= (I-M) D^{-1},
	\end{align*}
	where the last equality follows from observing that $I+D^{-1}$ and $I-M$ are
	both diagonal matrices and then verifying that for each entry it holds that
	$\left[ (I+D^{-1})^{-1} \right]_{ii} = [I-M]_{ii}$.

	By combining the last two equations we get the statement from the lemma.

\subsection{Proof of Lemma \ref{lem:algo-identity}}\label{app:pagerank}

\subsubsection{Local Personalized PageRank with (Non-Lazy) Random Walks}
\label{sec:personalized-pagerank}
We start by discussing how the algorithm by Andersen et
al.~\cite{andersen2006local} can be adapted to work with (non-lazy) random walks.

We let $\alpha\in(0,1)$ be the teleportation constant, $s\in[0,1]^n$ be a vector
consisting of a distribution and we set $W_s = D^{-1} A$. Then we define
$\pr'(\alpha,s)$ as the unique solution of the equation
$\pr'(\alpha,s) = \alpha s + (1-\alpha) W_s \pr'(\alpha, s)$.
Note that this differs from the classic personalized PageRank only by fact that
we use (non-lazy) random walks (based on $W_s$) rather than lazy random walks
(based on $W = \frac{1}{2}(I + D^{-1}A)$).

Recall that 
\begin{align}
	\pr'(\alpha,s) = \alpha \sum_{i=0}^\infty (1-\alpha)^i (W_s^i s).
\end{align}
Thus, we get that
\begin{align*}
	\pr'(\alpha,s)
	&= R' s \\
	&= \alpha s + (1-\alpha) R' W_s s \\
	&= \alpha s + (1-\alpha) R' \pr'(\alpha, W_s s).
\end{align*}

Similar to~\cite{andersen2006local} this means that we can maintain a vector $p$
that approximates $\pr'(\alpha,s)$ and a residual vector~$r$.  Then we can
perform push-operations in which move an $\alpha$-fraction of weight from $r$ to
$p$ and then we distribute the weight in $r$ based on $W_s$.  Algorithmically,
this is formally stated in Algorithm~\ref{alg:push} and we use it as a
subroutine in Algorithm~\ref{alg:local}.

\begin{lemma}[Andersen et al.~\cite{andersen2006local}]
\label{lem:performance-local}
Algorithm~\ref{alg:local} runs in time $O(\frac{1}{\varepsilon \alpha})$, and computes a vector $p$ which approximates $\pr'(\alpha,s)$ with a residual vector~$r$ such that $\max_{u\in V}\frac{r(u)}{d_u}< \varepsilon$ and $\mathrm{vol}(\mathrm{Supp}(p)) \leq \frac{1}{\varepsilon \alpha}$. Here, $\mathrm{vol}(S)=\sum_{x\in S}d_x$ is the volume of a subset $S \subseteq V$, and $\mathrm{Supp}(p)=\{v \mid p(v) \neq 0\}$ is the support of a distribution $p$.
\end{lemma}

	\begin{algorithm}[!t]
	\begin{algorithmic}[1]
 \REQUIRE A graph $G=(V,E,w)$, a parameter $\alpha \in (0,1)$ and an error parameter $\varepsilon$
		\STATE $p\gets \vec{0}$ and $r\gets \mathbbm{1}_v$
		\REPEAT
			\STATE	Choose any vertex $u$ where $\frac{r(u)}{d_u} \geq \varepsilon$
				or where $r(u)=1$
			\STATE	Apply Algorithm~\ref{alg:push} at vertex $u$, updating $p$ and $r$
		\UNTIL{$\max_{u\in V} \frac{r(u)}{d_u} < \varepsilon$}
		\RETURN $(p, r)$
	\end{algorithmic}
	\caption{Local personalized PageRank algorithm}
	\label{alg:local}
	\end{algorithm}

	\begin{algorithm}[!t]
	\begin{algorithmic}[1]
 \REQUIRE A graph $G=(V,E,w)$, a vector $p$ and a residual vector $r$
		\STATE $p'\gets p$ and $r'\gets r$
		\STATE	$p'(u)=p(u)+\alpha r(u)$
  \STATE	$r'(u)=(1-\alpha) r(u)$
		\FOR{each $v$ such that $(u,v)\in E$}
			\STATE $r'(v)=r(v)+(1-\alpha)r(u)/d_u$
		\ENDFOR
		\RETURN $(p', r')$
	\end{algorithmic}
	\caption{$\textsc{Push}_u$}
	\label{alg:push}
	\end{algorithm}

\subsubsection{Proof of Lemma \ref{lem:algo-identity}}

Next, we prove Lemma \ref{lem:algo-identity}. 
First, observe that $z_u^* = \mathbbm{1}_u^{\intercal} z^*$.
Hence, we have that 
$z_u^* = \mathbbm{1}_u^{\intercal} z^*
	= \mathbbm{1}_u^{\intercal} (I+L)^{-1} s$, where $s$ is the vector consisting of the innate opinions of all vertices.
Now our strategy is to show that $\pr'(\alpha,s) = (I+L)^{-1} \mathbbm{1}_u^{\intercal}$ 
and then we will argue that Algorithm~\ref{alg:pagerank} returns a good enough
approximation $p$ of $\pr'(\alpha,s)$ such that
$p^\intercal s \approx \pr'(\alpha,s) s 
	= \mathbbm{1}_u^{\intercal} (I+L)^{-1} s = z_u^*$.
When in Algorithm~\ref{alg:pagerank} we use the algorithm
from~\cite{andersen2006local} as a subroutine, we actually use
Algorithm~\ref{alg:local}.

We first give a closed form of $(I+L)^{-1}$ for any $d$-degree bounded graphs,
i.e., graphs with maximum degree at most $d$.

\begin{lemma}
\label{lem:identity}
	For $d$-degree bounded graphs, we have that
	$(I+L)^{-1} = \left(
				\sum_{i=0}^{\infty} 
				((I-(D+I)^{-1}) D^{-1} A)^i
			\right)
			(D+I)^{-1}$.
\end{lemma}
\begin{proof}
	First, we have that
	\begin{align*}
		&(I+L)^{-1}\\
		=~&(I+D-A)^{-1} \\
		=~&[D^{1/2}(D^{-1}+I-D^{-1/2} A D^{-1/2})D^{1/2}]^{-1} \\
		=~&D^{-1/2}[D^{-1}+I-D^{-1/2} A D^{-1/2}]^{-1}D^{-1/2} \\
		=~&D^{-1/2}[(D^{-1}+I)(I - (D^{-1}+I)^{-1} (D^{-1/2} A D^{-1/2}))]^{-1}D^{-1/2} \\
		=~&D^{-1/2}[I - (D^{-1}+I)^{-1} (D^{-1/2} A D^{-1/2})]^{-1} (D^{-1}+I)^{-1} D^{-1/2} \\
		=~&D^{-1/2}
			\left(
				\sum_{i=0}^{\infty} ((D^{-1}+I)^{-1} (D^{-1/2} A D^{-1/2}))^{i}
			\right)
			(D^{-1}+I)^{-1}
			D^{-1/2}.
	\end{align*}
	Next, observe the identities $(D^{-1} + I)^{-1} = D (D+I)^{-1}$ and
	$(D+I)^{-1}D = I - (D+I)^{-1}$ which can be checked for each diagonal
	element.  Now using that the matrix multiplication of diagonal matrices is
	commutative, we obtain that
	\begin{align*}
		&((D^{-1}+I)^{-1} (D^{-1/2} A D^{-1/2}))^{i} \\
		=~&(D^{-1}+I)^{-1} D^{-1/2} A D^{-1/2}
			\cdot (D^{-1}+I)^{-1} D^{-1/2} A D^{-1/2}
			\cdots\\
			&(D^{-1}+I)^{-1} D^{-1/2} A D^{-1/2} \\
		=~&D (D+I)^{-1} D^{-1/2} A D^{-1/2}
			\cdot D (D+I)^{-1} D^{-1/2} A D^{-1/2}
			\cdots\\
			&D (D+I)^{-1} D^{-1/2} A D^{-1/2} \\
		=~&D^{1/2} ((D+I)^{-1} A)^i D^{-1/2} \\
		=~&D^{1/2} ((D+I)^{-1} D \cdot D^{-1} A)^i D^{-1/2} \\
		=~&D^{1/2} ((I-(D+I)^{-1}) D^{-1} A)^i D^{-1/2}. 
	\end{align*}
	Thus, we can continue our calculation from above to get that
	\begin{align*}
		&(I+L)^{-1}\\
		=~&D^{-1/2}
			\left(
				\sum_{i=0}^{\infty} ((D^{-1}+I)^{-1} (D^{-1/2} A D^{-1/2}))^{i}
			\right)
			(D^{-1}+I)^{-1}
			D^{-1/2} \\
		=~&D^{-1/2}
			\left(
				\sum_{i=0}^{\infty} 
				D^{1/2} ((I-(D+I)^{-1}) D^{-1} A)^i D^{-1/2}
			\right)
			(D^{-1}+I)^{-1}
			D^{-1/2} \\
		=~&\left(
				\sum_{i=0}^{\infty} 
				((I-(D+I)^{-1}) D^{-1} A)^i
			\right)
			(D^{-1}+I)^{-1}
			D^{-1} \\
		=~&\left(
				\sum_{i=0}^{\infty} 
				((I-(D+I)^{-1}) D^{-1} A)^i
			\right)
			(D+I)^{-1},
	\end{align*}
	where in the final step we again used the identity from above.
\end{proof}

Next, set $M=(D+I)^{-1}$. Note that then the equality in
Lemma~\ref{lem:identity} becomes
$(I+L)^{-1} = \left(
				\sum_{i=0}^{\infty} 
				((I-M) D^{-1} A)^i
			\right)
			M$.

In the following, we assume the input graph is $d$-regular, i.e., $D_{ii}=d$ for
all $i\in V$. Note that this implies that $M=\frac{1}{d+1}I$. By setting $\alpha=\frac{1}{d+1}$, 
Lemma~\ref{lem:identity} implies that
\begin{align*}
	(I+L)^{-1}
	&= \left( \sum_{i=0}^{\infty} ((I-M) D^{-1} A)^i \right) M \\
	&= \alpha \sum_{i=0}^\infty \left((1-\alpha)D^{-1}A\right)^i \\
	&= \alpha \sum_{i=0}^\infty \left((1-\alpha) W_s\right)^i.
\end{align*}
We therefore get the following corollary.
\begin{corollary}
	Let $\alpha=\frac{1}{d+1}$ and $s$ be the vector of innate opinions.
	Then for $d$-regular graphs it holds that $\pr'(\alpha,s) = (I+L)^{-1} s$.
\end{corollary}
\begin{proof}
	This follows from combining the identity
	$\pr'(\alpha,s) = \alpha \sum_{i=0}^\infty (1-\alpha)^i (W_s^i s)$
	from Section~\ref{sec:personalized-pagerank}
	with the discussion of Lemma~\ref{lem:identity} for $d$-regular graphs.
\end{proof}

\begin{proof}[Proof of Lemma \ref{lem:algo-identity}]
	Observe that this property holds when we initialize $p\gets0$ and
	$r\gets \mathbbm{1}_u$ at the beginning of Algorithm~\ref{alg:pagerank}.

	Now suppose one iteration of the while-loop ended.  Then by induction it is
	enough if we show that
	\begin{align*}
		\pr'(\alpha,r) = \sum_i r(i) p^{(i)} + \pr'\left(\alpha,\sum_i r^{(i)}\right).
	\end{align*}

	First, observe that $r = \sum_i r(i) \mathbbm{1}_i$.
	Next, note that for each $i$, it holds that
	$p^{(i)} + \pr'(\alpha,r^{(i)}) = \pr'(\alpha,\mathbbm{1}_i)$.
	Using Equation~\eqref{eq:pr-identity} we now obtain that
	\begin{align*}
		\pr'(\alpha,r)
		&= \pr'\left(\alpha, \sum_i r(i) \mathbbm{1}_i\right) \\
		&= \alpha \sum_{j=0}^\infty (1-\alpha)^j
				\left( W_s^j \left(\sum_i r(i) \mathbbm{1}_i\right) \right) \\
		&= \sum_i r(i) \alpha \sum_{j=0}^\infty (1-\alpha)^j ( W_s^j \mathbbm{1}_i ) \\
		&= \sum_i r(i) \pr'(\alpha,\mathbbm{1}_i) \\
		&= \sum_i r(i) (p^{(i)} + \pr'(\alpha,r^{(i)})).
	\end{align*}
	This proves the statement of the lemma.
\end{proof}

	For intuition on the correctness of our algorithm, observe that the first
	part of the sum $\sum_i r(i) p^{(i)}$ are exactly the changes that our
	algorithm makes to the vector~$p$. Furthermore, observe that again applying
	Equation~\eqref{eq:pr-identity}, we get that
	\begin{align*}
		\sum_i r(i) \pr'(\alpha,r^{(i)})
		&= \sum_i r(i) \alpha \sum_{j=0}^\infty (1-\alpha)^j \left( W_s^j r^{(i)} \right) \\
		&= \alpha \sum_{j=0}^\infty (1-\alpha)^j
				\left( W_s^j \left(\sum_i r(i) r^{(i)}\right) \right) \\
		&= \pr'\left(\alpha,\sum_i r(i) r^{(i)}\right)
	\end{align*}
	which is exactly why we set $r \gets \sum_i r(i) r^{(i)}$ in the algorithm.

\subsection{Proof of Lemma~\ref{lem:approx-error}}
	Using Lemma~\ref{lem:algo-identity} we get that 
	\begin{align*}
		\abs{z_u^* - p^{\intercal} s}
		= \abs{\pr'(\alpha,\mathbbm{1}_u)^{\intercal} s - p^{\intercal} s}
		= \abs{\pr'(\alpha,r)^\intercal s}
		\leq \norm{\pr'(\alpha,r)}_1,
	\end{align*}
	where in the last step we used that $s_u\in[0,1]$ for all $u\in V$.

	Now observe that since $W_s=D^{-1}A$ is a row-stochastic matrix, we have that
	$\norm{r^{\intercal} W_s^i}_1 \leq \norm{r}_1$ for all $i$. Hence, by
	Equation~\eqref{eq:pr-identity}, we get that
	\begin{align*}
		\norm{\pr'(\alpha,r)}_1
		\leq \alpha \norm{r}_1 \sum_i^\infty (1-\alpha)^i
		= \norm{r}_1
		\leq \varepsilon,
	\end{align*}
	where in the last step we used that our algorithm only terminates when
	$\norm{r}_1 \leq \varepsilon$.

\subsection{Proof of Lemma~\ref{lem:vector-drop}}
First, recall that Line~\ref{alg:personalizedpr}
	(i.e., the local personalized PageRank) in
	Algorithm~\ref{alg:pagerank} consists of the two subroutines in
	Algorithms~\ref{alg:local} and~\ref{alg:push} in Section~\ref{app:pagerank}.
	
	To prove the first claim, consider any iteration of the for-loop. Observe
	that we start the PageRank algorithm {(see
	Algorithm \ref{alg:local})} with residual vector $r = \mathbbm{1}_i$.
	Note that when the PageRank algorithm performs the first push-operation (see
	Algorithm \ref{alg:push}), this $\ell_1$-norm of the residual vector drops
	by a factor of $\alpha=\frac{1}{d+1}$.  Hence, for every entry $r(i)$ we get
	that
	\begin{align*}
		\norm{r(i) \mathbbm{1}_i}_1
		\leq r(i) \left(1-\frac{1}{d+1}\right) \norm{\mathbbm{1}_i}_1
		= r(i) \left(1-\frac{1}{d+1}\right).
	\end{align*}
	Since this holds for all $i$ and we assume that $d=O(1)$, we obtain that
	$\norm{r'}_1 \leq \left(1-\frac{1}{d+1}\right) \norm{r}_1$ at the end of
	each while-loop, where $r'$ is the resulting vector after applying the PageRank algorithm over $r$ after one iteration. 

	The second claim follows from the fact (see \cite{andersen2006local}) that the PageRank algorithm creates
	at most $O(d/\varepsilon)$ new non-zero entries when called upon any vector
	$\mathbbm{1}_i$. As we call the PageRank algorithm for every $i$ with $r(i)\neq0$,
	this increases the number of non-zero entries $p$ and $r$ by a factor of at
	most $O(d/\varepsilon)$.

\subsection{Proof of Corollary~\ref{cor:running-time}}
By Lemma~\ref{lem:vector-drop}, we get that $\norm{r}_1$ drops by a constant
	factor after each iteration of the while-loop. Hence, if we perform
	$k=(d+1)\lg(1/\varepsilon)$~iterations of the while-loop we have that
	\begin{align*}
		\left(1-\frac{1}{d+1}\right)^k
		\leq \exp\left( - \frac{k}{d+1} \right)
		\leq \varepsilon.
	\end{align*}
	Hence, we only have to perform $O(d \lg(1/\varepsilon))$~iterations of the
	while-loop.
	
	Note that each iteration of the while-loop increases the number of non-zero
	entries in $r$ by at most $O(d/\varepsilon)$. Hence, the number of
	non-zeros in $r$ is bounded by
	$(d/\varepsilon)^{O(k)} = (d/\varepsilon)^{O(d \lg(1/\varepsilon))}$ at all
	times. 
 
	By Lemma~\ref{lem:performance-local} in Section~\ref{app:pagerank}, we need
	time $O(\frac{1}{\varepsilon \alpha}) = O(\frac{d}{\varepsilon})$ every time
	we use the local personalized PageRank algorithm on a vector
	$\mathbbm{1}_i$ since $\alpha = \frac{1}{d+1}$.  Combining this with our
	observations from above, each iteration of the while-loop in
	Algorithm~\ref{alg:pagerank} takes time $O(d/\varepsilon) \cdot \norm{r}_0$,
	where $\norm{r}_0$ is the number of non-zero entries in $r$. Thus we obtain
	a total running time of $(d/\varepsilon)^{O(d \lg(1/\varepsilon))}$.

\subsection{Estimating the Measures Given Oracle Access to Innate Opinions~$s_u$}
\label{app:measures-oracle-innate}

Here, we prove the bounds claimed in the penultimate column of
Table~\ref{tab:measures}.

\begin{lemma}
\label{lem:estimation-oracle-innate-opinions}
	Let $\varepsilon,\delta\in(0,1)$,  $\bar{\kappa}$ be an upper bound on
	$\kappa(\tilde{S})$, where $\tilde{S} \triangleq (I+D)^{-1/2} (I+L) (I+D)^{-1/2}$, and
	$r=O(\bar{\kappa} \log(\varepsilon^{-1} n\bar{\kappa} (\max_u w_u) ))$. Then
	with probability at least $1-\delta$:
	\begin{itemize}
		\item We can return an estimate of $\avgExpressed$ with additive error
			$\pm\varepsilon$ in time $O(\varepsilon^{-2}\lg\delta^{-1})$.
		\item We can return an estimate of $\SumOpinions$ and $\norm{s}_2^2$ with additive error
			$\pm\varepsilon n$ in time $O(\varepsilon^{-2}\lg\delta^{-1})$.
		\item We can return an estimate of $\Controversy, \Polarization, \Internal, \Disagreement$ and $\DisCon$ with additive error
			$\pm\varepsilon n$ in time $O(\varepsilon^{-4}r^3  \log \delta^{-1} \lg r)$.
	\end{itemize}
\end{lemma}
\begin{proof}
	\emph{Estimating $\avgExpressed$, $\SumOpinions$ and $\norm{s}_2^2$:}
 First, we note that $\sum_{u\in V} s_u = \sum_{u\in V} z_u^*$ since
	$s=(I+L)z^*$ and hence we have that
	\begin{align*}
		\sum_{u\in V} s_u
		= \mathbbm{1}^\intercal s
		= \mathbbm{1}^\intercal (I+L) z^*
		= \mathbbm{1}^\intercal z^*
		= \sum_{u\in V} z_u^*.
	\end{align*}
	Hence, to estimate $\avgExpressed$ and $\SumOpinions$ it suffices to
	estimate $\sum_{u\in V} s_u$. 
 In addition, $\norm{s}_2^2=\sum_{u\in V} s_u^2$.
 Given that we have query access to $s$, we can apply
	Lemma~\ref{lem:sum} to obtain our results. 

	\emph{Estimating $\Controversy$:} Recall that $\Controversy=\sum_{u\in V}(z_u^*)^2$. We set 
         $\varepsilon_1 = \frac{\varepsilon}{6}$,  $r=O(\bar{\kappa} \log(\varepsilon_1^{-1} n\bar{\kappa} 
				(\max_u w_u) ))$, $\delta_1=\frac{1}{r}=\frac{\delta}{2C}$,
	$\varepsilon_2 = \frac{\varepsilon}{2}$, $\delta_2 = \frac{\delta}{2}$ and $C = \varepsilon_2^{-2} \log \delta_2^{-1}$.
        To return an estimate of $\Controversy$ with additive error $\pm\varepsilon n$ and success probability $1-\delta$, we perform
	the following procedure. We sample
	$C$~vertices, i.e., $i_1, \dots, i_{C}$, from $V$ using Lemma~\ref{lem:sum} and obtain $\tilde{z}^*_{i_1}, \dots, \tilde{z}^*_{i_{C}}$ using
	Proposition~\ref{pro:solver} with error parameter $\varepsilon_1$ and success probability $1-\delta_1$. 
        We return $\frac{n}{C} \sum_{j=1}^{C} (\tilde{z}^*_{i_j})^2$.

        Next, we analyze the running time of this procedure and we also prove the error guarantee.

	We start with the running time analysis. According to Proposition~\ref{pro:solver}, estimating each $z^*_u$ takes time $T_1=O(\varepsilon_1^{-2} r^3 \lg r)=O(\varepsilon^{-2} r^3 \lg r)$. According to Lemma~\ref{lem:sum}, sampling $C$ vertices from $V$ takes time $T_2 = O(\varepsilon_2^{-2} \log \delta_2^{-1})= O(\varepsilon^{-2} \log \delta^{-1})$. Therefore, the running time is $C T_1+T_2=O(\varepsilon^{-4}r^3  \log \delta^{-1} \lg r)$. 

        Next, we analyze the error guarantee. According to Proposition~\ref{pro:solver}, for each $j\in [C]$, with probability at least $1-\delta_1$, we have $\abs{\tilde{z}^*_{i_j}-z^*_{i_j}} \leq \varepsilon_1=\frac{\varepsilon}{6}$. Therefore, $\abs{(\tilde{z}^*_{i_j})^2-(z^*_{i_j})^2} \leq \abs{\tilde{z}^*_{i_j}+z^*_{i_j}} \cdot \abs{\tilde{z}^*_{i_j}-z^*_{i_j}} \leq 3\varepsilon_1 = \frac{\varepsilon}{2}$.
        Then by union bound, with probability at least $1-C \cdot \delta_1=1-\frac{\delta}{2}$, we have $\abs{\sum_{j=1}^{C} (\tilde{z}^*_{i_j})^2 - \sum_{j=1}^{C} (z^*_{i_j})^2} \leq C\cdot \frac{\varepsilon}{2} = \frac{\varepsilon C}{2}$.
        According to Lemma~\ref{lem:sum}, with probability at least $1-\delta_2=1-\frac{\delta}{2}$, we have $\abs{\frac{n}{C}\sum_{j=1}^{C} (z^*_{i_j})^2 - \Controversy}\leq \varepsilon_2 n = \frac{\varepsilon n}{2}$. By union bound, with probability at least $1-\frac{\delta}{2}-\frac{\delta}{2}=1-\delta$, we have 
        \begin{align*}
            &\abs{\frac{n}{C}\sum_{j=1}^{C} (\tilde{z}^*_{i_j})^2 - \Controversy}\\
            \leq~ 
            &\abs{\frac{n}{C}\sum_{j=1}^{C} (\tilde{z}^*_{i_j})^2 - \frac{n}{C}\sum_{j=1}^{C} (z^*_{i_j})^2} + \abs{\frac{n}{C}\sum_{j=1}^{C} (z^*_{i_j})^2 - \Controversy}\\
            \leq~&\frac{n}{C}\cdot\frac{ \varepsilon C}{2} + \frac{\varepsilon n}{2}\\
            =~&\varepsilon n.
        \end{align*}

	\emph{Estimating $\Internal$:} Recall that $\Internal=\sum_{u\in V} (s_u - z_u^*)^2$. We set 
         $\varepsilon_1 = \frac{\varepsilon}{6}$,  $r=O(\bar{\kappa} \log(\varepsilon_1^{-1} n\bar{\kappa} 
				(\max_u w_u) ))$, $\delta_1=\frac{1}{r}=\frac{\delta}{2C}$,
	$\varepsilon_2 = \frac{\varepsilon}{2}$, $\delta_2 = \frac{\delta}{2}$ and $C = \varepsilon_2^{-2} \log \delta_2^{-1}$.
        To return an estimate of $\Internal$ with additive error $\pm \varepsilon n$ and success probability $1-\delta$, we perform
	the following procedure. We sample
	$C$~vertices, i.e., $i_1, \dots, i_{C}$, from $V$ and query $s_{i_1}, \dots,
	s_{i_{C}}$ using Lemma~\ref{lem:sum} and $\tilde{z}^*_{i_1}, \dots, \tilde{z}^*_{i_{C}}$ using
	Proposition~\ref{pro:solver} with error parameter $\varepsilon_2$ and success probability $1-\delta_1$.
        We return $\frac{n}{C} \sum_{j=1}^{C} (s_{i_j}-\tilde{z}^*_{i_j})^2$.

        Next, we analyze the running time of this procedure and we also prove the error guarantee.

	We start with the running time analysis. According to Proposition~\ref{pro:solver}, estimating each $z^*_u$ takes time $T_1=O(\varepsilon_1^{-2} r^3 \lg r)=O(\varepsilon^{-2} r^3 \lg r)$. According to Lemma~\ref{lem:sum}, sampling $C$ vertices from $V$ takes time $T_2 = O(\varepsilon_2^{-2} \log \delta_2^{-1})= O(\varepsilon^{-2} \log \delta^{-1})$. Therefore, the running time is $C T_1 +T_2=O(\varepsilon^{-4}  r^3 \log \delta^{-1} \lg r)$. 

        Next, we analyze the error guarantee. According to Proposition~\ref{pro:solver}, for each $j\in [C]$, with probability at least $1-\delta_1$, we have $\abs{\tilde{z}^*_{i_j}-z^*_{i_j}} \leq \varepsilon_1 = \frac{\varepsilon}{6}$. Therefore, $\abs{(s_{i_j}-\tilde{z}^*_{i_j})^2-(s_{i_j}-z_{i_j}^*)^2} \leq \abs{\tilde{z}^*_{i_j}+z^*_{i_j}-2s_{i_j}} \cdot \abs{\tilde{z}^*_{i_j}-z^*_{i_j}} \leq 3\varepsilon_1 = \frac{\varepsilon}{2}$.
        Then by union bound, with probability at least $1-C\cdot \delta_1=1-\frac{\delta}{2}$, we have $\abs{\sum_{j=1}^{C} (s_{i_j}-\tilde{z}_{i_j}^*)^2 - \sum_{j=1}^{C} (s_{i_j}-z_{i_j}^*)^2} \leq C\cdot \frac{\varepsilon}{2} = \frac{\varepsilon C}{2}$.
        According to Lemma~\ref{lem:sum}, with probability at least $1-\delta_2 = 1-\frac{\delta}{2}$, we have $\abs{\frac{n}{C}\sum_{j=1}^{C} (s_{i_j}-z_{i_j}^*)^2 - \Internal}\leq \varepsilon_2 n= \frac{\varepsilon n}{2}$. By union bound, with probability at least $1-\frac{\delta}{2}-\frac{\delta}{2}=1-\delta$, we have 
        \begin{align*}
           &\abs{\frac{n}{C}\sum_{j=1}^{C} (s_{i_j}-\tilde{z}_{i_j}^*)^2 - \Internal}\\
            \leq~ 
            &\abs{\frac{n}{C}\sum_{j=1}^{C} (s_{i_j}-\tilde{z}_{i_j}^*)^2 - \frac{n}{C}\sum_{j=1}^{C} (s_{i_j}-z_{i_j}^*)^2} + \abs{\frac{n}{C}\sum_{j=1}^{C} (s_{i_j}-z_{i_j}^*)^2 - \Internal}\\
            \leq~&\frac{n}{C}\cdot \frac{\varepsilon C}{2} + \frac{\varepsilon n}{2} \\
            =~&\varepsilon n.
        \end{align*}

\emph{Estimating $\DisCon$:} Note that $\DisCon = s^\intercal z^* = \sum_{u\in V} s_u z_u^*$.
  We set 
         $\varepsilon_3 = \frac{\varepsilon}{2}$,  $r=O(\bar{\kappa} \log(\varepsilon_3^{-1} n\bar{\kappa} 
				(\max_u w_u) ))$, $\delta_3=\frac{1}{r}=\frac{\delta}{2C}$,
	$\varepsilon_2 = \frac{\varepsilon}{2}$, $\delta_2 = \frac{\delta}{2}$ and $C = \varepsilon_2^{-2} \log \delta_2^{-1}$.
        To return an estimate of $\DisCon$ with additive error $\pm\varepsilon n$ and success probability $1-\delta$, we perform
	the following procedure. We sample
	$C$~vertices, i.e., $i_1, \dots, i_{C}$, from $V$ and query $s_{i_1}, \dots,
	s_{i_{C}}$ using Lemma~\ref{lem:sum} and obtain $\tilde{z}^*_{i_1}, \dots, \tilde{z}^*_{i_{C}}$ using
	Proposition~\ref{pro:solver} with error parameter $\varepsilon_3$ and success probability $1-\delta_1$. 
        We return $\frac{n}{C} \sum_{j=1}^{C} s_{i_j}\tilde{z}^*_{i_j}$. 

        Next, we analyze the running time of this procedure and we also prove the error guarantee.

	We start with the running time analysis. According to Proposition~\ref{pro:solver}, estimating each $z^*_u$ takes time $T_3=O(\varepsilon_3^{-2} r^3 \lg r)=O(\varepsilon^{-2} r^3 \lg r)$. According to Lemma~\ref{lem:sum}, sampling $C$ vertices from $V$ takes time $T_2 = O(\varepsilon_2^{-2} \log \delta_2^{-1})= O(\varepsilon^{-2} \log \delta^{-1})$. Therefore, the running time is $C T_3+T_2=O(\varepsilon^{-4}r^3  \log \delta^{-1} \lg r)$. 

        Next, we analyze the error guarantee. According to Proposition~\ref{pro:solver}, for each $j\in [C]$, with probability at least $1-\delta_3$, we have $\abs{\tilde{z}^*_{i_j}-z^*_{i_j}} \leq \varepsilon_3=\frac{\varepsilon}{2}$. Therefore, $\abs{s_{i_j}\tilde{z}^*_{i_j}-s_{i_j}z^*_{i_j}} \leq s_{i_j} \cdot \abs{\tilde{z}^*_{i_j}-z^*_{i_j}} \leq \varepsilon_3 = \frac{\varepsilon}{2}$.
        Then by union bound, with probability at least $1-C \cdot \delta_3=1-\frac{\delta}{2}$, we have $\abs{\sum_{j=1}^{C} s_{i_j}\tilde{z}^*_{i_j} - \sum_{j=1}^{C} s_{i_j}z^*_{i_j}} \leq C\cdot \frac{\varepsilon}{2} = \frac{\varepsilon C}{2}$.
        According to Lemma~\ref{lem:sum}, with probability at least $1-\delta_2=1-\frac{\delta}{2}$, we have $\abs{\frac{n}{C}\sum_{j=1}^{C} s_{i_j}z^*_{i_j} - \DisCon}\leq \varepsilon_2 n = \frac{\varepsilon n}{2}$. By union bound, with probability at least $1-\frac{\delta}{2}-\frac{\delta}{2}=1-\delta$, we have 
        \begin{align*}
            &\abs{\frac{n}{C}\sum_{j=1}^{C} s_{i_j}\tilde{z}^*_{i_j} - \DisCon}\\
            \leq~
            &\abs{\frac{n}{C}\sum_{j=1}^{C} s_{i_j}\tilde{z}^*_{i_j} - \frac{n}{C}\sum_{j=1}^{C} s_{i_j}z^*_{i_j}} + \abs{\frac{n}{C}\sum_{j=1}^{C} s_{i_j}z^*_{i_j} - \DisCon}\\
            \leq~&\frac{n}{C}\cdot\frac{ \varepsilon C}{2} + \frac{\varepsilon n}{2} \\
            =~&\varepsilon n.
        \end{align*}
        
        \emph{Estimating $\Polarization$:} 
        Recall that
	$\Polarization = \sum_{u\in V} (z_u^* - \avgExpressed)^2$, where we use
	$\avgExpressed = \frac{1}{n} \sum_{u\in V} z_u^*$. Now we use the well-known
	equality that $\sum_i \sum_{j>i} (a_i - a_j)^2 = n \sum_{i} (a_i - c)^2$,
	where $c=\frac{1}{n} \sum_{i} a_i$. This gives us that
	\begin{equation*}
		\Polarization
		= \sum_{u\in V} (z_u^* - \avgExpressed)^2 
		= \frac{1}{2n} \sum_{u,v\in V} (z_u^* - z_v^*)^2.
	\end{equation*}
	We set 
        $\varepsilon_4 = \frac{\varepsilon}{18}$,  $r=O(\bar{\kappa} \log(\varepsilon_4^{-1} n\bar{\kappa} 
				(\max_u w_u) ))$, $\delta_4=\frac{1}{r}=\frac{\delta}{4C}$,
	$\varepsilon_2 = \frac{\varepsilon}{2}$, $\delta_2 = \frac{\delta}{2}$ and $C = \varepsilon_2^{-2} \log \delta_2^{-1}$.
Consider a vector $x$ of length $n^2$
	which has entries $x_{u,v} = (z_u^* - z_v^*)^2 \in [0,1]$ for $u,v\in V$.
 Therefore, $\Polarization = \frac{1}{2n} \sum_{u,v\in V} x_{u,v}$. We sample
	$C$~indices, denoted as $i_1, \dots, i_{C}$, from $x$ using Lemma~\ref{lem:sum}. Then we obtain $\tilde{x}_{i_1}, \dots, \tilde{x}_{i_{C}}$ using
	Proposition~\ref{pro:solver}. Note that for each $j\in [C]$, $\tilde{x}_{i_j} = (\tilde{z}^*_{j_1}-\tilde{z}^*_{j_2})^2$ supposing that $j_1$ and $j_2$ are vertices associated with $x_{i_j}$. 
        We return $\frac{n^2}{C} \sum_{j=1}^{C} \tilde{x}_{i_j}$.

        Next, we analyze the running time of this procedure and we also prove the error guarantee.

	We start with the running time analysis. According to Proposition~\ref{pro:solver}, estimating each $z^*_u$ takes time $T_4=O(\varepsilon_4^{-2} r^3 \lg r)=O(\varepsilon^{-2} r^3 \lg r)$. According to Lemma~\ref{lem:sum}, sampling $C$ entries from $x$ takes time $T_2 = O(\varepsilon_2^{-2} \log \delta_2^{-1}) = O(\varepsilon^{-2} \log \delta^{-1})$. Therefore, the running time is at most $2C T_4+T_2=O(\varepsilon^{-4}r^3  \log \delta^{-1} \lg r)$. 

        Next, we analyze the error guarantee. 
 According to Proposition~\ref{pro:solver} and union bound, with probability at least $1-2 \delta_4=1-\frac{\delta}{2C}$, we have $\abs{\tilde{x}_{i_j} - x_{i_j}} = \abs{(\tilde{z}_{j_1}^*-\tilde{z}_{j_2}^*)^2 - (z_{j_1}^*-z_{j_2}^*)^2} \leq (\abs{\tilde{z}_{j_1}^*+z_{j_1}^*}+\abs{\tilde{z}_{j_2}^*+z_{j_2}^*})(\abs{\tilde{z}_{j_1}^*-z_{j_1}^*}+\abs{\tilde{z}_{j_2}^*-z_{j_2}^*}) \leq 9\varepsilon_4 = \frac{\varepsilon}{2}$. 
 Then by union bound, with probability at least $1-C\cdot \frac{\delta}{2C}=\frac{\delta}{2}$, we have $\abs{\sum_{j=1}^{C} \tilde{x}_{i_j} - \sum_{j=1}^{C} x_{i_j}} \leq C\cdot \frac{\varepsilon}{2} = \frac{\varepsilon C}{2}$.
        According to Lemma~\ref{lem:sum}, with probability at least $1-\delta_2=1-\frac{\delta}{2}$, we have $\abs{\frac{n^2}{C}\sum_{j=1}^{C} x_{i_j} - \Polarization}\leq \varepsilon_2 n^2=\frac{\varepsilon n^2}{2}$. By union bound, with probability at least $1-\frac{\delta}{2}-\frac{\delta}{2}=1-\delta$, we have 
        \begin{align*}
           \abs{\frac{n^2}{C}\sum_{j=1}^{C} \tilde{x}_{i_j} - \Polarization}
            &\leq 
            \abs{\frac{n^2}{C}\sum_{j=1}^{C} \tilde{x}_{i_j} - \frac{n^2}{C}\sum_{j=1}^{C} x_{i_j}} + \abs{\frac{n^2}{C}\sum_{j=1}^{C} x_{i_j} - \Polarization} \\
            &\leq \frac{n^2}{C}\cdot \frac{\varepsilon C}{2} + \frac{\varepsilon n^2}{2}\\
            &= \varepsilon n^2.
        \end{align*}

	As $\Polarization = \frac{1}{2n} \sum_{u,v\in V} x_{u,v}$ we get an error for the polarization
	of $\varepsilon n$.

       \emph{Estimating $\Disagreement$:} Recall that $\DisCon=\Disagreement+\Controversy$, which implies that $\Disagreement=\DisCon-\Controversy$.
    	Using the results from above, we can compute approximations of $\DisCon$ and
    	$\Controversy$ with additive error $\pm\varepsilon n/2$ in time
		$O(\varepsilon^{-4}r^3  \log \delta^{-1} \lg r)$. Using the triangle
		inequality we get that the total error is bounded by $\pm\varepsilon n$.
\end{proof}

\subsection{Proof of Lemma~\ref{lem:estimate-s}}
	Recall that $s = (I+L)z^*$. Hence, we have that
	$s_u = (1+w_u)z_u^* - \sum_{(u,v)\in E} w_{uv} z_v^*$.
	We can compute the first term of this sum in time $O(1)$ using our query
	access to $z_u^*$ and $w_u$. Furthermore, by querying the values of $z_v^*$
	for all neighbors~$v$ of~$u$, we can compute the second term in time
	$O(d_u)$. 
	
	We can compute the second term
	in time~$O(w_u^2\varepsilon^{-2}\lg\delta^{-1})$ with additive
	error~$\varepsilon$ as follows.
	For convenience, we set $$S_u = \sum_{(u,v)\in E} z_v^* w_{uv}.$$
	Let $X_1$ be a random variable that takes value~$z_v^*$ with probability
	$w_{uv}/w_u$. Note that $\Exp{X_1} = S_u/w_u$.
	Furthermore, we have that
	\begin{align*}
		\Var{X_1} &= \Exp{X_1^2} - \Exp{X_1}^2
			\leq \sum_{(u,v)\in E} (z_v^*)^2 w_{uv}/w_u
			\leq S_u/w_u,
	\end{align*}
	where we used that $(z_v^*)^2 \leq z_v^*$ since $z_v^*\in[0,1]$.
	Now consider the random variable $Y_k = w_u \frac{1}{k} \sum_{i=1}^k X_i$,
	where the $X_i$ are i.i.d.\ copies of $X_1$. Then we have that
	$\Exp{Y_k}=S_u$.
	Furthermore,
	\begin{align*}
		\Var{Y_k}
		= w_u^2 \frac{1}{k} \Var{X_1}
		\leq \frac{w_u}{k} S_u.
	\end{align*}
	Now applying Chebyshev's inequality, we obtain that
	\begin{align*}
		\Prob{\abs{Y_k - S_u}\geq \varepsilon} 
		\leq \frac{\Var{Y_k}}{\varepsilon^2}
		\leq \frac{w_u S_u}{k \varepsilon^2}
		\leq 0.1,
	\end{align*}
	where we set $k = \frac{10 w_u^2}{\varepsilon^2}$ and used that $S_u\leq w_u$. 
	Applying the median trick, we obtain that with probability $1-\delta$ we
	return an estimate with additive error at most $1+\varepsilon$.

	We conclude that the second term can be computed in time
	$O(\min\{d_u,w_u^2\varepsilon^{-2}\lg\delta^{-1}\})$.

\subsection{Proof of Proposition~\ref{prop:estimate-S}}
 	We first observe that, by querying the values of $z_v^*$
	for all neighbors~$v$ of~$u$, we can compute $S_u$ in time
	$O(d_u)$. 

    Additionally, we can compute $S_u$
	in time~$O(d_u^{1/2}\varepsilon^{-1}\lg\delta^{-1})$ with
	$(1\pm\varepsilon)$-multiplicative error as follows. 
	 Let $m$ be a parameter that we will set below.
	Our algorithm starts by picking vertices $v_1, \dots, v_m$ independently at
	random (with repetitions) from all neighbors~$v$ of~$u$ with probabilities
	proportional to their weights, i.e., $w_{uv}/w_u$.
    Let $T$ be the set of sampled vertices, and for each $t\in T$ define $c_t$
	to be the number of times vertex $t$ is sampled, i.e.,
	$c_t = \abs{ \{ i \colon v_i = t \} }$.
    Define $Y_{ij}$ to be $z_v^*/w_{uv}$ if $v_i = v_j$ and $0$ otherwise. 
    We consider the estimator
    \begin{equation*}
        \tilde{S}_u=w_u^2 \cdot {\binom{m}{2}}^{-1} \cdot \sum_{t \in T} \frac{\binom{c_t}{2}\cdot z^*_t}{w_{ut}}.
    \end{equation*}
    
    We have
    \begin{equation*}
        \Exp{Y_{ij}}
        = \sum_{(u,v)\in E} \frac{z_v^*}{w_{uv}} \cdot \frac{w_{uv}}{w_u}\cdot \frac{w_{uv}}{w_u} = \sum_{(u,v)\in E} \frac{z_v^* w_{uv}}{w_u^2}=\frac{S_u}{w_u^2},
    \end{equation*}
    
    \begin{equation*}
        \Var{Y_{ij}} \leq \Exp{Y_{ij}^2}
        = \sum_{(u,v)\in E} \frac{(z_v^*)^2}{w_{uv}^2} \cdot \frac{w_{uv}}{w_u}\cdot \frac{w_{uv}}{w_u} 
        = \sum_{(u,v)\in E} \frac{(z_v^*)^2}{w_u^2}. 
    \end{equation*}
    
    Since
    \begin{equation*}
        \tilde{S}_u=w_u^2 \cdot {\binom{m}{2}}^{-1} \cdot \sum_{t \in T} \frac{\binom{c_t}{2}\cdot z^*_t}{w_{ut}}
        =w_u^2 \cdot {\binom{m}{2}}^{-1} \cdot \sum_{1\leq i < j \leq m} Y_{ij},
    \end{equation*}
    we have
    \begin{equation*}
        \Exp{\tilde{S}_u}=w_u^2 \cdot {\binom{m}{2}}^{-1} \cdot \sum_{1\leq i < j \leq m} \Exp{Y_{ij}}=w_u^2 \cdot \frac{S_u}{w_u^2} =S_u.
    \end{equation*}
    
    To bound $\Var{\tilde{S}_u}$, we need to bound $\Var{\sum_{1\leq i < j \leq m} Y_{ij}}$.
    Denote $\bar{Y}_{ij} \triangleq Y_{ij}-\Exp{Y_{ij}}$. We need to deal with
	the fact that the $Y_{ij}$'s are \emph{not} pairwise independent.
    Specifically, for four \emph{distinct} $i, j, i', j'$, indeed $Y_{ij}$ and $Y_{i'j'}$ are independent, and thus $\Exp{\bar{Y}_{ij} \bar{Y}_{i'j'}}=\Exp{\bar{Y}_{ij}}\cdot \Exp{\bar{Y}_{i'j'}}=0$;
    but for $i<j\neq k$, the random variables $Y_{ij}$ and $Y_{ik}$ are \emph{not} independent. 
    We have $\Exp{Y_{ij}Y_{ik}}=\sum_{(u,v)\in E} \frac{(z_v^*)^2}{w_{uv}^2} \cdot \frac{w_{uv}^3}{w_u^3}=\frac{1}{w_u^3} \sum_{(u,v)\in E} (z_v^*)^2 w_{uv}$.
    Therefore,
    \begin{align*}
            &\Var{\sum_{1\leq i < j \leq m} Y_{ij}}
			 = \Exp{\left(\sum_{1\leq i < j \leq m} \bar{Y}_{ij}\right)^2} \\
            &=\sum_{1\leq i < j \leq m} \Exp{\bar{Y}_{ij}^2}+2 \cdot \sum_{1\leq i < j \neq k \leq m} \Exp{\bar{Y}_{ij} \bar{Y}_{ik}} \\
            &\leq\sum_{1\leq i < j \leq m} \Exp{Y_{ij}^2}+2 \cdot \sum_{1 \leq i \leq m-1} \sum_{i+1 \leq j \neq k \leq m} \Exp{Y_{ij} Y_{ik}} \\
            &\leq \frac{m^2}{w_u^2} \cdot \sum_{(u,v)\in E} (z_v^*)^2 +  \frac{m^3}{w_u^3} \cdot \sum_{(u,v)\in E} (z_v^*)^2 w_{uv}.
    \end{align*}
    
    Furthermore, since $m^2 \leq 3\binom{m}{2}$ if $m \geq 3$, we have
        \begin{align*}
            \Var{\tilde{S}_u} &= w_u^4 \cdot {\binom{m}{2}}^{-2} \cdot \Var{\sum_{1\leq i < j \leq m} Y_{ij}}\\
            &\leq \frac{9w_u^2}{m^2} \cdot \sum_{(u,v)\in E} (z_v^*)^2 +  \frac{9w_u}{m} \cdot \sum_{(u,v)\in E} (z_v^*)^2 w_{uv}.
        \end{align*}
    
    By Chebyshev's inequality  , we have
    \begin{align*}
		&\Prob{\left|\tilde{S}_u-S_u\right| \geq \varepsilon S_u}
		 \leq \frac{\Var{\tilde{S}_u}}{\varepsilon^2 S_u^2} \\
		&\leq \frac{9(\sum_{(u,v)\in E} w_{uv})^2 \cdot (\sum_{(u,v)\in E} (z_v^*)^2)}{m^2\varepsilon^2 (\sum_{(u,v)\in E} z_v^* w_{uv})^2} \\
		&\quad +\frac{9(\sum_{(u,v)\in E} w_{uv}) (\sum_{(u,v)\in E} (z_v^*)^2 w_{uv})}{m\varepsilon^2 (\sum_{(u,v)\in E} z_v^* w_{uv})^2}\\
		&\leq \frac{9(\sum_{(u,v)\in E} w_{uv})^2 \cdot d_u}{m^2\varepsilon^2 (\sum_{(u,v)\in E} c w_{uv})^2} \\ 
		&\quad +\frac{9(\sum_{(u,v)\in E} w_{uv}) (\sum_{(u,v)\in E} z_v^* w_{uv})}{m\varepsilon^2 (\sum_{(u,v)\in E} z_v^* w_{uv})^2}\\
		&\leq \frac{9(\sum_{(u,v)\in E} w_{uv})^2 \cdot d_u}{m^2\varepsilon^2 c^2(\sum_{(u,v)\in E} w_{uv})^2} + 
		\frac{9\sum_{(u,v)\in E} w_{uv}}{m\varepsilon^2 c \sum_{(u,v)\in E}  w_{uv}}\\
		&\leq \frac{9d_u}{m^2\varepsilon^2 c^2} + \frac{9}{m\varepsilon^2 c}\\
		&\leq 0.1,
    \end{align*}
    where we set $m=O(d_u^{1/2}\varepsilon^{-1})$ and used that $(z_v^*)^2 \leq z_v^*$ and $\sum_{(u,v)\in E } (z_v^*)^2 \leq d_u$.
    Applying the median trick, we obtain that with probability $1-\delta$ we
	return an estimate with multiplicative error $1\pm\varepsilon$.

	We conclude that the total time we need to compute $S_u$ is at most
	$O(\min\{d_u,d_u^{1/2}\varepsilon^{-1}\lg\delta^{-1}\})$.

\subsection{Proof of Corollary~\ref{cor:estimate-s-new}}
    Recall that $s = (I+L)z^*$. Hence, we have that
	$s_u = (1+w_u)z_u^* - \sum_{(u,v)\in E} w_{uv} z_v^* = (1+w_u)z_u^* - S_u$.
	We can compute the first term of this sum in time $O(1)$ using our query
	access to $z_u^*$ and $w_u$. Furthermore, using to
	Proposition~\ref{prop:estimate-S}, with probability at least $1-\delta$ we
	can estimate the second term in
	time~$O(\min\{d_u,d_u^{1/2}\varepsilon^{-1}\lg\delta^{-1}\})$ with
	multiplicative error $(1\pm\varepsilon)$, i.e.,
	$\left|\tilde{S}_u-S_u\right| \leq \varepsilon S_u$. 
    We set $\tilde{s}_u=(1+w_u)z_u^* - \tilde{S}_u$.
    
    If $S_u \leq 1$, then it follows immediately from Proposition~\ref{prop:estimate-S} that $\left|\tilde{s}_u-s_u\right|=\left|\tilde{S}_u-S_u\right| \leq \varepsilon S_u \leq \varepsilon$.

    If $S_u \leq \frac{(1+w_u)z_u^*}{2}$, then according to Proposition~\ref{prop:estimate-S}, we have
    \begin{align*}
        \tilde{s}_u-s_u &= \tilde{S}_u-S_u \\
        &\leq (1+\varepsilon)S_u - S_u\\
        &= \varepsilon S_u \\
        &\leq \varepsilon \cdot \frac{(1+w_u)z_u^*}{2} \\
        &\leq \varepsilon ((1+w_u)z_u^*+S_u) \\
        &= \varepsilon s_u,
    \end{align*}
    and 
    \begin{align*}
        \tilde{s}_u-s_u &= \tilde{S}_u-S_u \\
        &\geq (1-\varepsilon)S_u - S_u\\
        &= -\varepsilon S_u \\
        &\geq -\varepsilon \cdot \frac{(1+w_u)z_u^*}{2} \\
        &\geq -\varepsilon ((1+w_u)z_u^*+S_u) \\
        &= -\varepsilon s_u.
    \end{align*}
    Therefore, $\left|\tilde{s}_u-s_u\right| \leq \varepsilon s_u$.

\subsection{Estimating the Measures Given Oracle Access to Expressed Opinions~$z_u^*$}
\label{app:measures-oracle-expressed}
In this section, we prove the results of the last column of
Table~\ref{tab:experiments}.

\begin{lemma}
\label{lem:estimation-oracle-expressed-opinions}
	Let $\varepsilon,\delta\in(0,1)$. Then with probability at least
	$1-\delta$:
	\begin{itemize}
		\item We can return an estimate of $\avgExpressed$ with additive error
			$\pm\varepsilon$ in time $O(\varepsilon^{-2}\lg\delta^{-1})$.
		\item We can return an estimate of $\SumOpinions, \Controversy$ and $\Polarization$ with additive error
			$\pm\varepsilon n$ in time $O(\varepsilon^{-2}\lg\delta^{-1})$.
	\end{itemize}
\end{lemma}
\begin{proof}
	\emph{Estimating $\avgExpressed$, $\SumOpinions$ and $\Controversy$:} These
	three claims follow directly from Lemma~\ref{lem:sum}.

	\emph{Estimating the polarization~$\Polarization$:} Recall that
	$\Polarization = \sum_{u\in V} (z_u^* - \avgExpressed)^2$, where
	$\avgExpressed = \frac{1}{n} \sum_{u\in V} z_u^*$. Now we use the well-known
	equality that $\sum_i \sum_{j>i} (a_i - a_j)^2 = n \sum_{i} (a_i - c)^2$,
	where $c=\frac{1}{n} \sum_{i} a_i$. This gives us that
	\begin{equation*}
		\Polarization
		= \sum_{u\in V} (z_u^* - \avgExpressed)^2 
		= \frac{1}{2n} \sum_{u,v\in V} (z_u^* - z_v^*)^2.
	\end{equation*}
	Hence, we can apply Lemma~\ref{lem:sum} with a vector $x$ of length $n^2$
	which has entries $x_{u,v} = (z_u^* - z_v^*)^2 \in [0,1]$ for $u,v\in V$.
	Thus the lemma gives us an estimate~$\tilde{\Sigma}$ of
	$\Sigma=\sum_{u,v \in V} (z_u^* - z_v^*)^2$ with additive error
	$\abs{\tilde{\Sigma} - \Sigma} \leq \varepsilon n^2$.
	As $\Polarization = \frac{1}{2n} \Sigma$ we get an error for the polarization
	of $\varepsilon n$. 
\end{proof}

\begin{lemma}
\label{lem:improved-time}
   Let $\varepsilon,\delta\in(0,1)$. Then with probability at least
	$1-\delta$, we can return an estimate of $\norm{s}_2^2, \Internal, \Disagreement$ and $\DisCon$ with additive error
			$\pm\varepsilon n$ in time $O(\varepsilon^{-2} \bar{d} \lg^2 \delta^{-1})$.
\end{lemma}
\begin{proof}
  In the following, we use the estimation of $\norm{s}_2^2$ as an example to
  illustrate. The estimation of $\Internal$ and $\DisCon$ works
  similarly. The disagreement $\Disagreement$ can be estimated using
  Lemma~\ref{lem:conservation-law} and our results for $\norm{s}_2^2$,
  $\Internal$ and $\Controversy$.
   
   Recall that
	$\norm{s}_2^2=\sum_{u\in V} s_u^2$. 
       We set 
        $\varepsilon_1 = \frac{\varepsilon}{6}$, $\delta_1 = \frac{\delta}{2}$, 
	$\varepsilon_2 = \frac{\varepsilon}{2}$, $\delta_2 = \frac{\delta}{2}$ and $C = \varepsilon_2^{-2} \log \delta_2^{-1} = O(\varepsilon^{-2} \log \delta^{-1})$.
	According to Lemma~\ref{lem:opinion-sampler}, in time $O(C \bar{d} \lg \delta^{-1})$, we can sample a (multi-)set of vertices~$S=\{i_1, i_2, \dots, i_C\}$ uniformly at random from~$V$ and obtain
	estimated innate opinions $\tilde{s}_u$ for all $u\in S$ such that with
	probability $1-\delta_1$ it holds that 
	$\abs{s_u - \tilde{s}_u} \leq \varepsilon_1$ for all $u\in S$. We return $\frac{n}{C} \sum_{u\in S} \tilde{s}_{u}^2$.

The running time is 
$O(C \bar{d} \lg \delta^{-1}) = O(\varepsilon^{-2} \bar{d} \lg^2 \delta^{-1})$. Now we analyze the error guarantee.
According to Lemma~\ref{lem:opinion-sampler}, for all $u\in S$, with probability
at least $1-\delta_1=1-\frac{\delta}{2}$, we have $\abs{\tilde{s}_u-s_u} \leq
\varepsilon_1 = \frac{\varepsilon}{6}$. Therefore, $\abs{\tilde{s}_u^2-s_u^2}
\leq \abs{\tilde{s}_u+s_u} \cdot \abs{\tilde{s}_u-s_u} \leq 3\varepsilon_1 =
\frac{\varepsilon}{2}$, since $\tilde{s}_u+s_u \leq (1 + \varepsilon_1) + 1 \leq
3$ and $\abs{\tilde{s}_u-s_u}\leq \varepsilon_1$ by our assumption above.
	Then we have $\abs{\sum_{u\in S} \tilde{s}_u^2 - \sum_{u\in S} s_u^2} \leq \abs{S}\cdot \frac{\varepsilon}{2} = \frac{\varepsilon C}{2}$.
        According to Lemma~\ref{lem:sum}, with probability at least $1-\delta_2 = 1-\frac{\delta}{2}$, we have $\abs{\frac{n}{C}\sum_{u\in S} s_u^2 - \norm{s}_2^2}\leq \varepsilon_2 n = \frac{\varepsilon n}{2}$. By union bound, with probability at least $1-\frac{\delta}{2}-\frac{\delta}{2}=1-\delta$, we have 
        \begin{align*}
            \abs{\frac{n}{C}\sum_{u \in S} \tilde{s}_u^2 - \norm{s}_2^2}
            &\leq 
            \abs{\frac{n}{C}\sum_{u \in S} \tilde{s}_u^2 - \frac{n}{C}\sum_{u \in S} s_u^2} + \abs{\frac{n}{C}\sum_{u \in S} s_u^2 - \norm{s}_2^2}\\
            &\leq \frac{n}{C}\cdot \frac{\varepsilon C}{2} + \frac{\varepsilon n}{2}\\
            &= \varepsilon n.
			\qedhere
        \end{align*}
\end{proof}

\section{Additional Experiments}
\label{app:exp}
In this section, we present additional experimental results and elaborate on the
technical details of our experiment setup.

\subsection{Datasets and Parameter Choices}
First, we note that we obtained the bounds on $\kappa(\tilde{S})$ in
Table~\ref{tab:experiments}, where $\tilde{S} = (I+D)^{-1/2} (I+L)
(I+D)^{-1/2}$, using power iteration with $100$~iterations.

Next, we provide details on the opinion distributions that we used. We already
mentioned the uniform distribution in the main text. Additionally, we sampled
opinions from an exponential distribution and then rescaled the values we obtain
so that all opinions are in the interval~$[0,1]$; this was done exactly as
in~\cite{xu2021fast}. Since for the first two distributions, the opinions do not
depend on the graph structure, we also compute an approximation~$v$ of the
second eigenvector of~$L$ using power iteration with $100$~iterations. Given
this approximation, we rescale all entries in~$v$ such that they are in the
interval~$[0,1]$ by setting $v_i=\frac{v_i - \min(v)}{\max(v)-\min(v)}$.
Intuitively, this vector takes into account the community structure of the graph
and thus we obtain opinions that depend on the graph structure. Interestingly,
we find that on the datasets we consider, this distribution is relatively close
to the uniform distribution.

Next, let us briefly justify our parameter choices in
Tables~\ref{tab:experiments-oracle-innate-uniform}
and~\ref{tab:experiments-oracle-expressed-uniform}.
We picked \numprint{4000}~random walks since Figures~\ref{fig:estimating-opinions-oracle-innate}(\subref{fig:pokec-walks})
and~\ref{fig:estimating-opinions-oracle-innate}(\subref{fig:livejournal-walks}) suggest that by going from \numprint{4000} to
\numprint{6000} walks, the error and standard deviations do not reduce
significantly anymore, whereas going from \numprint{2000} to
\numprint{4000}~random walks reduces the standard deviation of the error
significantly. Additionally, the number of steps was determined by observing
that (see Figures~\ref{fig:estimating-opinions-oracle-innate-time}(\subref{fig:pokec-steps-time})
and~\ref{fig:estimating-opinions-oracle-innate-time}(\subref{fig:livejournal-steps-time})) using \numprint{600} steps costs barely
more running time than using \numprint{400} steps, but does provide slightly less
error (see Figures~\ref{fig:estimating-opinions-oracle-innate}(\subref{fig:pokec-steps})
and~\ref{fig:estimating-opinions-oracle-innate}(\subref{fig:livejournal-steps})). We considered \numprint{10000} nodes to obtain running
times similar to the baselines. We note that if we use less nodes (e.g.,
		\numprint{1000} nodes) then we still obtain similar average errors when computing
polarization, controversy, etc. (the error increases by approx. 1\%, but we
		obtain higher standard deviations).

\subsection{Further Analysis Given Oracle Access to $s_u$}
Next, we provide further experimental results when we have query access to
the innate opinions $s_u$.

\textbf{Running time analysis.}
In Figure~\ref{fig:estimating-opinions-oracle-innate-time}, we present the
running time of Algorithm~\ref{alg:random-walks} on the \Pokec and \LiveJournal
datasets with uniformly distributed innate opinions.
We observe that the algorithm's running time scales linearly in the number of
random walks, as well as in the number of sampled vertices for which the
opinions shall be estimated. We also observe that after setting the number of
random walk steps to 400, the running time stops increasing even for larger
numbers of random walk steps.  This behavior is explained by the timeout of the
random walks, which at vertex~$v$ terminate with probability
$1/2 - w_v / (2(1+w_v))$ (see Section~\ref{sec:proof:pro:solver} for details).
In other words, the probability that the random walks perform more than
400~steps without terminating is very small and therefore the running time stops
increasing.

We note that here we only report running time results for uniformly distributed
innate opinions. We do this for conciseness, since the results using the other
two opinion distributions are almost identical.

\begin{figure*}
	\centering
	\begin{subfigure}{0.24\textwidth}
		\includegraphics[width=\textwidth]{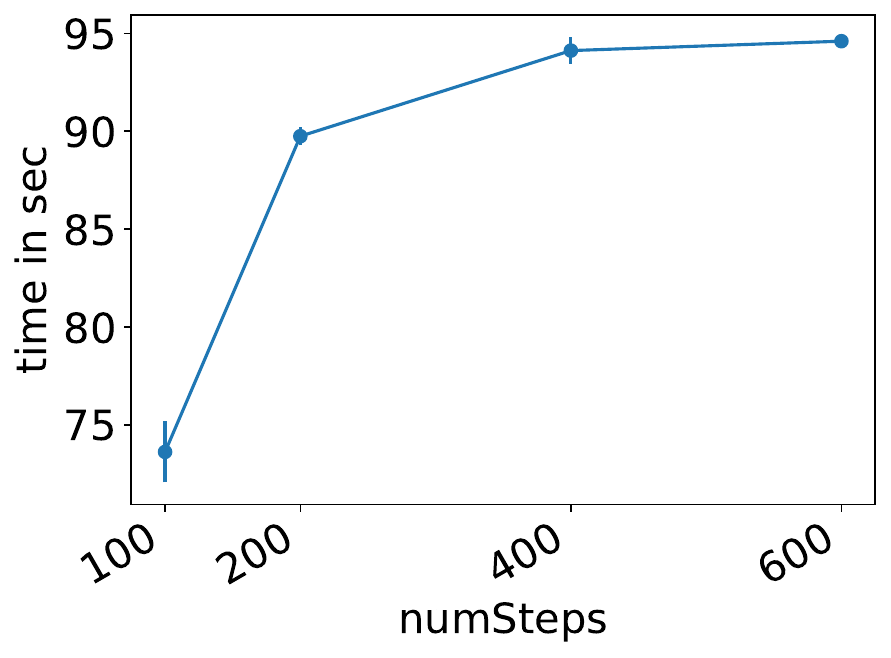}
		\caption{\Pokec, vary \#steps}
		\label{fig:pokec-steps-time}
	\end{subfigure}
	\hfill
	\begin{subfigure}{0.24\textwidth}
		\includegraphics[width=\textwidth]{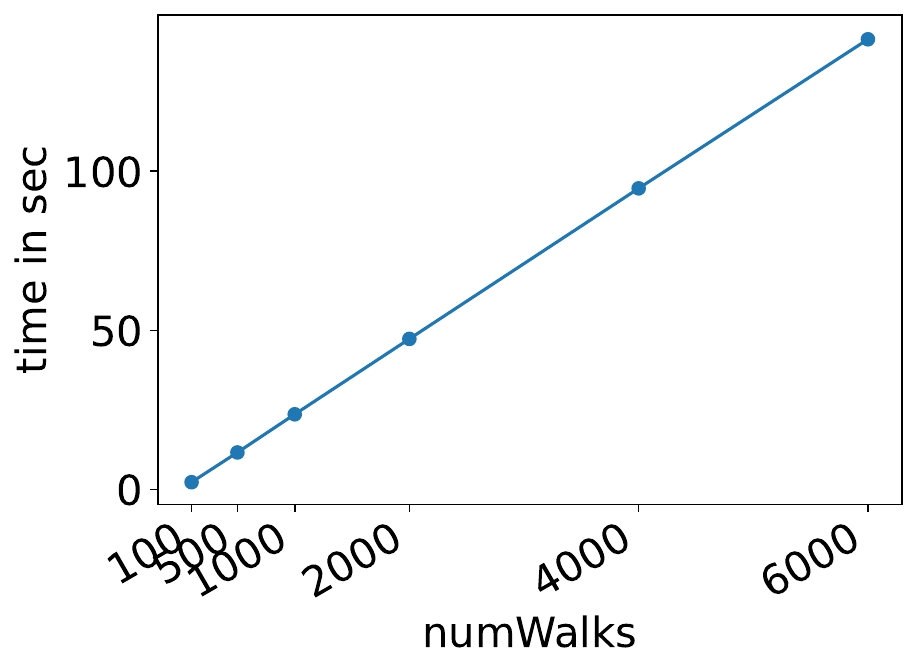}
		\caption{\Pokec, vary \#walks.}
		\label{fig:pokec-walks-time}
	\end{subfigure}
	\hfill
	\begin{subfigure}{0.24\textwidth}
		\includegraphics[width=\textwidth]{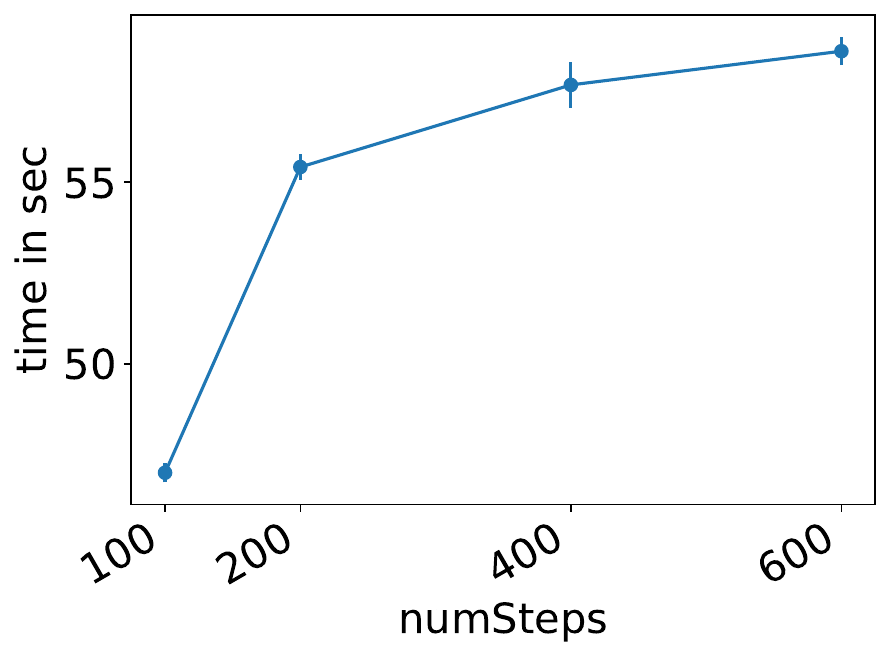}
		\caption{\LiveJournal, vary \#steps}
		\label{fig:livejournal-steps-time}
	\end{subfigure}
	\hfill
	\begin{subfigure}{0.24\textwidth}
		\includegraphics[width=\textwidth]{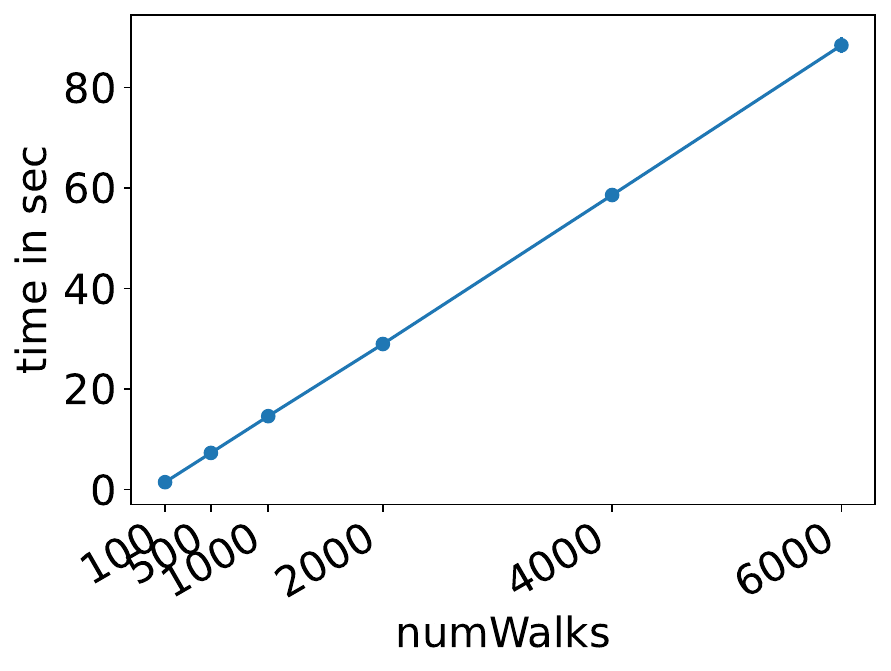}
		\caption{\LiveJournal, vary \#walks}
		\label{fig:livejournal-numsamples-time}
	\end{subfigure}
	\caption{Running time of Algorithm~\ref{alg:random-walks} for
	   	estimating expressed opinions~$z_u^*$ using an oracle for innate
		opinions~$s_u$.
		When not mentioned otherwise, we sampled \numprint{10000}~vertices and for each of
		them we performed \numprint{4000}~random walks with \numprint{600}~steps.
		We report means and standard deviations across 10~experiments.
		Innate opinions were generated using the uniform distribution.
	}
	\label{fig:estimating-opinions-oracle-innate-time}
\end{figure*}

\textbf{Additional error analysis.}
Next, we present additional error analysis with different opinion distributions
on the datasets that we consider.

First, we consider estimating the expressed opinions~$z_u^*$ using
Algorithm~\ref{alg:random-walks} and an oracle for innate opinions~$s_u$.  We
present the results using innate opinions generated from the second eigenvector
of the Laplacian in Table~\ref{tab:experiments-oracle-innate-eigenvalue} and the
results using a rescaled exponential distribution in
Table~\ref{tab:experiments-oracle-innate-exponential}.  We observe that for
approximating the measures, our results are highly similar to what we present
in the main text for uniformly distributed innate opinions. That is, for all
measures except disagreement we can compute estimates with relative error at
most 6\%. Interestingly, we observe that for exponentially distributed innate
opinions the average absolute error for estimating $z_u^*$ is only $\pm0.003$
(rather than $\pm0.01$ for the other two distributions), but this does not lead
to significantly lower relative error when approximating the measures.

\begin{table*}[t!]
\caption{Errors for different datasets given an oracle for innate opinions; we
	report means and standard deviations (in parentheses) across 10~experiments. We ran
	Algorithm~\ref{alg:random-walks} with \numprint{600}~steps and
	\numprint{4000}~random walks; we estimated the opinions of
	\numprint{10000}~random vertices.  Innate opinions were generated using the
	the second eigenvector of the Laplacian.}
\label{tab:experiments-oracle-innate-eigenvalue}
\centering
\begin{adjustbox}{max width=\textwidth}
\begin{tabular}{c cccccccc}
\toprule
\textbf{Dataset} & \textbf{Absolute Error} & \multicolumn{7}{c}{\textbf{Relative Error in \%}} \\
\cmidrule(lr){3-9} 
& $z_u^*$ & $\SumOpinions$ & $\Polarization$ & $\Disagreement$ &
$\Internal$ & $\Controversy$ & $\DisCon$ & $\norm{s}_2^2$ \\
\midrule
\input{tables/measures_Eigenvalue_givenS_error_numStepsSamples600_numWalksRepetitions4000_numSampledVertices10000.tex}
\bottomrule
\end{tabular}
\end{adjustbox}
\end{table*}

\begin{table*}[t!]
\caption{Errors for different datasets given an oracle for innate opinions; we
	report means and standard deviations (in parentheses) across 10~experiments. We ran
	Algorithm~\ref{alg:random-walks} with \numprint{600}~steps and
	\numprint{4000}~random walks; we estimated the opinions of
	\numprint{10000}~random vertices.  Innate opinions were generated using the
	exponential distribution.}
\label{tab:experiments-oracle-innate-exponential}
\centering
\begin{adjustbox}{max width=\textwidth}
\begin{tabular}{c cccccccc}
\toprule
\textbf{Dataset} & \textbf{Absolute Error} & \multicolumn{7}{c}{\textbf{Relative Error in \%}} \\
\cmidrule(lr){3-9} 
& $z_u^*$ & $\SumOpinions$ & $\Polarization$ & $\Disagreement$ &
$\Internal$ & $\Controversy$ & $\DisCon$ & $\norm{s}_2^2$ \\
\midrule
\input{tables/measures_Exponential_givenS_error_numStepsSamples600_numWalksRepetitions4000_numSampledVertices10000.tex}
\bottomrule
\end{tabular}
\end{adjustbox}
\end{table*}

\textbf{Impact of vertex degrees.}
Next, our goal is to understand how the degree of a vertex~$u$ impacts the
performace of Algorithm~\ref{alg:random-walks} when estimating $z_u^*$. Our
experiment setup to obtain this understanding is as follows.
For each dataset, we sort the nodes by their degrees. Then we create
20~buckets, where the first bucket contains the 5\% of nodes of lowest degree,
the second bucket contains the nodes with the 5\% to 10\% lowest degrees, etc.
For each bucket, we randomly pick 500~nodes and estimate their opinions.
We report the error and running time for each bucket (averaged over 10
runs and we will also report standard deviations); we normalize the errors and
standard deviations for each bucket, by dividing by the average error and the
standard deviation across the experiments for \emph{all} buckets. This will
allow us to understand whether on some buckets we obtain larger/smaller errors
than on others.

We report our experiments on the datasets with uniform opinions.  In our
experiments, we used Algorithm~\ref{alg:random-walks} with
\numprint{4000}~random walks with \numprint{600}~steps.

Regarding the errors, we present the results on the four largest datasets in
Figure~\ref{fig:error-buckets} (for the smaller datasets, the results were
analogous). Our results show that the normalized errors and standard deviations
are both $1$ across all datasets and buckets with only minor fluctuations. In
other words, the errors and standard deviations are the same on all buckets and
there is essentially no difference.
We note that this is consistent with the results of
Table~\ref{tab:experiments-oracle-innate-uniform}, where the errors
and standard deviations are highly similar across all datasets; our new results
show that we can say the same about the buckets for each dataset. The only
slight outlier is LiveJournal, where (only) the final bucket has an average
error that is approximately 40\% higher than the other buckets and it has much
higher standard deviation.
Overall, we conclude that the error is almost not impacted by the degree of the
vertices.

\begin{figure*}
	\centering
	\begin{subfigure}{0.24\textwidth}
		\includegraphics[width=\textwidth]{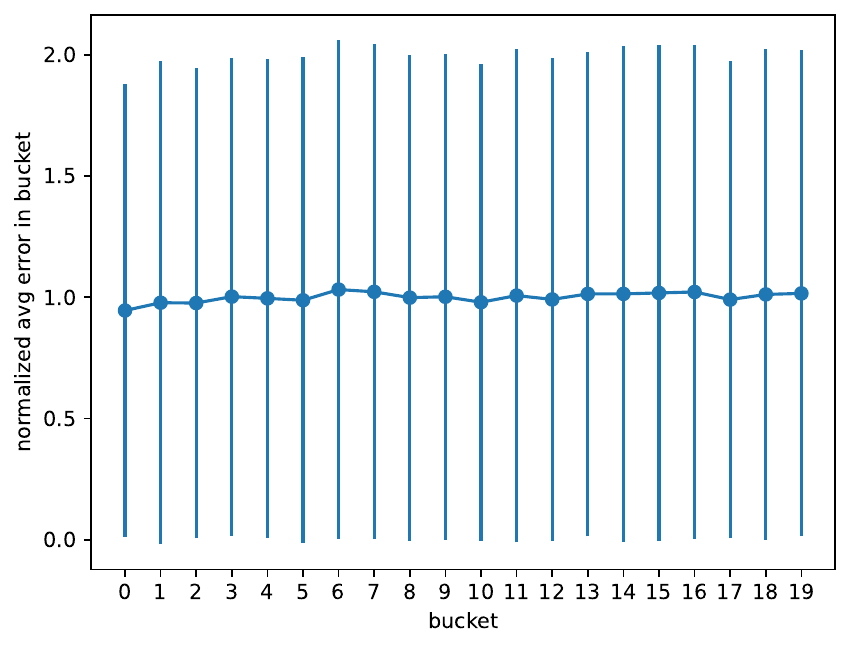}
		\caption{\Pokec}
		\label{fig:pokec-error-buckets}
	\end{subfigure}
	\hfill
	\begin{subfigure}{0.24\textwidth}
		\includegraphics[width=\textwidth]{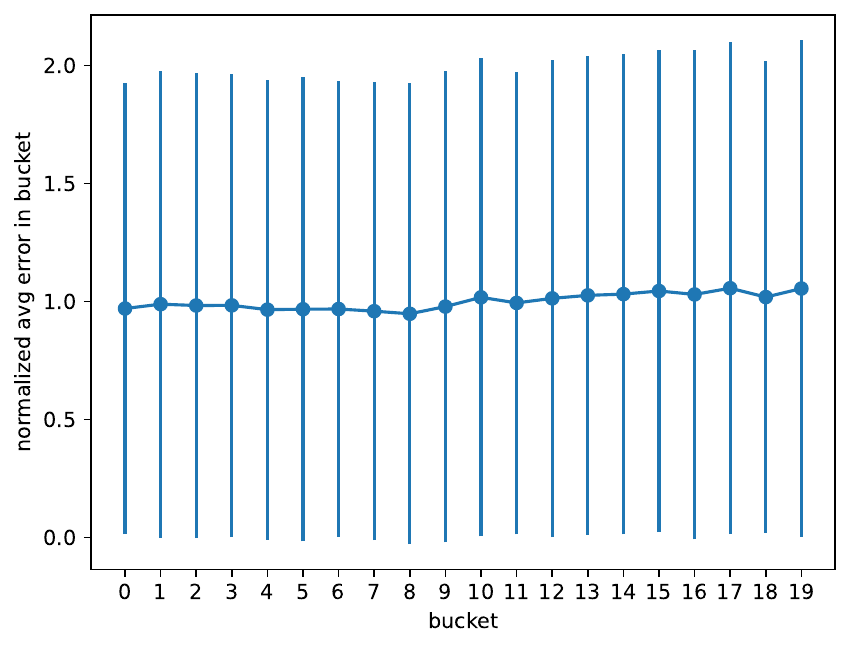}
		\caption{\Flickr}
		\label{fig:flickr-error-buckets}
	\end{subfigure}
	\hfill
	\begin{subfigure}{0.24\textwidth}
		\includegraphics[width=\textwidth]{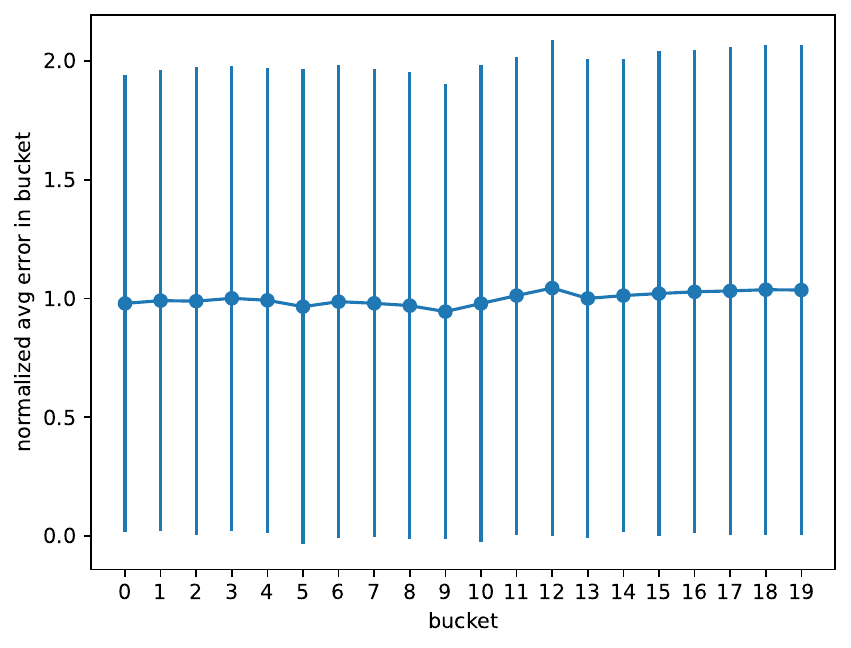}
		\caption{\YouTube}
		\label{fig:youtube-error-buckets}
	\end{subfigure}
	\hfill
	\begin{subfigure}{0.24\textwidth}
		\includegraphics[width=\textwidth]{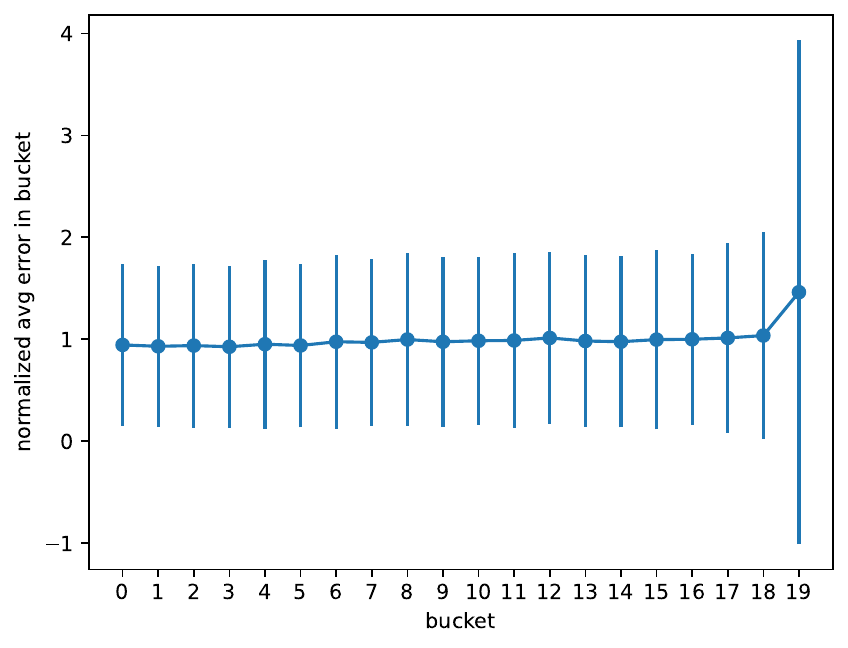}
		\caption{\LiveJournal}
		\label{fig:livejournal-error-buckets}
	\end{subfigure}
	\caption{Normalized error of Algorithm~\ref{alg:random-walks} for
	   	estimating expressed opinions~$z_u^*$ using an oracle for innate
		opinions~$s_u$.
		We sorted the vertices by degrees (from low to high) and partitioned
		them into 20~equally sized buckets.
		We used Algorithm~\ref{alg:random-walks} with \numprint{4000}~random
		walks with \numprint{600}~steps.
		We report means and standard deviations across 10~experiments.
		Innate opinions were generated using the uniform distribution.
	}
	\label{fig:error-buckets}
\end{figure*}

Regarding the running time, we present the results on our four largest datasets
in Figure~\ref{fig:time-buckets}. We find that our query procedure takes more
time for high-degree nodes. While this might sound surprising at first, it is
actually not: we note that our random walks have a timeout, i.e., at each step
they terminate with probability $1/(2(1+w_v))$, where $w_v$ is the degree of the
current node. Hence, when starting in a high-degree node, $w_v$ is large and we
have a much smaller probability of stopping ``early''. This explains our
phenomenon.  Indeed, we can see that for denser datasets (higher average degree,
see Tables~\ref{tab:experiments} and~\ref{tab:running-times}), the query
procedure gets slower. When comparing the query times across the different
buckets, we find that typically they differ by a factor of at most 2 to 3 across
the different datasets. It is not clear to us why the running times of some of
the medium-range buckets are sometimes the fastest.

\begin{figure*}
	\centering
	\begin{subfigure}{0.24\textwidth}
		\includegraphics[width=\textwidth]{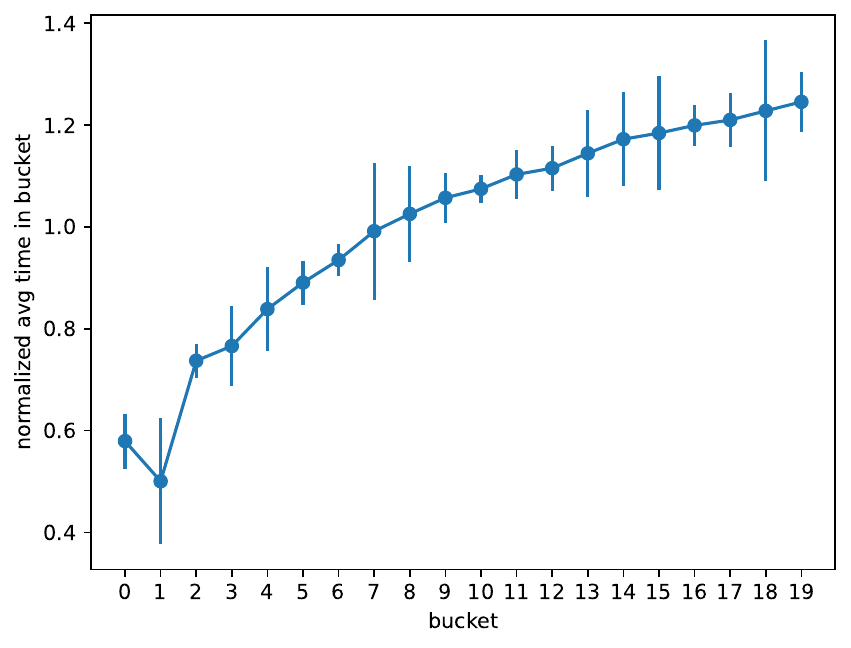}
		\caption{\Pokec}
		\label{fig:pokec-time-buckets}
	\end{subfigure}
	\hfill
	\begin{subfigure}{0.24\textwidth}
		\includegraphics[width=\textwidth]{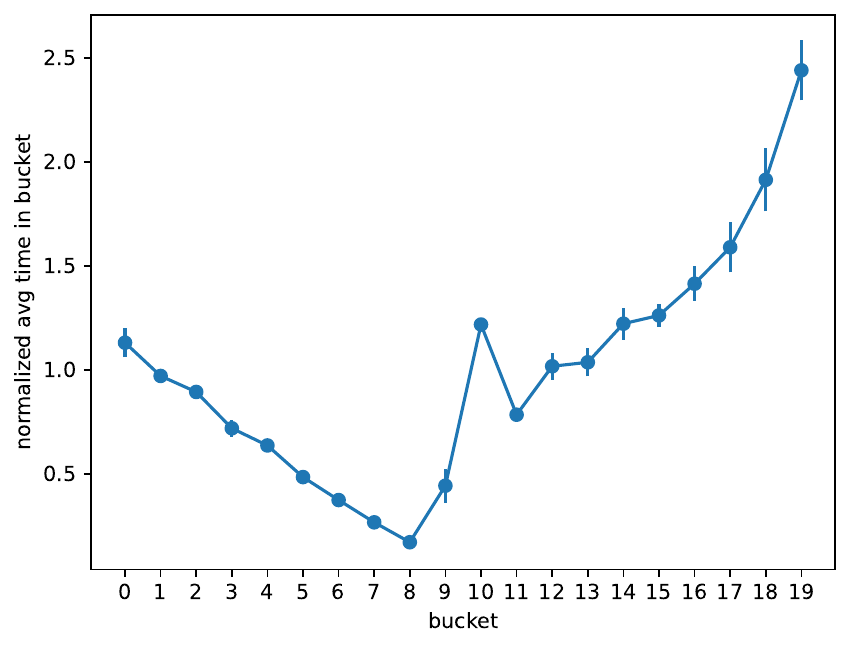}
		\caption{\Flickr}
		\label{fig:flickr-time-buckets}
	\end{subfigure}
	\hfill
	\begin{subfigure}{0.24\textwidth}
		\includegraphics[width=\textwidth]{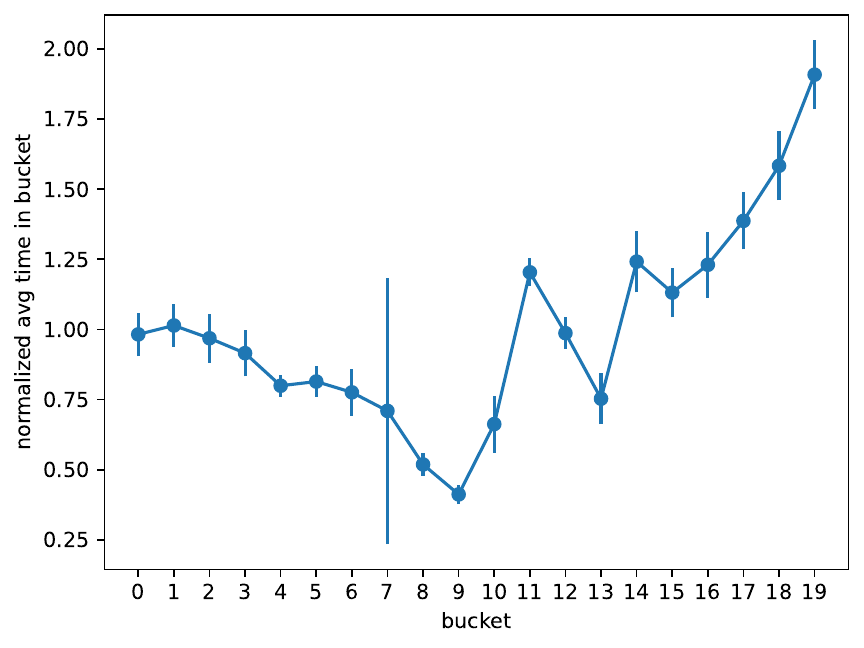}
		\caption{\YouTube}
		\label{fig:youtube-time-buckets}
	\end{subfigure}
	\hfill
	\begin{subfigure}{0.24\textwidth}
		\includegraphics[width=\textwidth]{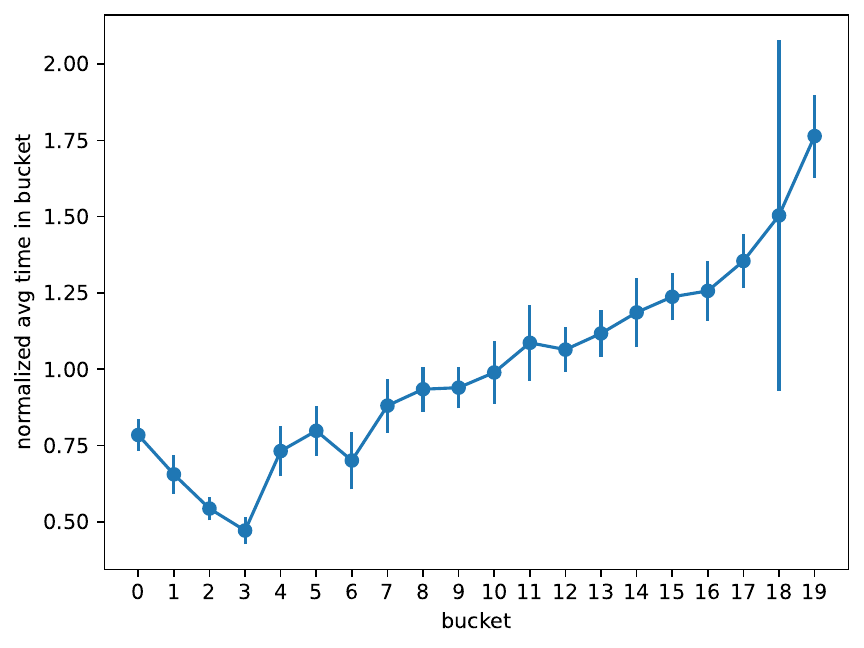}
		\caption{\LiveJournal}
		\label{fig:livejournal-time-buckets}
	\end{subfigure}
	\caption{Normalized running time of Algorithm~\ref{alg:random-walks} for
	   	estimating expressed opinions~$z_u^*$ using an oracle for innate
		opinions~$s_u$.
		We sorted the vertices by degrees (from low to high) and partitioned
		them into 20~equally sized buckets.
		We used Algorithm~\ref{alg:random-walks} with \numprint{4000}~random
		walks with \numprint{600}~steps.
		We report means and standard deviations across 10~experiments.
		Innate opinions were generated using the uniform distribution.
	}
	\label{fig:time-buckets}
\end{figure*}

\subsection{Further Analysis Given Oracle Access to $z_u^*$}
Next, we consider estimating the innate opinions~$s_u$ using the sampling
scheme from Lemma~\ref{lem:estimate-s} and an oracle for expressed
opinions~$z_u^*$.  Again, we present additional error analysis with different
opinion distributions on the datasets that we consider.
We present the results using innate opinions generated from
the second eigenvector of the Laplacian in
Table~\ref{tab:experiments-oracle-expressed-eigenvalue} and the results using a
rescaled exponential distribution in
Table~\ref{tab:experiments-oracle-expressed-exponential}.  We again observe that
overall the results are similar to what we reported in the main text for
uniformly distributed innate opinions. The main difference is that for
exponentially distributed innate opinions, the relative error for internal
conflict~$\Internal$ and, to a lesser extent, for polarization~$\Polarization$
is higher. We explain this by the fact that for a highly skewed distribution
like the exponential distribution, a small number of vertices make up for a
large fraction of the measures' values. Therefore, sampling-based schemes like
ours perform worse and require estimating more vertex opinions (compared innate
opinion distributed based on less skewed distributions).

\begin{table*}[t!]
\caption{Errors for different datasets given an oracle for expressed opinions;
	we report means and standard deviations (in parentheses) across 10~experiments. We ran our
	algorithm with threshold~400 and 5~repetitions; we estimated the opinions of
	\numprint{10000}~random vertices.  Innate opinions were generated using the
	the second eigenvector of the Laplacian.}
\label{tab:experiments-oracle-expressed-eigenvalue}
\centering
\begin{adjustbox}{max width=\textwidth}
\begin{tabular}{c cccccccc}
\toprule
\textbf{Dataset} & \textbf{Absolute Error} & \multicolumn{7}{c}{\textbf{Relative Error in \%}} \\
\cmidrule(lr){3-9} 
& $s_u$ & $\SumOpinions$ & $\Polarization$ & $\Disagreement$ &
$\Internal$ & $\Controversy$ & $\DisCon$ & $\norm{s}_2^2$ \\
\midrule
\input{tables/measures_Eigenvalue_givenZ_error_numStepsSamples400_numWalksRepetitions5_numSampledVertices10000.tex}
\bottomrule
\end{tabular}
\end{adjustbox}
\end{table*}

\begin{table*}[t!]
\caption{Errors for different datasets given an oracle for expressed opinions;
	we report means and standard deviations (in parentheses) across 10~experiments. We ran our
	algorithm with threshold~400 and 5~repetitions; we estimated the opinions of
	\numprint{10000}~random vertices.  Innate opinions were generated using the
	exponential distribution.}
\label{tab:experiments-oracle-expressed-exponential}
\centering
\begin{adjustbox}{max width=\textwidth}
\begin{tabular}{c cccccccc}
\toprule
\textbf{Dataset} & \textbf{Absolute Error} & \multicolumn{7}{c}{\textbf{Relative Error in \%}} \\ 
\cmidrule(lr){3-9} 
& $s_u$ & $\SumOpinions$ & $\Polarization$ & $\Disagreement$ &
$\Internal$ & $\Controversy$ & $\DisCon$ & $\norm{s}_2^2$ \\
\midrule
\input{tables/measures_Exponential_givenZ_error_numStepsSamples400_numWalksRepetitions5_numSampledVertices10000.tex}
\bottomrule
\end{tabular}
\end{adjustbox}
\end{table*}

\subsection{Analysis of the PageRank-Style Update Rule}

We also evaluate the error of the PageRank-style update
rule~\cite{friedkin2014two,proskurnikov2016pagerank} from
Proposition~\ref{prop:page-rank-equation}. To this end, we implement an
algorithm which initializes a vector in which all entries are set to the average
of the node opinions~$\frac{1}{n} \sum_u s_u$ and then we apply the update rule
from the proposition for 50~iterations.
We compare the error of this iterative algorithm against the solution computed
by the algorithm of Xu et al.~\cite{xu2021fast}. We report the $\ell_2$-norm of
the error, as well as the differences of the iterates
$\norm{z^{(t)} - z^{(t-1)}}_2$; we note that in the figure, the error is not
normalized by the number of vertices~$n$.

We report our findings in Figure~\ref{fig:error-pr}, showing that the error
decays exponentially.

\begin{figure*}
	\centering
	\begin{subfigure}{0.24\textwidth}
		\includegraphics[width=\textwidth]{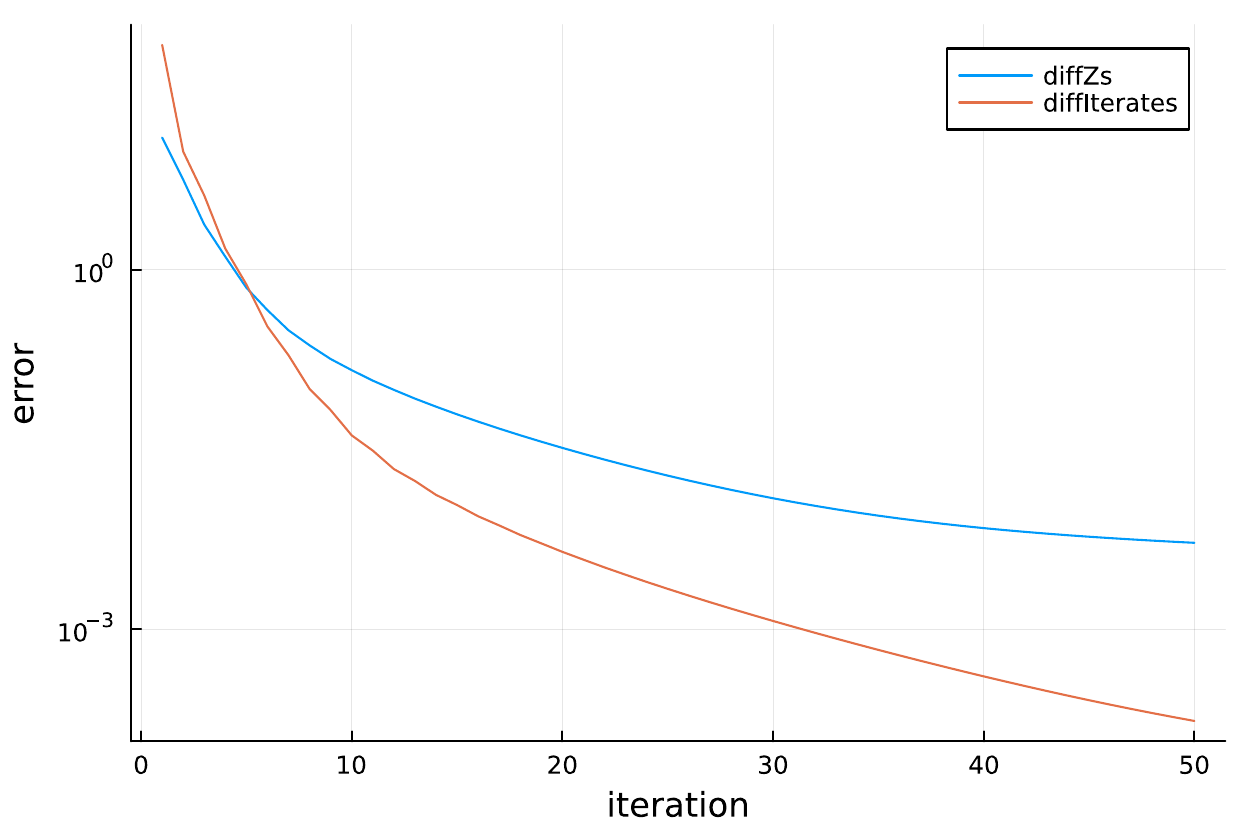}
		\caption{\Pokec}
		\label{fig:pokec-error-pr}
	\end{subfigure}
	\hfill
	\begin{subfigure}{0.24\textwidth}
		\includegraphics[width=\textwidth]{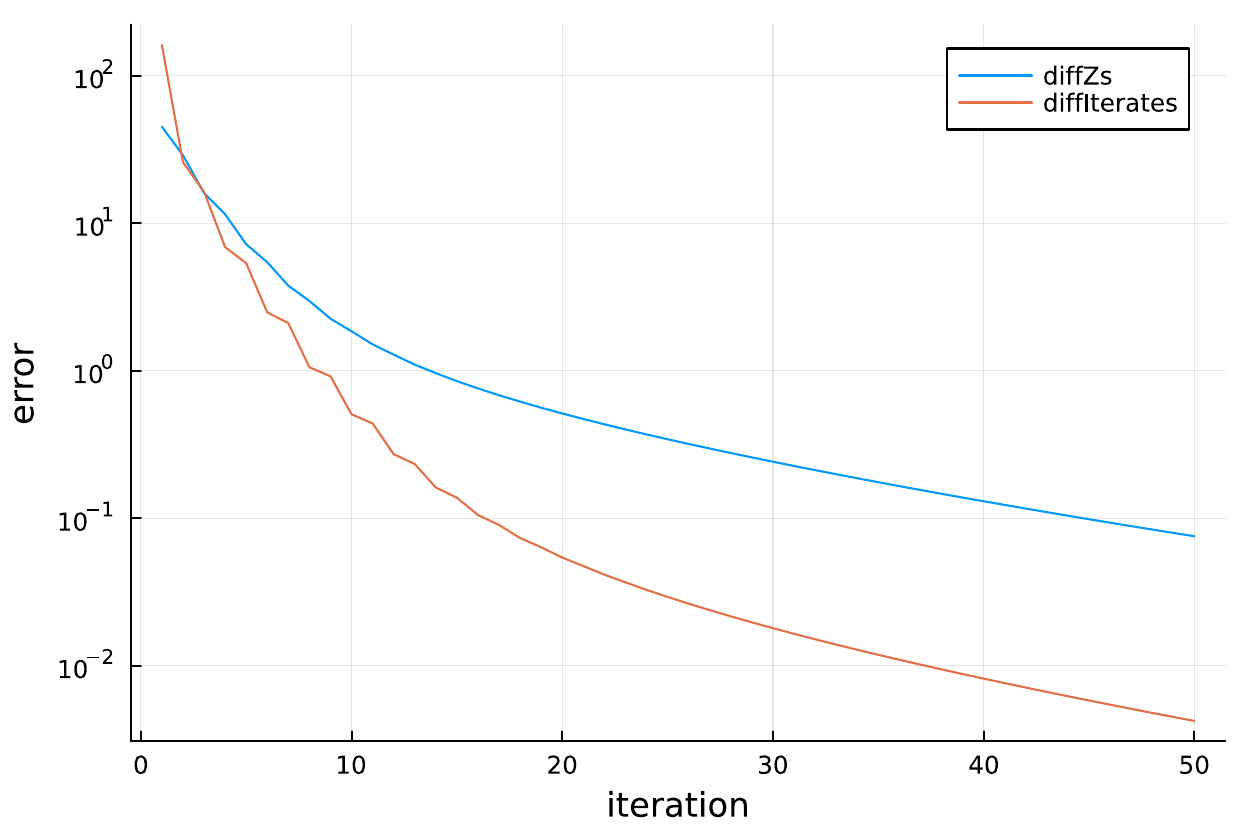}
		\caption{\Flickr}
		\label{fig:flickr-error-pr}
	\end{subfigure}
	\hfill
	\begin{subfigure}{0.24\textwidth}
		\includegraphics[width=\textwidth]{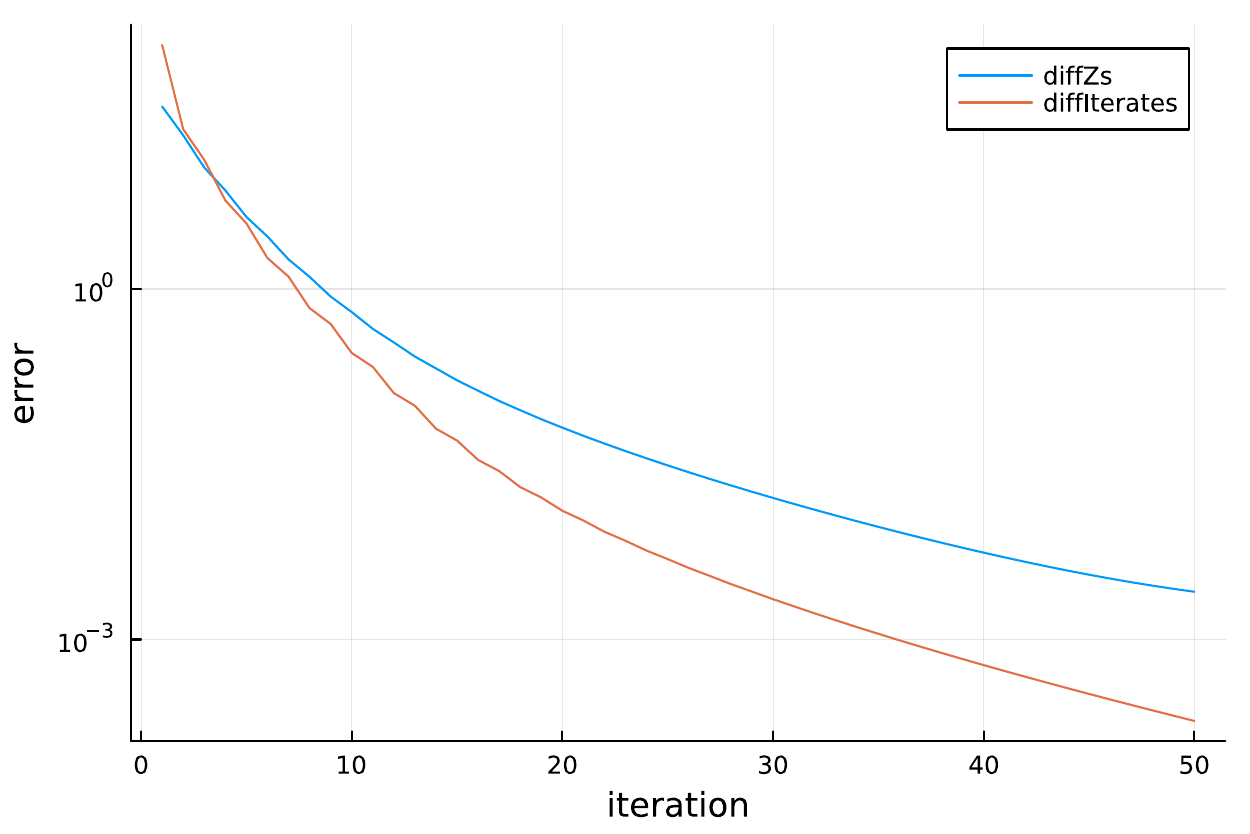}
		\caption{\YouTube}
		\label{fig:youtube-error-pr}
	\end{subfigure}
	\hfill
	\begin{subfigure}{0.24\textwidth}
		\includegraphics[width=\textwidth]{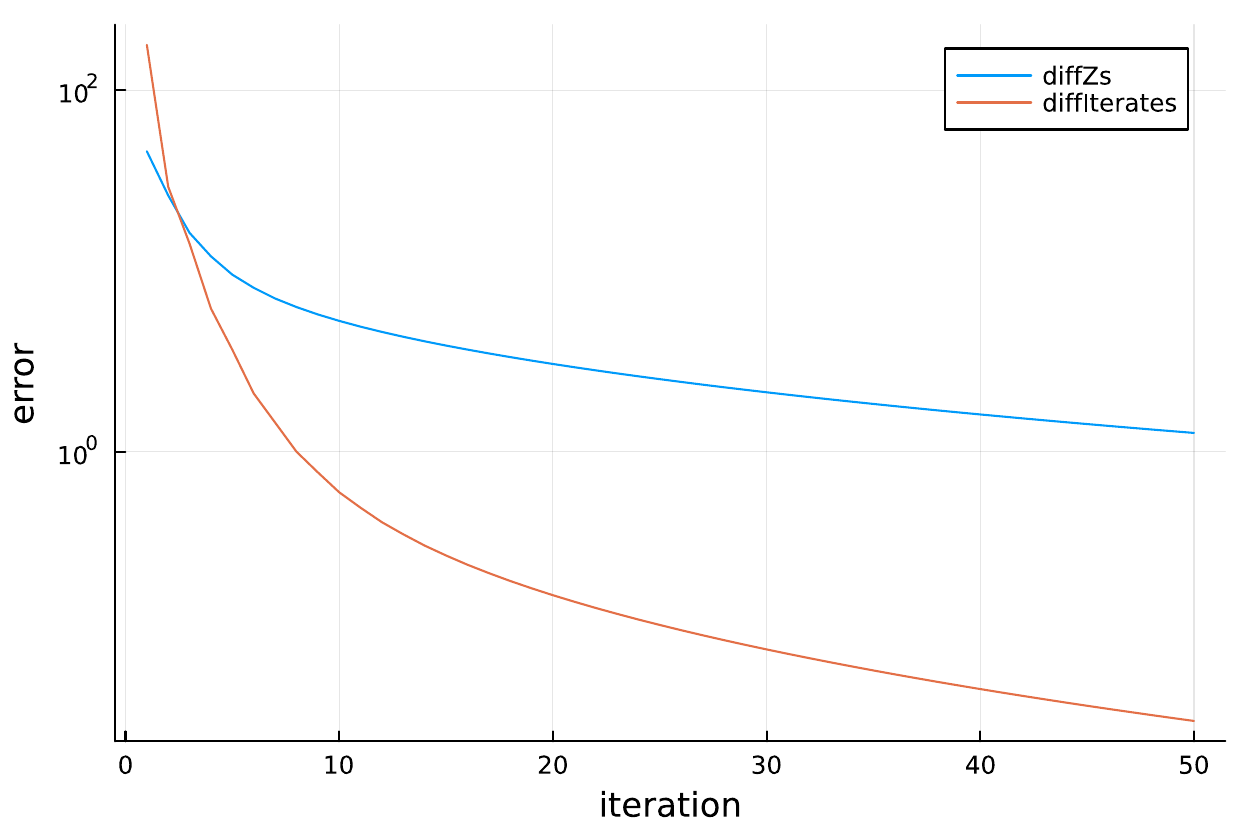}
		\caption{\LiveJournal}
		\label{fig:livejournal-error-pr}
	\end{subfigure}
	\caption{Error of applying the PageRank-style update rule from
		Proposition~\ref{prop:page-rank-equation} for multiple iterations.
		We compare the error of this iterative algorithm against the solution computed
		by the algorithm of Xu et al.~\cite{xu2021fast}. We report the $\ell_2$-norm of
		the error, as well as the differences of the iterates $\norm{z^{(t)} - z^{(t-1)}}_2$.
		Innate opinions were generated using the uniform distribution.
	}
	\label{fig:error-pr}
\end{figure*}

\subsection{Estimating the Disagreement}
Previously, we have seen that estimating the disagreement is more difficult than
estimating other measures (see Tables~\ref{tab:experiments-oracle-innate-uniform}
and~\ref{tab:experiments-oracle-expressed-uniform}). Hence, here we consider a
setting in which we can sample edges uniformly at random from unweighted graphs.
In this setting, we can again use Lemma~\ref{lem:sum} to obtain that using
$O(\varepsilon^{-2} \log \delta^{-1})$~samples, we can obtain an approximation
with error $\pm \varepsilon m$. We report our results for this algorithm on
datasets with uniform opinions. As before, we report relative errors; errors and
standard deviations are for 10~independent runs of the algorithm.

First, we present our results when we have query access to the entries $z_u^*$,
i.e., for each edge $(u,v)$ that we sample we can compute $(z_u^* - z_v^*)^2$
exactly in time $O(1)$. We present our results in
Figure~\ref{fig:error-disagreement-z}, showing that the error is much smaller
than what we reported in Table~\ref{tab:experiments-oracle-expressed-uniform},
but it does require at least $10^4$ edges to obtain good results with relative
error less than 5\%.

\begin{figure*}
	\centering
	\begin{subfigure}{0.24\textwidth}
		\includegraphics[width=\textwidth]{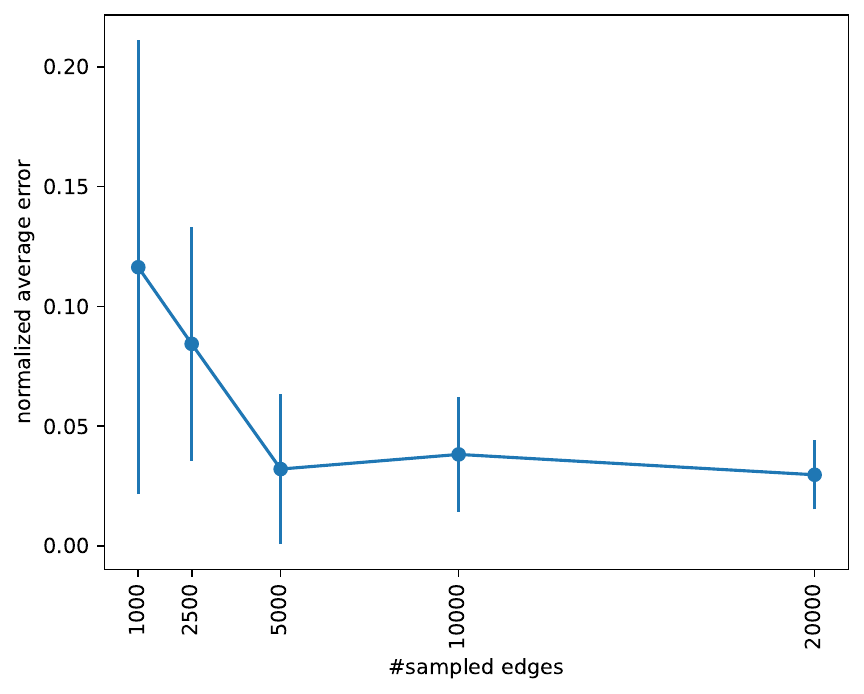}
		\caption{\Pokec}
		\label{fig:pokec-error-disagreement-z}
	\end{subfigure}
	\hfill
	\begin{subfigure}{0.24\textwidth}
		\includegraphics[width=\textwidth]{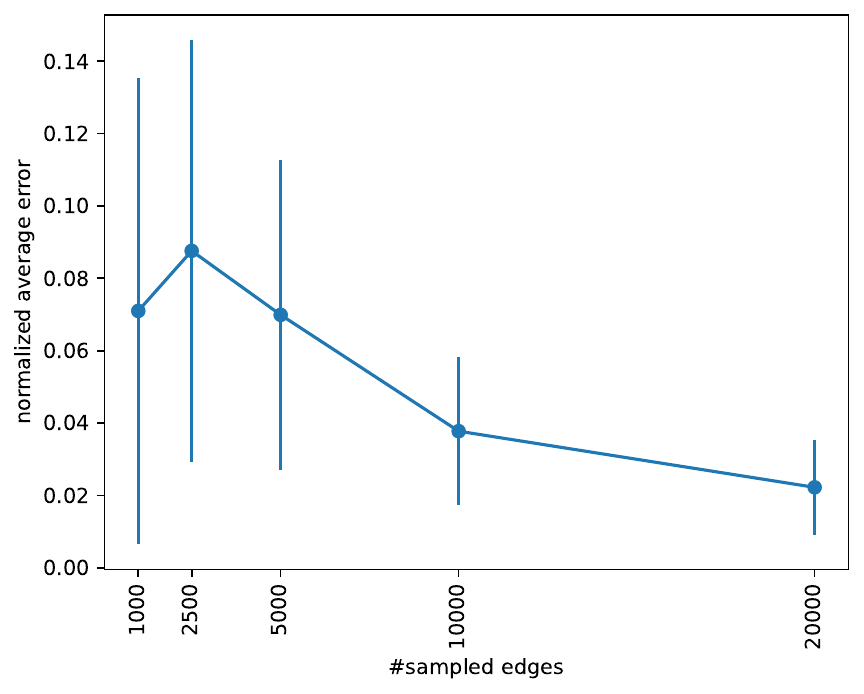}
		\caption{\Flickr}
		\label{fig:flickr-error-disagreement-z}
	\end{subfigure}
	\hfill
	\begin{subfigure}{0.24\textwidth}
		\includegraphics[width=\textwidth]{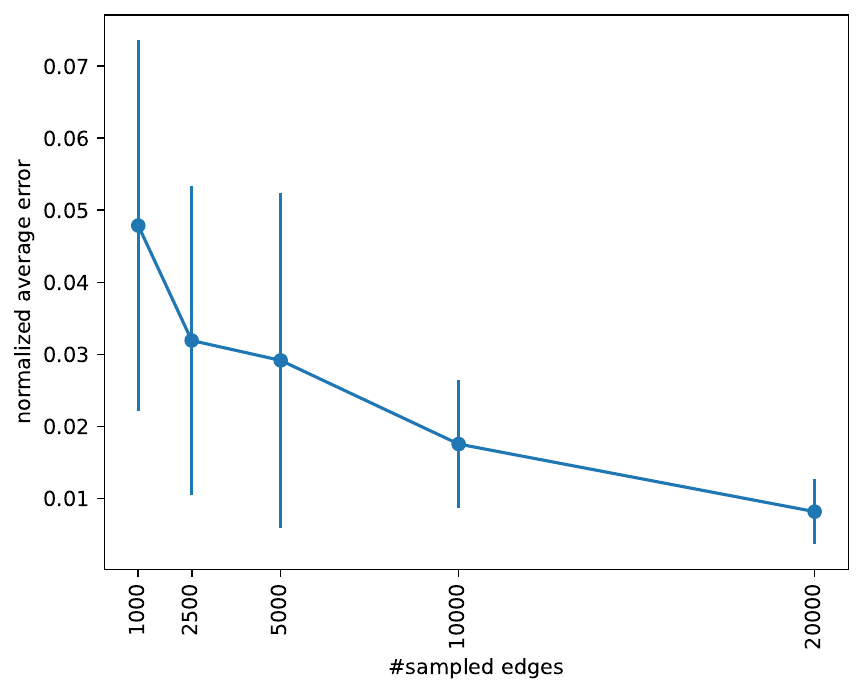}
		\caption{\YouTube}
		\label{fig:youtube-error-disagreement-z}
	\end{subfigure}
	\hfill
	\begin{subfigure}{0.24\textwidth}
		\includegraphics[width=\textwidth]{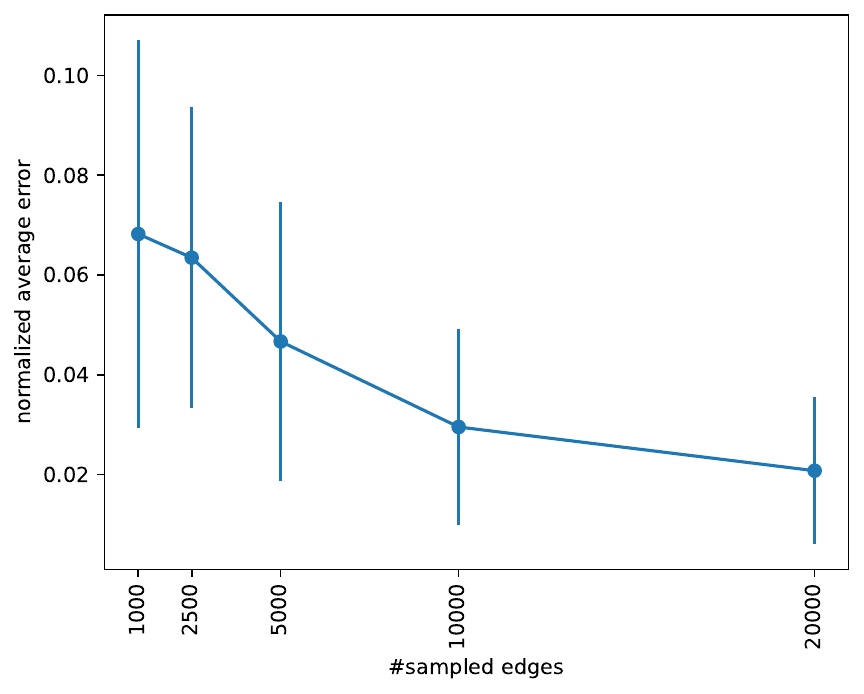}
		\caption{\LiveJournal}
		\label{fig:livejournal-error-disagreement-z}
	\end{subfigure}
	\caption{Error of the estimator from Lemma~\ref{lem:sum} for computing the
		disagreement using random sampling of edges, having query access to
		$z_u^*$. We report means and standard deviations across 10~experiments.
		Innate opinions were generated using the uniform distribution.
	}
	\label{fig:error-disagreement-z}
\end{figure*}

First, we present our results when we have query access to the entries $s$
and when we have to estimate the $z_u^*$ using Algorithm~\ref{alg:random-walks}.
As before, we use Algorithm~\ref{alg:random-walks} with \numprint{600}~steps and
\numprint{4000}~random walks.  We present our results in
Figure~\ref{fig:error-disagreement-s}, showing that the errors are generally
smaller than what we reported in
Table~\ref{tab:experiments-oracle-innate-uniform}, but
that the errors are still relatively high.

\begin{figure*}
	\centering
	\begin{subfigure}{0.24\textwidth}
		\includegraphics[width=\textwidth]{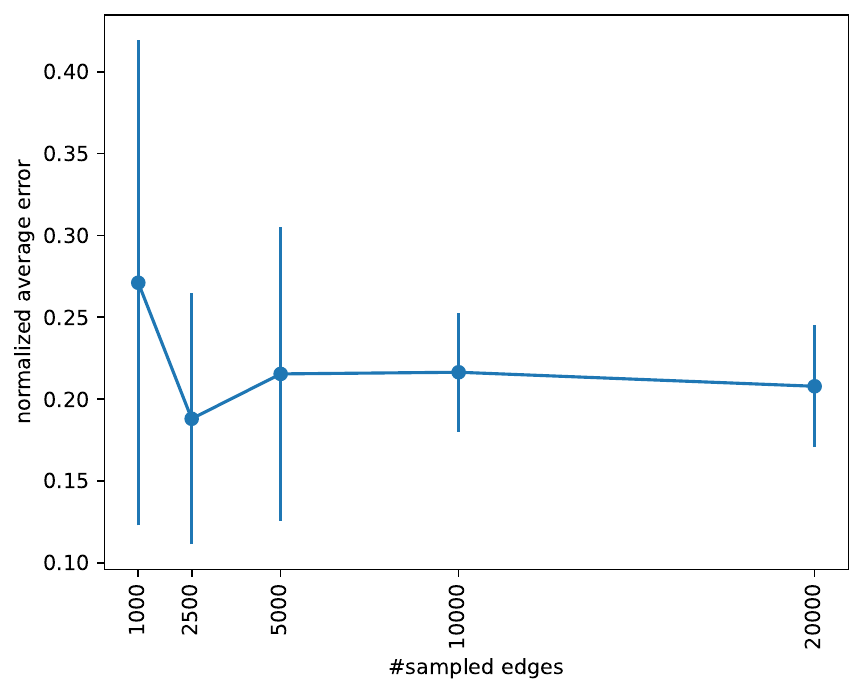}
		\caption{\Pokec}
		\label{fig:pokec-error-disagreement-s}
	\end{subfigure}
	\hfill
	\begin{subfigure}{0.24\textwidth}
		\includegraphics[width=\textwidth]{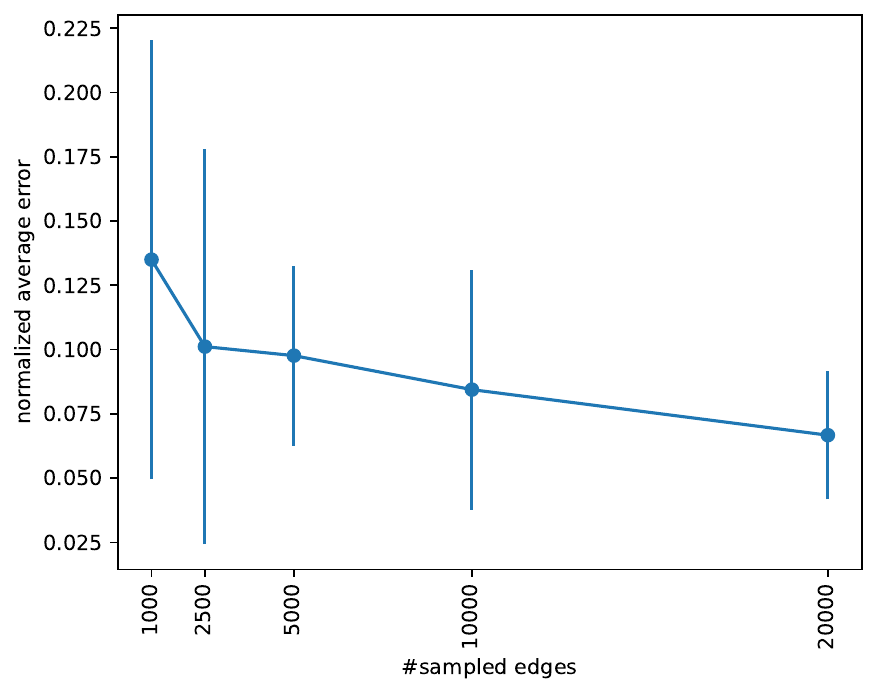}
		\caption{\Flickr}
		\label{fig:flickr-error-disagreement-s}
	\end{subfigure}
	\hfill
	\begin{subfigure}{0.24\textwidth}
		\includegraphics[width=\textwidth]{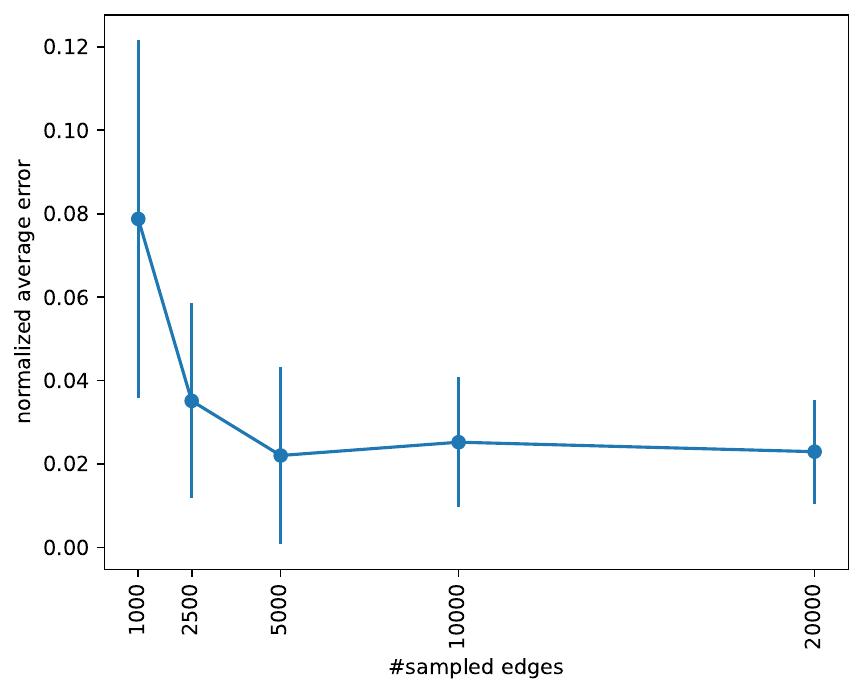}
		\caption{\YouTube}
		\label{fig:youtube-error-disagreement-s}
	\end{subfigure}
	\hfill
	\begin{subfigure}{0.24\textwidth}
		\includegraphics[width=\textwidth]{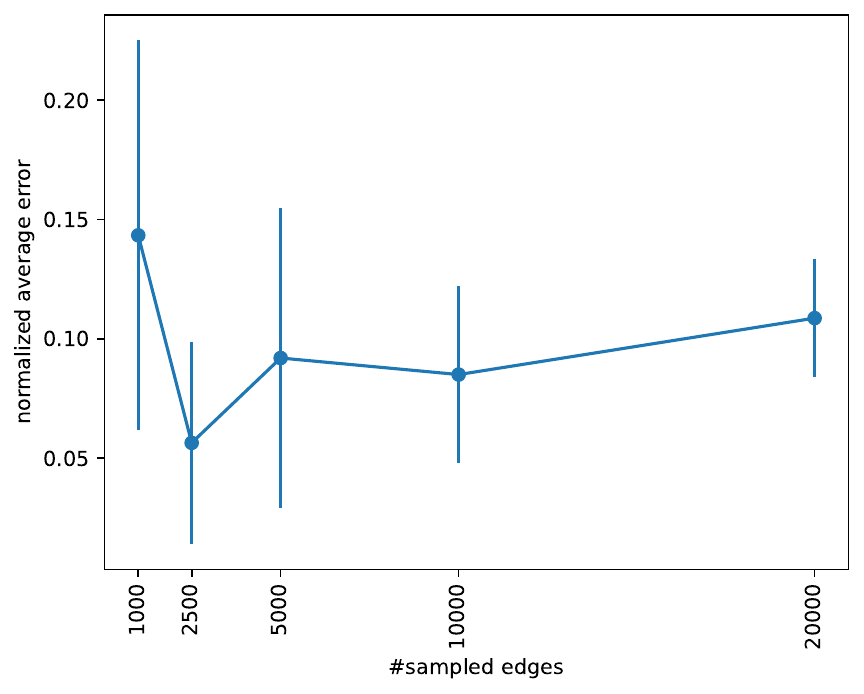}
		\caption{\LiveJournal}
		\label{fig:livejournal-error-disagreement-s}
	\end{subfigure}
	\caption{Error of the estimator from Lemma~\ref{lem:sum} for computing the
		disagreement using random sampling of edges, computing $z_u^*$ using
		Algorithm~\ref{alg:random-walks} with \numprint{600}~steps and
		\numprint{4000}~random walks. We report means and standard deviations
		across 10~experiments.
		Innate opinions were generated using the uniform distribution.
	}
	\label{fig:error-disagreement-s}
\end{figure*}

%% file: tables/measures_Eigenvalue_givenS_error_numStepsSamples600_numWalksRepetitions4000_numSampledVertices10000.tex
\GooglePlus & 0.010 ($\pm$0.007) & 0.2 ($\pm$0.2) & 2.9 ($\pm$1.7) & 156.5 ($\pm$8.1) & 8.8 ($\pm$0.7) & 3.3 ($\pm$0.2) & 1.6 ($\pm$0.3) & 0.3 ($\pm$0.3) \\ 
\TwitterFollows & 0.011 ($\pm$0.006) & 0.3 ($\pm$0.2) & 6.8 ($\pm$1.0) & 197.0 ($\pm$17.1) & 6.8 ($\pm$1.4) & 4.2 ($\pm$0.3) & 2.1 ($\pm$0.5) & 0.5 ($\pm$0.4) \\ 
\Flixster & 0.013 ($\pm$0.007) & 0.3 ($\pm$0.2) & 4.2 ($\pm$2.0) & 126.9 ($\pm$8.5) & 2.5 ($\pm$1.4) & 4.4 ($\pm$0.3) & 2.0 ($\pm$0.5) & 0.6 ($\pm$0.4) \\ 
\Pokec & 0.010 ($\pm$0.007) & 0.4 ($\pm$0.2) & 6.8 ($\pm$2.3) & 83.3 ($\pm$13.4) & 1.6 ($\pm$1.0) & 3.4 ($\pm$0.3) & 1.7 ($\pm$0.6) & 0.7 ($\pm$0.3) \\ 
\Flickr & 0.012 ($\pm$0.007) & 0.4 ($\pm$0.3) & 1.7 ($\pm$1.2) & 37.9 ($\pm$4.8) & 1.0 ($\pm$0.7) & 4.7 ($\pm$0.2) & 2.6 ($\pm$0.4) & 0.7 ($\pm$0.6) \\ 
\YouTube & 0.012 ($\pm$0.007) & 0.6 ($\pm$0.5) & 1.6 ($\pm$0.9) & 31.9 ($\pm$8.7) & 0.8 ($\pm$0.7) & 4.2 ($\pm$0.5) & 2.3 ($\pm$0.9) & 1.0 ($\pm$0.9) \\ 
\LiveJournal & 0.011 ($\pm$0.008) & 0.5 ($\pm$0.3) & 4.9 ($\pm$3.1) & 63.1 ($\pm$11.1) & 1.8 ($\pm$0.3) & 3.5 ($\pm$0.4) & 1.4 ($\pm$0.7) & 1.1 ($\pm$0.5) \\ 

%% file: tables/measures_Exponential_givenS_error_numStepsSamples600_numWalksRepetitions4000_numSampledVertices10000.tex
\GooglePlus & 0.003 ($\pm$0.002) & 0.4 ($\pm$0.3) & 3.2 ($\pm$0.9) & 47.5 ($\pm$10.1) & 2.1 ($\pm$2.2) & 3.7 ($\pm$0.3) & 1.8 ($\pm$0.5) & 0.7 ($\pm$0.6) \\ 
\TwitterFollows & 0.003 ($\pm$0.002) & 0.7 ($\pm$0.5) & 6.6 ($\pm$4.0) & 43.5 ($\pm$12.2) & 5.3 ($\pm$1.9) & 3.4 ($\pm$0.7) & 1.3 ($\pm$1.0) & 2.2 ($\pm$1.1) \\ 
\Flixster & 0.003 ($\pm$0.002) & 0.4 ($\pm$0.2) & 3.8 ($\pm$3.0) & 44.4 ($\pm$7.8) & 2.3 ($\pm$1.4) & 4.3 ($\pm$0.6) & 1.9 ($\pm$0.9) & 1.0 ($\pm$0.5) \\ 
\Pokec & 0.003 ($\pm$0.002) & 0.4 ($\pm$0.1) & 5.0 ($\pm$2.3) & 69.3 ($\pm$17.4) & 2.2 ($\pm$1.2) & 3.5 ($\pm$0.1) & 1.7 ($\pm$0.4) & 0.6 ($\pm$0.3) \\ 
\Flickr & 0.003 ($\pm$0.002) & 0.5 ($\pm$0.2) & 4.6 ($\pm$3.2) & 50.8 ($\pm$9.5) & 1.8 ($\pm$1.2) & 4.0 ($\pm$0.6) & 1.7 ($\pm$0.9) & 1.2 ($\pm$0.4) \\ 
\YouTube & 0.003 ($\pm$0.002) & 0.7 ($\pm$0.3) & 1.5 ($\pm$0.7) & 46.2 ($\pm$10.0) & 3.7 ($\pm$3.3) & 3.9 ($\pm$0.6) & 1.6 ($\pm$1.0) & 1.5 ($\pm$0.6) \\ 
\LiveJournal & 0.003 ($\pm$0.002) & 0.3 ($\pm$0.2) & 1.7 ($\pm$1.2) & 52.4 ($\pm$5.4) & 1.9 ($\pm$1.3) & 3.7 ($\pm$0.3) & 1.9 ($\pm$0.4) & 0.6 ($\pm$0.4) \\ 

%% file: tables/measures_Eigenvalue_givenZ_error_numStepsSamples400_numWalksRepetitions5_numSampledVertices10000.tex
\GooglePlus & 0.000 ($\pm$0.004) & 0.2 ($\pm$0.1) & 1.8 ($\pm$0.9) & 10.4 ($\pm$6.0) & 1.8 ($\pm$1.1) & 0.3 ($\pm$0.1) & 0.3 ($\pm$0.2) & 0.4 ($\pm$0.1) \\ 
\TwitterFollows & 0.001 ($\pm$0.017) & 0.1 ($\pm$0.1) & 0.9 ($\pm$0.7) & 10.5 ($\pm$4.6) & 13.4 ($\pm$4.6) & 0.2 ($\pm$0.1) & 0.3 ($\pm$0.2) & 0.4 ($\pm$0.2) \\ 
\Flixster & 0.001 ($\pm$0.016) & 0.2 ($\pm$0.1) & 1.7 ($\pm$1.1) & 9.2 ($\pm$7.0) & 4.1 ($\pm$1.9) & 0.3 ($\pm$0.2) & 0.5 ($\pm$0.3) & 0.6 ($\pm$0.4) \\ 
\Pokec & 0.000 ($\pm$0.004) & 0.1 ($\pm$0.0) & 2.3 ($\pm$2.0) & 8.5 ($\pm$3.9) & 1.4 ($\pm$1.1) & 0.1 ($\pm$0.1) & 0.2 ($\pm$0.1) & 0.3 ($\pm$0.2) \\ 
\Flickr & 0.001 ($\pm$0.020) & 0.1 ($\pm$0.1) & 1.7 ($\pm$0.6) & 4.9 ($\pm$1.3) & 1.9 ($\pm$1.0) & 0.4 ($\pm$0.2) & 0.4 ($\pm$0.3) & 0.5 ($\pm$0.4) \\ 
\YouTube & 0.000 ($\pm$0.011) & 0.2 ($\pm$0.1) & 1.6 ($\pm$1.2) & 6.2 ($\pm$3.0) & 0.6 ($\pm$0.3) & 0.4 ($\pm$0.2) & 0.7 ($\pm$0.3) & 0.9 ($\pm$0.4) \\ 
\LiveJournal & 0.000 ($\pm$0.008) & 0.2 ($\pm$0.1) & 1.8 ($\pm$0.8) & 4.4 ($\pm$2.5) & 0.8 ($\pm$0.3) & 0.4 ($\pm$0.1) & 0.5 ($\pm$0.2) & 0.6 ($\pm$0.2) \\ 

%% file: tables/measures_Exponential_givenZ_error_numStepsSamples400_numWalksRepetitions5_numSampledVertices10000.tex
\GooglePlus & 0.000 ($\pm$0.002) & 0.1 ($\pm$0.1) & 3.1 ($\pm$2.6) & 8.9 ($\pm$3.0) & 3.6 ($\pm$1.8) & 0.3 ($\pm$0.3) & 0.6 ($\pm$0.3) & 1.2 ($\pm$0.5) \\ 
\TwitterFollows & 0.001 ($\pm$0.014) & 0.2 ($\pm$0.1) & 2.5 ($\pm$2.5) & 6.3 ($\pm$4.3) & 11.6 ($\pm$5.6) & 0.6 ($\pm$0.4) & 0.9 ($\pm$0.6) & 1.5 ($\pm$0.6) \\ 
\Flixster & 0.000 ($\pm$0.009) & 0.2 ($\pm$0.1) & 3.0 ($\pm$1.3) & 7.4 ($\pm$5.1) & 8.0 ($\pm$5.6) & 0.4 ($\pm$0.2) & 0.6 ($\pm$0.5) & 1.0 ($\pm$1.1) \\ 
\Pokec & 0.000 ($\pm$0.002) & 0.1 ($\pm$0.1) & 8.2 ($\pm$6.0) & 16.7 ($\pm$6.9) & 3.6 ($\pm$2.2) & 0.3 ($\pm$0.2) & 0.6 ($\pm$0.4) & 1.2 ($\pm$0.4) \\ 
\Flickr & 0.000 ($\pm$0.008) & 0.3 ($\pm$0.1) & 3.1 ($\pm$2.5) & 7.2 ($\pm$5.7) & 5.1 ($\pm$3.1) & 0.6 ($\pm$0.4) & 0.9 ($\pm$0.6) & 1.0 ($\pm$0.2) \\ 
\YouTube & 0.000 ($\pm$0.005) & 0.2 ($\pm$0.1) & 2.9 ($\pm$1.5) & 8.0 ($\pm$6.9) & 4.0 ($\pm$3.2) & 0.4 ($\pm$0.1) & 0.7 ($\pm$0.5) & 1.1 ($\pm$1.0) \\ 
\LiveJournal & 0.000 ($\pm$0.006) & 0.2 ($\pm$0.1) & 3.5 ($\pm$2.4) & 6.4 ($\pm$5.3) & 3.2 ($\pm$4.3) & 0.5 ($\pm$0.3) & 0.5 ($\pm$0.4) & 0.9 ($\pm$0.7) \\ 